\documentclass[11pt]{article}
\usepackage{tikz}
\usepackage{wrapfig,lipsum}   

\usepackage{xspace}

\usepackage{amsmath,amssymb}
\usepackage{amsthm}
\usepackage{mathtools}

\usepackage{thmtools}
\usepackage{thm-restate}

\usepackage{hyperref}

\usepackage{url}

\hypersetup{
  colorlinks   = true, %Colours links instead of ugly boxes
  urlcolor     = blue, %Colour for external hyperlinks
  linkcolor    = blue, %Colour of internal links
  citecolor   = red %Colour of citations
}

\usepackage[noend]{algorithmic}
\usepackage[ruled,vlined]{algorithm2e}
\usepackage{url}
\usepackage{makeidx}
\usepackage{enumerate}
\usepackage[top=1in, bottom=1in, left=1in, right=1.25in]{geometry}
\usepackage{float,psfrag,epsfig,caption}

\usepackage{epstopdf}

\usepackage{color}

\usepackage{subfig}
\usepackage{caption}
\usepackage[utf8]{inputenc}

\usepackage{scalerel,stackengine}
\usepackage{enumitem}

\title{Optimizing Mean Field Spin Glasses with External Field}
\author{Mark Sellke}

\date{}

\DeclareSymbolFont{rsfs}{U}{rsfs}{m}{n}
\DeclareSymbolFontAlphabet{\mathscrsfs}{rsfs}

\numberwithin{equation}{section}

\def\wt{\widetilde}
\def\wh{\widehat}

\def\cuU{\mathscrsfs{U}}
\def\cuL{\mathscrsfs{L}}
\newcommand{\R}{\mathbb{R}}
\newcommand{\ul}{\underline{\ell}}
\newcommand{\ubl}{\overline{\ell}}
\def\Par{{\sf P}}

\def\supp{{\overline S}}

\def\ss{\mathsf{sph}}

\def\bf{{\boldsymbol{f}}}
\def\bg{{\boldsymbol{g}}}
\def\bm{{\boldsymbol{m}}}
\def\bn{{\boldsymbol{n}}}

\def\bv{{\boldsymbol{v}}}
\def\bw{{\boldsymbol{w}}}
\def\bx{{\boldsymbol{x}}}
\def\by{{\boldsymbol{y}}}
\def\bz{{\boldsymbol{z}}}
\def\bW{{\boldsymbol{w}}}

\def\bsig{{\boldsymbol{\sigma}}}
\def\bbE{{\mathbb E}}

\newcommand{\ubq}{\bar{q}}
\newcommand{\lbq}{\underline{q}}

\newtheorem{claim}{Claim}[section]
\newtheorem{lemma}[claim]{Lemma}

\newtheorem{theorem}{Theorem}
\newtheorem{proposition}[claim]{Proposition}
\newtheorem{corollary}[claim]{Corollary}
\newtheorem{definition}[claim]{Definition}
\newtheorem{remark}[claim]{Remark}

\newcommand{\bea}{\begin{eqnarray}}
\newcommand{\eea}{\end{eqnarray}}
\newcommand{\<}{\langle}
\renewcommand{\>}{\rangle}
\newcommand{\E}{{\mathbb E}}

\def\la{\langle}
\def\ra{\rangle}
\def\lt{\left}
\def\rt{\right}

\def\eps{{\varepsilon}}

\def\bh{\boldsymbol{h}}

\def\supp{{\rm supp}}

\def\cuU{\mathscrsfs{U}}
\def\cuL{\mathscrsfs{L}}

\def\ind{{\mathbb I}}

\def\bsigma{{\boldsymbol{\sigma}}}

\def\bW{{\boldsymbol{W}}}
\def\bSigma{{\boldsymbol{\Sigma}}}

\def\bg{{\boldsymbol{g}}}

\def\op{\mbox{\footnotesize\rm op}}

\def\reals{{\mathbb R}}

\def\normal{{\sf N}}

\def\bv{{\boldsymbol{v}}}
\def\bz{{\boldsymbol{z}}}
\def\bx{{\boldsymbol{x}}}

\def\bT{\boldsymbol{T}}
\def\bm{\boldsymbol{m}}

\def\Par{{\sf P}}
\def\de{{\rm d}}

\def\bW{\boldsymbol{W}}
\def\prob{{\mathbb P}}
\def\E{{\mathbb E}}

\def\<{\langle}
\def\>{\rangle}

\def\sign{{\rm sign}}

\def\by{{\boldsymbol{y}}}
\def\bw{{\boldsymbol{w}}}

\def\b0{{\boldsymbol{0}}}

\DeclareMathOperator*{\plim}{p-lim}

\def\bn{{\boldsymbol n}}

\renewcommand{\b}{\mathbf{b}}

\begin{document}

\maketitle
\abstract{We consider the Hamiltonians of mean-field spin glasses, which are certain random functions $H_N$ defined on high-dimensional cubes or spheres in $\mathbb R^N$. The asymptotic maximum values of these functions were famously obtained by Talagrand and later by Panchenko and by Chen. The landscape of approximate maxima of $H_N$ is described by various forms of replica symmetry breaking exhibiting a broad range of possible behaviors. We study the problem of efficiently computing an approximate maximizer of $H_N$. 

We give a two-phase message pasing algorithm to approximately maximize $H_N$ when a no overlap-gap condition holds. This generalizes the recent works \cite{subag2018following,mon18,ams20} by allowing a non-trivial external field. For even Ising spin glasses with constant external field, our algorithm succeeds exactly when existing methods fail to rule out approximate maximization for a wide class of algorithms. Moreover we give a branching variant of our algorithm which constructs a full ultrametric tree of approximate maxima. }

{
\small
\tableofcontents}

\newpage

\section{Introduction}

Optimizing non-convex functions in high dimensions is well-known to be computationally intractible in general. In this paper we study the optimization of a natural class of \emph{random} non-convex functions, namely the Hamiltonians of mean-field spin glasses. These functions $H_N$ are defined on either the cube $\Sigma_N=\{-1,1\}^N$ or the sphere $\mathbb S^{N-1}(\sqrt{N})$ of radius $\sqrt{N}$ and have been studied since \cite{sherrington1975solvable} as models for the behavior of disordered magnetic systems. 

The distribution of an $N$-dimensional mean-field spin glass Hamiltonian $H_N$ is described by an exponentially decaying sequence $(c_p)_{p\geq 2}$ of non-negative real numbers as well as an external field probability distribution $\mathcal L_h$ on $\mathbb R$ with finite second moment. Given these data, one samples $h_1,\dots,h_N\sim \mathcal L_h$ and standard Gaussians $g_{i_1,\dots,i_p}\sim\normal(0,1)$ and then defines $H_N:\mathbb R^N\to\mathbb R$ by
\begin{align*}
    H_N(\bx)&=\sum_i h_ix_i +\wt H_N(\bx),\\
    \wt H_N(\bx)&=\sum_{p=2}^{\infty}
    \frac{c_p}{N^{(p-1)/2}}\sum_{i_1,\dots,i_p=1}^N g_{i_1,\dots,i_p}x_{i_1}\dots x_{i_p}.
\end{align*}

The distribution of the non-linear part $\wt H_N$ is characterized by the mixture function $\xi(z)=\sum_{p\geq 2} c_p^2 z^p$ - there are no issues of convergence for $|z|\leq 1+\eta$ thanks to the exponential decay assumption. We assume throughout that $\xi$ is not the zero function so that we study a genuine spin glass. $\wt H_N$ is then a centered Gaussian process with covariance 
\[
    \mathbb E\big[\wt H_N(\bx_1)\wt H_N(\bx_2)\big]
    =
    N\xi\lt(\frac{\langle \bx_1,\bx_2\rangle}{N}\rt).
\]

Spin glasses were introduced to model the magnetic properties of diluted materials and have been studied in statistical physics and probability since the seminal work \cite{sherrington1975solvable}. In this context, the object of study is the Gibbs measure $\frac{e^{\beta H_N(\bx)}d\mu(\bx)}{Z_{N,\beta}}$ where $\beta>0$ is the inverse-temperature, $\mu(\bx)$ is a fixed reference measure and $Z_{N,\beta}$ is a random normalizing constant known as the partition function. The most common choice is to take $\mu(\cdot)$ the uniform measure on $\bSigma_N=\{-1,1\}^N$, and another canonical choice is the uniform measure on $\mathbb S^{N-1}(\sqrt{N})$. These two choices define \emph{Ising} and \emph{spherical} spin glasses. The quantity of primary interest is the free energy
\[
F_N(\beta)=\log\mathbb E^{\bx\sim\mu}[e^{\beta H_N(\bx)}].
\]
The in-probability normalized limit $F(\beta)=\plim_{N\to\infty}\frac{ F_N(\beta)}{N}$ of the free energy at temperature $\beta$ is famously given by an infinite-dimensional variational problem known as the Parisi formula (or the Cristanti-Sommers formula in the spherical case) as we review in the next section. These free energies are well-concentrated and taking a second limit $\lim_{\beta\to \infty}\frac{F(\beta)}{\beta}$ yields the asymptotic ground state energies 

\begin{align*}GS(\xi,\mathcal L_h)&=\plim_{N\to\infty}\max_{\bx\in\bSigma_N}\frac{H_N(\bx)}{N},\\
GS_{\ss}(\xi,\mathcal L_h)&=\plim_{N\to\infty}\max_{\bx\in\mathbb S^{N-1}(\sqrt{N})}\frac{H_N(\bx)}{N}.\end{align*}

From the point of view of optimization, spin glass Hamiltonians serve as natural examples of highly non-convex functions. Indeed, the landscape of $H_N$ can exhibit quite complicated behavior. For instance $H_N$ may have exponentially many near-maxima on $\bSigma_N$ \cite{chatterjee2009disorder,ding2015multiple,chen2018energy}. The structure of these near-maxima is highly nontrivial; the Gibbs measures on $\bSigma_N$ are approximate ultrametrics in a certain sense, at least in the so-called generic models \cite{jagannath2017approximate,chatterjee2019average}. Moreover spherical spin glasses typically have exponentially many \emph{local} maxima and saddle points, which are natural barriers to gradient descent and similar optimization algorithms \cite{auffinger2013complexity,auffinger2013random,subag2017complexity,arous2019landscape}. The utility of a rich model of random functions is made clear by a comparison to the theory of high-dimensional non-convex optimization in the worst-case setting. In the black-box model of optimization based on querying function values, gradients, and Hessians, approximately optimizing an unknown non-convex function in high-dimension efficiently is trivially impossible and substantial effort has gone towards the more modest task of finding a local optimum or stationary point \cite{carmon2017convex,jin2017escape,agarwal2017finding,carmon2018accelerated,carmon2019lower}. Even for quadratic polynomials in $N$ variables, it is quasi-NP hard to reach within a factor $\log(N)^{\eps}$ of the optimum \cite{arora2005non}.  For polynomials of degree $p\geq 3$ on the sphere, \cite{barak2012hypercontractivity} proves that even an approximation ratio $e^{(\log N)^{\eps}}$ is computationally infeasible to obtain.

Despite the worst-case obstructions just outlined, a series of recent works have found great success in approximately maximizing certain spin glass Hamiltonians. By \emph{approximate maximization} we always mean maximization up to a factor $(1+\eps)$, where $\eps>0$ is an arbitrarily small positive constant; we similarly refer to a point $\bx\in\bSigma_N$ or $\bx\in \mathbb S^{N-1}(\sqrt{N})$ achieving such a nearly optimal value as an \emph{approximate maximizer} (where the small constant $\eps$ is implicit). Subag showed in \cite{subag2018following} how to approximately maximize spherical spin glasses by using top eigenvectors of the Hessian $\nabla^2 H_N$. Subsequently \cite{mon18,ams20} developed a message passing algorithm with similar guarantees for the Ising case. These works all operate under an assumption of no overlap gap, a condition which is expected (known in the spherical setting) to hold for some but not all models $(\xi,\mathcal L_h)$ - otherwise they achieve an explicit, sub-optimal energy value. Such a no overlap gap assumption is expected to be necessary to find approximate maxima efficiently. Indeed, the works \cite{arous2018spectral,gamarnik2019overlap,GJW20} rule out various algorithms for optimizing spin glasses when an overlap gap holds. Variants of the overlap gap property have been shown to rule out $(1+\eps)$-approximation by certain classes of algorithms for random optimization problems on sparse graphs \cite{mezard2005clustering,achlioptas2011solution,gamarnik2014limits,rahman2017local,gamarnik2017performance,chen2019suboptimality,wein2020optimal}. Overlap gaps have also been proposed as evidence of computational hardness for a range of statistical tasks including planted clique, planted dense submatrix, sparse regression, and sparse principal component analysis \cite{gamarnik2017sparse,gamarnik2018finding,gamarnik2019ogp,gamarnik2019landscape,arous2020free}. We discuss overlap gaps more in Subsection~\ref{subsec:optogp} and Section~\ref{sec:OGP}.

The aforementioned algorithms in \cite{subag2018following,mon18,ams20} are all restricted to settings with no external field, i.e. with $h_i=0$ for all $i$. Our main purpose is to generalize these results to allow for an external field. We focus primarily on the Ising case and explain in Section~\ref{sec:sphere} how to handle the easier spherical models. Our main algorithm consists of two stages of message passing. The first stage is inspired by the work \cite{bolthausen2014iterative} which constructs solutions to the TAP equations for the SK model at high temperature. We construct approximate solutions to the generalized TAP equations of \cite{subag2018free,chen2018generalized,chen2019generalized}, which heuristically amounts to locating the root of the ultrametric tree of approximate maxima. The second stage is similar to \cite{mon18,ams20} and uses incremental approximate message passing to descend the ultrametric tree by simulating the SDE corresponding to a candidate solution for the Parisi variational problem. A related two-stage message passing algorithm was recently introduced in our joint work with A.E. Alaoui on the negative spherical perceptron \cite{AS20}.

While the primary goal in this line of work is to construct a single approximate maximizer, Subag beautifully observed in \cite[Remark 6]{subag2018following} that an extension of his Hessian-based construction for spherical models produces approximate maximizers arranged into a completely arbitary ultrametric space obeying an obvious diameter upper bound. The overlap gap property essentially states that distances between approximate maximizers cannot take certain values, and so this is a sort of constructive converse result.  In Section~\ref{sec:branch} we give a branching version of our main algorithm, following a suggestion of \cite{alaoui2020algorithmic}, which constructs an arbitrary ultrametric space of approximate maximizers in the Ising case (again subject to a diameter upper bound).

\subsection{Optimizing Ising Spin Glasses}\label{subsec:result}

To state our results we require the Parisi formula for the ground state of a mean field Ising spin glass as given in \cite{auffinger2017parisi}. Let $\cuU$ be the function space
\[
\cuU=\left\{\gamma:[0,1)\to [0,\infty):\gamma\text{ is non-decreasing},\int_0^1\gamma(t)\de t<\infty \right\}.
\]

The functions $\gamma$ are meant to correspond to cumulative distribution functions - for finite $\beta$ the corresponding Parisi formula requires $\gamma(1)=1$, but this constraint disappears in renormalizing to obtain a zero-temperature limit. For $\gamma\in\cuU$ we take $\Phi_{\gamma}(t,x):[0,1]\times\R\to\R$ to be the solution of the following Parisi PDE: 
\begin{align*}
    \partial_t\Phi_{\gamma}(t,x) 
        + \frac{1}{2}\xi''(t)
        \lt(
            \partial_{xx}\Phi_{\gamma}(t,x)
            + \gamma(t)(\partial_x\Phi_{\gamma}(t,x))^2
        \rt)
    &= 0, \\
    \Phi_{\gamma}(1,x) &= |x|.
\end{align*}
{\color{black} This PDE is known to be well-posed, see Proposition~\ref{prop:phireg}. }
Intimately related to the above PDE is the stochastic differential equation 
\begin{equation}
\label{eq:parisiSDE}
    \de X_t = \xi''(t)\gamma(t)\partial_x\Phi_{\gamma}(t,X_t)\de t
        + \sqrt{\xi''(t)}\de B_t,
    \quad X_0\sim \mathcal L_h.
\end{equation}
which we call the Parisi SDE. The Parisi functional $\Par:\cuU\to\R$ with external field distribution $\mathcal L_h$ is given by:
\begin{equation}
\label{eq:Par}
    \Par_{\xi,\mathcal L_h}(\gamma)
    = \mathbb E^{h\sim \mathcal L_h}[\Phi_{\gamma}(0,h)]
    - \frac{1}{2}\int_0^1 t \xi''(t) \gamma(t) \de t.
\end{equation}
The Parisi formula for the ground state energy is as follows.

\begin{theorem}
[{\cite{talagrand2006parisi,panchenko2014parisi,auffinger2017parisi,chen2018energy}}]
\label{thm:GSE-U}
\[
  GS(\xi,\mathcal L_h)=\inf_{\gamma\in\cuU}\Par_{\xi,\mathcal L_h}(\gamma).
\]
Moreover the minimum is attained at a unique $\gamma^{\cuU}_*\in\cuU$.
\end{theorem}

Through the paper, $\gamma^{\cuU}_*$ will always refer to the minimizer of Theorem~\ref{thm:GSE-U}. We now turn to algorithms. In \cite{mon18}, Montanari introduced the class of \emph{incremental approximate message passing} (IAMP) algorithms to optimize the SK model. These are a special form of the well-studied approximate message passing (AMP) algorithms, reviewed in Subsection~\ref{subsec:amp}. The work \cite{ams20} showed that the maximum asympototic value of $H_N$ achievable by IAMP algorithms is given by the minimizer of $\Par$, assuming it exists, over a larger class of non-monotone functions, when $\mathcal L_h=\delta_0$ so there is no external field. This larger class is:
\[
    \cuL=\lt\{\gamma:[0,1)\to [0,\infty): \gamma\text{ is right-continuous}, \Vert \xi''\cdot \gamma\Vert_{TV[0,t]}<\infty\forall t\in [0,1),\int_0^1\xi''(t) \gamma(t)\de t<\infty\rt\}.
\]
Here $TV[0,t]$ denotes the total variation norm
\[
    \|f\|_{TV[0,t]} \equiv \sup_n\sup_{0 \le t_0<t_1<\dots<t_k \le t} \sum_{i=1}^k |f(t_i) - f(t_{i-1})|\, .
\]

The Parisi PDE \eqref{eq:Par} and associated SDE extend also to $\cuL$. We denote by $\gamma_*^{\cuL}\in \cuL$ the minimizer of $\Par$ over $\cuL$, assuming that it exists. Note that uniqueness always holds by Lemma~\ref{lem:identity} below. We remark that \cite{ams20} does not include the right-continuity constraint in defining $\cuL$, however this constraint only serves to eliminate ambiguities between $\gamma_1,\gamma_2$ differing on sets of measure $0$. In fact \cite{ams20} assumes $\gamma\in\cuL$ is right-continuous from Lemma 6.9 onward.

\begin{theorem}[{\cite[Theorem 3]{ams20}}]
\label{thm:amsmain}
With $\mathcal L_h=\delta_0$, suppose $\inf_{\gamma\in\cuL}\Par(\gamma)$ is achieved at $\gamma_*^{\cuL}\in\cuL$. Then for any $\eps>0$ there exists an efficient AMP algorithm which outputs
$\bsigma\in\Sigma_N$ satisfying
\begin{align}
\frac{H_N(\bsigma)}{N} {\color{black}\geq\Par(\gamma_*^{\cuL})-\eps}
\end{align}
with probability tending to $1$ as $N\to\infty$.
\end{theorem}

We clarify our use of the word ``efficient" in Subsection~\ref{subsec:amp} - in short, it means that $O_{\eps}(1)$ evaluations of $\nabla \wt{H}_N$ and first/second partial derivatives of $\Phi_{\gamma_*^{\cuL}}$ are required. In general, minimizing over the larger space $\cuL$ instead of $\cuU$ may decrease the infimum value of $\Par$, so that IAMP algorithms fail to approximately maximize $H_N$. However if $\gamma_*^{\cuU}$ is strictly increasing, then the infima are equal. 

\begin{corollary}[{\cite[Corollary 2.2]{ams20}}]
\label{cor:q=0}

With $\mathcal L_h=\delta_0$, suppose that $\gamma_*^{\cuU}$ is strictly increasing on $[0,1)$. Then $\gamma_*^{\cuU}=\gamma_*^{\cuL}$. Consequently for any $\eps>0$ there is an efficient AMP algorithm which outputs
$\bsigma\in\Sigma_N$ satisfying
\begin{align}
\frac{H_N(\bsigma)}{N} {\color{black}\geq GS(\xi,\mathcal L_h)-\eps}
\end{align}
with probability tending to $1$ as $N\to\infty$.

\end{corollary}

We define the support $\supp(\gamma)$ of $\gamma\in\cuL$ to be the closure in $[0,1)$ of $S(\gamma)\equiv \{x\in [0,1):\gamma(x)>0\}$. Note that this is not the same as the support of the signed measure with CDF $\gamma$.

We now present our new results when there is a non-trivial external field distribution $\mathcal L_h\neq \delta_0$. The following proposition shows that this forces $\gamma_*^{\cuU}(t)=0$ in a neighborhood of $t=0$, hence Corollary~\ref{cor:q=0} cannot apply. The proof is exactly the same as \cite[Lemma A.19]{panchenkoextra} (which is the same result for positive temperature). 

\begin{proposition}
We have $0\in\supp(\gamma_*^{\cuU})$ if and only if $\mathcal L_h=\delta_0$.
\end{proposition}

Despite this, we will show that approximate maximization is still possible with an external field if $\gamma_*^{\cuU}$ is strictly increasing on $[\lbq,1)$ for $\lbq=\inf(\supp(\gamma_*^{\cuU}))$. If this condition holds, we give a two-phase approximate message passing algorithm which first locates a suitable point $\bm_{\ul}$ with $L^2$ norm $\Vert \bm_{\ul}\Vert\approx \sqrt{\lbq N}$, and then proceeds as in the no-external-field case. The relevant condition is precisely defined as follows.

\begin{definition}
For $\gamma_*\in\cuL$, let $\lbq=\inf(\supp(\gamma_*))$. We say $\gamma_*$ is $\lbq$-optimizable if, with $X_t$ given by \eqref{eq:parisiSDE}:
\begin{equation}\label{eq:opt}\mathbb E[\partial_{x}\Phi_{\gamma_*}(t,X_t)^2]=t, \quad t\in [\lbq,1).\end{equation}
We say $\gamma_*\in\cuL$ is optimizable if it is $\lbq$-optimizable for $\lbq=\inf(\supp(\gamma_*))$. We say that $(\xi,\mathcal L_h)$ is optimizable, or equivalently that the \textbf{no overlap gap} property holds for $(\xi,\mathcal L_h)$, if the function $\gamma_*^{\cuU}$ is optimizable.
\end{definition}

In \cite{mon18}, the widely believed conjecture that (in our language) the Sherrington-Kirkpatrick model with $\xi(x)=x^2/2$ is optimizable was assumed. Our preliminary numerical simulations suggest that the SK model remains optimizable with any constant external field $\mathcal L_h=\delta_h$. However even without external field, proving this conjecture rigorously seems to be very difficult.

For $\lbq\in [0,1)$, let $\cuL_{\lbq}=\{\gamma\in\cuL:\inf(\supp(\gamma))\geq\lbq\}$ consist of functions in $\cuL$ vanishing on $[0,\lbq)$. The next lemma shows optimizability is equivalent to minimizing $\Par$ over either $\cuL$ or $\cuL_{\lbq}$. It is related to results in \cite{auffinger2015parisi,jagannath2016dynamic,ams20} which show that $\gamma_*^{\cuU}$ and $\gamma_*^{\cuL}$ satisfy \eqref{eq:opt}, in the former case when $t$ is a point of increase for $\gamma_*^{\cuU}$. The proof is given in Section~\ref{sec:identity}.

\begin{restatable}{lemma}{lemidentity}
\label{lem:identity}

For $\gamma_*\in\cuL$ and $\lbq=\inf(\supp(\gamma_*))$, the following are equivalent:

\begin{enumerate}
\item \label{it:opt1} $\gamma_*$ is optimizable.
\item \label{it:opt2} $\Par(\gamma_*)=\inf_{\gamma\in\cuL}\Par(\gamma).$
\item \label{it:opt3} $\Par(\gamma_*)=\inf_{\gamma\in\cuL_{\lbq}}\Par(\gamma).$
\end{enumerate}
Moreover if a minimizer exists in either variational problem just above, then it is unique.
\end{restatable}

Lemma~\ref{lem:identity} implies that any optimizable $\gamma_*$ is in fact the unique minimizer $\gamma^{\cuL}_*\in\cuL$ of the Parisi functional. However throughout much of the paper we will use $\gamma_*$ to denote general optimizable function without making use of this result. We made this choice because while Lemma~\ref{lem:identity} is important to make sense of our results, it is not necessary for proving e.g. Theorem~\ref{thm:main} below.
We now state our main results.

\begin{theorem}
\label{thm:main}

Suppose $\gamma_*\in\cuL$ is optimizable. Then for any $\eps>0$ there exists an efficient AMP algorithm which outputs $\bsigma\in\Sigma_N$ such that
\begin{align*}
\frac{H_N(\bsigma)}{N} \geq {\color{black}\Par(\gamma_*)-\eps}
\end{align*}
with probability tending to $1$ as $N\to\infty$.

\end{theorem}

% \begin{corollary}
% \label{cor:main}

% Suppose $\gamma_*^{\cuL}$ exists. Then for any $\eps$ there exists an efficient AMP algorithm which outputs $\bsigma\in\Sigma_N$ such that

% \begin{align}
% \frac{H_N(\bsigma)}{N} \in \lt[\Par(\gamma_*^{\cuL})-\eps,\Par(\gamma_*^{\cuL})+\eps\rt]
% \end{align}
% with probability tending to $1$ as $N\to\infty$.

% \end{corollary}

% It follows from the above that approximate maximization of $H_N$ is possible for all optimizable models. In particular this always holds when $\gamma_*\in\cuU$ is strictly increasing on $[\lbq,1)$.

\begin{restatable}{lemma}{optimizable}
\label{lem:optimizable}
If $\gamma_*^{\cuU}$ strictly increases on $[\lbq,1)$ for $\lbq=\inf(\supp(\gamma_*^{\cuU}))$, then no overlap gap holds, i.e. $\gamma_*^{\cuU}$ is optimizable.
\end{restatable}

\begin{corollary}
\label{cor:generalq}
Suppose no overlap gap holds. Then for any $\eps>0$ an efficient AMP algorithm outputs $\bsigma\in\Sigma_N$ satisfying
\begin{align*}
    \frac{H_N(\bsigma)}{N} {\color{black}\geq GS(\xi,\mathcal L_h)-\eps}
\end{align*}
with probability tending to $1$ as $N\to\infty$.
\end{corollary}

\begin{remark}
\label{rem:variational}
Unlike for $\cuU$ the infimum $\inf_{\gamma\in\cuL}\Par(\gamma)$ need not be achieved, i.e. an optimizable $\gamma_*$ need not exist. For instance, one has $\xi''(0)=0$ whenever $c_2=0$. On the other hand if $\gamma$ is optimizable, Corollary~\ref{cor:ED2} and Lemma~\ref{lem:stationarity} (with $\lbq=0$) yield
\[
    \int_0^t \xi''(s)\E[\partial_{xx}\Phi_{\gamma_*}(s,X_s)^2] \de s=\E[\partial_x\Phi_{\gamma_*}(t,X_t)^2]\geq t,\quad t\geq 0.
\]
In light of Lemma~\ref{lem:continuous} the integrand on the left-hand side is $O(\xi''(s))=o(1)$ so the above cannot hold for small $t$. Hence if $c_2=0$ there exists no optimizable $\gamma_*$. We conjecture that conversely a minimizing $\gamma_*^{\cuL}\in\cuL$ exists whenever $c_2>0$, but we do not have a proof. 
\end{remark}

\begin{remark}
By the symmetry of $\wt{H}_N$, the external field can also be a deterministic vector $\bh=(h_1,\dots,h_N)$. As long as the empirical distribution of the values $(h_i)_{i\in [N]}$ is close to $\mathcal L_h$ in $W_2$ distance and the external field is independent of $\wt{H}_N$, exactly the same results hold - see \cite[Appendix A, Theorem 6]{ams20} for the corresponding AMP statement in this slightly more general setting. 
\end{remark}

\subsection{Optimizability and No Overlap Gap}\label{subsec:optogp}

In contrast to Corollary~\ref{cor:generalq}, the paper \cite{gamarnik2019overlap} rules out approximate maximization using AMP for pure $p$-spin models without external field based on an overlap gap property whenever $p\geq 4$ is even. In formulating this result, \cite{gamarnik2019overlap} defines an AMP algorithm to be any iteration of the form given in Subsection~\ref{subsec:amp} with Lipschitz non-linearities $f_0,f_1,\dots,f_{\ell}$ which outputs $\bsigma=\max(-1,\min(1,f_{\ell}(\bz^0,\dots,\bz^{\ell})))$, the final iterate $f_{\ell}$ projected into $[-1,1]^N$. Here $\ell$ is a large constant which cannot grow with $N$. (In \cite{ams20} and the present work, the final iterate is rounded to be in $\bSigma_N$ but this step does not change the asymptotic value of $H_N$ and is essentially irrelevant - see for instance Equation~\eqref{eq:approx1}.) In fact for a broad class of models, their main result based on the overlap gap property applies exactly when $\gamma_*^{\cuU}$ is not optimizable. This justifies our definition of $(\xi,\mathcal L_h)$ as having ``no overlap gap" if and only if it is optimizable.

\begin{proposition}\label{prop:OGP}
Suppose $\gamma_*^{\cuU}$ is not optimizable, where $\xi$ is an even polynomial and the external field $\mathcal L_h=\delta_h$ is constant. Then there is $\eps>0$ such that for any AMP algorithm with random output $\bsigma$,
\[
  \mathbb P\lt[\frac{H_N(\bsigma)}{N}\leq GS_{\xi,h}-\eps\rt]\geq 1-e^{-\Omega(N)}.
\]
\end{proposition}

The proof of Proposition~\ref{prop:OGP} is identical to that of \cite[Theorem 3.3]{gamarnik2019overlap} and we give an outline in Section~\ref{sec:OGP}. Taken together, Corollary~\ref{cor:generalq} and Proposition~\ref{prop:OGP} exactly characterize the mean-field Ising spin glasses for which approximate maximization is possible by AMP, at least when $\xi$ is even and the external field is constant. We remark that similar lower bounds were studied for the class of constant-degree polynomial algorithms in \cite{GJW20}. These results also extend to any non-optimizable Ising spin glass with even $\xi$ and constant $h$, ruling out approximate maximization algorithms in a slightly weaker sense. Constant-degree polynomials encompass AMP in most cases by approximating each non-linearity $f_{\ell}$ by a polynomial in a suitable sense, see e.g. \cite[Theorem 6]{mon18}.

We conclude this subsection with a brief discussion on our terminology. Our definition of optimizability is closely related to ``full" or ``continuous" replica symmetry breaking. For example, the definitions of full RSB used in \cite{mon18,subag2018following} essentially coincide with $0$-optimizability. However these terms seem to be slightly ambiguous, as they can also refer to functions $\gamma_*^{\cuU}$ which are strictly increasing on \textbf{any} nontrivial interval instead of being piece-wise constant as in finite RSB. For example, the physics paper \cite{charbonneau2014exact} describes 
``the
case where the function $\Delta(x)$ is allowed to have a continuous part: this can be thought
as an appropriate limit of the $k$-RSB construction when $k\to\infty$ and is therefore called
‘fullRSB’ or ‘$\infty$-RSB’~". Adding to the potential confusion, \cite{auffinger2017sk} uses the term ``infinite step" RSB to refer to functions $\gamma_*^{\cuU}$ with infinity many points of increase, possibly at a discrete set.
We therefore use ``no overlap gap" as an unambiguous term for the condition that $\gamma_*^{\cuU}$ is optimizable, while keeping in mind that it closely is implied via Lemma~\ref{lem:optimizable} by a strong, specific form of full RSB.

\subsection{Branching IAMP and Spherical Spin Glasses}
\label{subsec:branching-IAMP-intro}

Under no overlap gap, one expects that any finite ultrametric space of diameter at most $\sqrt{2(1-\lbq)}$ (with size independent of $N$) can be realized by approximate maximizers of $H_N$. In fact a modification of our $\lbq$-IAMP algorithm is capable of explicitly producing such realizations. In Section~\ref{sec:branch} we give a \emph{branching} $\lbq$-IAMP algorithm which for any finite ultrametric space $X$ and optimizable $\gamma_*$ constructs points $(\bsigma_x)_{x\in X}$ such that $\frac{H_N(\bsigma_x)}{N}\simeq\Par(\gamma_*)$ and $\frac{\|\bsigma_x-\bsigma_y\|_2}{\sqrt{N}}\simeq d_X(x,y)$ for each $x,y\in X$. 
{\color{black}
Recall that an ultrametric space $X$ is a metric space which satisfies the ultrametric triangle inequality 
\[
d_X(x,y)\leq\max(d_X(x,z),d_X(y,z)),
\quad
\forall~x,y,z\in X.
\]
Moreover any finite ultrametric can be canonically identified with the leaf set of a rooted tree, see e.g. \cite{bocker1998recovering}.
}

The idea is to occasionally reset the IAMP part of the algorithm with external randomness. A similar strategy was proposed but not analyzed in \cite{alaoui2020algorithmic}. 

\begin{restatable}{theorem}{branch}
\label{thm:branch}
Let $\gamma_*\in\cuL$ be optimizable, and fix a finite ultrametric space $(X,d_X)$ with diameter at most $\sqrt{2(1-\lbq)}$ as well as $\eps>0$. Then an efficient AMP algorithm constructs points $\{\bsigma_x|x\in X\}$ in $\bSigma_N$ satisfying

\begin{align*}
    \frac{H_N(\bsigma_x)}{N}
    &{\color{black}\geq \Par(\gamma_*)-\eps},
    \quad x\in X,
    \\
    \frac{\|\bsigma_x-\bsigma_y\|}{\sqrt N}
    &\in 
    \lt[d_X(x,y)-\eps,d_X(x,y)+\eps\rt],
    \quad x,y\in X
\end{align*}
with probability tending to $1$ as $N\to\infty$.
\end{restatable}

In Section~\ref{sec:sphere} we give corresponding results for spherical spin glasses, extending \cite{subag2018following} to the case of non-trivial external field. At zero temperature, \cite[Theorem 1]{chen2017parisi} determines the free energy in spherical spin glasses based on a positive, non-decreasing function $\alpha:[0,1)\to [0,\infty)$ as well as a constant $L$. (See also \cite{jagannath2017low} for related results.) More precisely, they show the asymptotic ground state energy is given by the unique minimizer to the variational problem:
\begin{align}
\label{eq:sphere-alpha}
    GS_{\ss}(\xi,h)
    &=
    \min_{L,\alpha\in \mathcal K}\mathcal Q(L,\alpha);
    \\
\nonumber
    \mathcal K
    &=
    \lt\{(L,\alpha)\in (0,\infty)\times \cuU ~:~L>\int_0^1 \alpha(s)\de s\rt\}
    ;
    \\
\nonumber
    2\mathcal Q(L,\alpha)
    &=
    (\xi'(1)+h^2)L-\int_0^1 \xi''(q)
    \lt(\int_0^q \alpha(s)\de s\rt)\de q + \int_0^1 \frac{\de q}{L-\int_0^q \alpha(s)\de s}.
\end{align}

The associated definition of no overlap gap is as follows.

\begin{definition}
    The spherical mixed $p$-spin model is said \emph{no overlap gap} if for some $\lbq_{\ss}\in [0,1)$, the unique minimizing $\alpha\in\cuU$ in \eqref{eq:sphere-alpha} is strictly increasing on $[\lbq_{\ss},1)$ and satisfies $\alpha(q)=0$ for $q\leq\lbq_{\ss}$. 
\end{definition}

Unlike the Ising case, we do not formulate a generalized variational principle and only show how to achieve a natural energy value, which coincides with the ground state energy when no overlap gap holds by \cite[Proposition 2]{chen2017parisi}. We also exactly characterize the spherical models exhibiting no overlap gap, which slightly extends the same result.

\begin{restatable}{theorem}{sphereGS}
\label{thm:sphereGS}
Suppose $\xi$ and $\mathcal L_h$ satisfy $\E[h^2]+\xi'(1)<\xi''(1)$, and let $\lbq_{\ss}\in (0,1)$ be the unique solution to $\E[h^2]+\xi'(\lbq_{\ss})=\lbq_{\ss}\xi''(\lbq_{\ss})$. Then the spherical spin glass with parameters $\xi,\mathcal L_h$ has no overlap gap if and only if $\xi''(q)^{-1/2}$ is concave on $q\in[\lbq_{\ss},1]$, in which case $\alpha$ is supported on $[\lbq_{\ss},1]$ and takes the explicit form
\[
  \alpha(s)=
  \bigg\{\begin{array}{lr}
    0, &s\in [0,\lbq_{\ss})\\
    \frac{\xi'''(s)}{2\xi''(s)^{3/2}}, &s\in [\lbq_{\ss},1].
  \end{array}
\] 
Moreover the ground-state energy satisfies
\[
  GS_{\ss}(\xi,\mathcal L_h)\geq \lbq_{\ss}\sqrt{\xi''(\lbq_{\ss})}+\int_{\lbq_{\ss}}^1 \sqrt{\xi''(q)}\de q
\] 
with equality if and only if no overlap gap occurs.
\end{restatable}

\begin{restatable}{theorem}{sphereopt}
\label{thm:sphereopt}
Suppose $\xi$ and $h\sim\mathcal L_h$ satisfy $\E[h^2]+\xi'(1)<\xi''(1)$, and let $\lbq_{\ss}\in (0,1)$ be the unique solution to $\E[h^2]+\xi'(\lbq_{\ss})=\lbq_{\ss}\xi''(\lbq_{\ss})$. Then there exists an efficient AMP algorithm outputting $\bsigma\in\mathbb S^{N-1}(\sqrt N)$ such that
\[
  \frac{H_N(\bsigma)}{N} \simeq\lbq_{\ss}\sqrt{\xi''(\lbq_{\ss})}+\int_{\lbq_{\ss}}^1 \sqrt{\xi''(q)}\de q.
\]
If on the other hand $\E[h^2]+\xi'(1)\geq\xi''(1)$, then there is an efficient AMP algorithm outputting $\bsigma\in\mathbb S^{N-1}(\sqrt N)$ with
\[
  \frac{H_N(\bsigma)}{N} \simeq \sqrt{\E[h^2]+\xi'(1)}.
\]
\end{restatable}

\begin{remark}
\label{rem:spherers} 
    If $\E[h^2]+\xi'(1)\geq\xi''(1)$ then the model is replica-symmetric by \cite[Proposition 1]{chen2017parisi}. When $\E[h^2]+\xi'(1)<\xi''(1)$, the function $f(q)=q\xi''(q)-\xi'(q)-\mathbb E[h^2]$ is increasing and satisfies $f(0)<0<f(1)$, hence has a unique root $\lbq_{\ss}\in (0,1)$. 
\end{remark}

\begin{remark}
    Subsequently to the present work and in collaboration with Brice Huang, we showed in \cite{huang2021tight} that the algorithms presented in this paper are optimal in some sense. More precisely, it is shown that Proposition~\ref{prop:OGP} can be strengthed to say that
    \begin{equation}
    \label{eq:B-OGP}
      \mathbb P\lt[\frac{H_N(\bsigma)}{N}\leq \inf_{\gamma\in\cuL}\Par(\gamma)+\eps\rt]\geq 1-e^{-\Omega(N)}
    \end{equation}
    for any $\eps>0$.
    Here, as in Proposition~\ref{prop:OGP}, $\bsigma\in [-1,1]^N$ is the output of an AMP algorithm whose number of iterations is independent of $N$. This result applies more generally to arbitrary algorithms with suitably Lipschitz dependence on the disorder variables defining $H_N$. In the spherical case, \cite{huang2021tight} similarly shows that the energy attained in Theorem \ref{thm:sphereopt} is asymptotically best possible for Lipschitz algorithms; see Proposition 2.2 therein.

    The essential idea of \cite{huang2021tight} is to consider general finite ultrametric spaces (with size independent of $N$) of points $\bsigma$ with large energy. They show that for $E>\inf_{\gamma\in\cuL}\Par(\gamma)$, the level sets
    \[
      S_E\equiv \lt\{\bsigma\in [-1,1]^N:\frac{H_N(\bsigma)}{N}\geq E\rt\}
    \]
    \textbf{do not} contain approximately isometric embeddings of sufficiently complicated finite ultrametrics. (Technically proving \eqref{eq:B-OGP} requires a more complicated obstruction involving a correlated family of different Hamiltonians.) Theorem~\ref{thm:branch} is a sharp converse to and was a key inspiration for this result, as it constructs arbitrary finite ultrametric configurations at energy $\Par(\gamma_*)$ for optimizable $\gamma_*$. See the introduction of \cite{huang2021tight} for further discussion and implications.
\end{remark}

\subsection*{Acknowledgement}

We thank Ahmed El Alaoui, Brice Huang, Andrea Montanari, and the anonymous referees for helpful comments. We thank David Gamarnik and Aukosh Jagannath for clarifying the terminological issues around full RSB. We thank Wei-Kuo Chen for communicating the proof of the lower bound with equality case for $GS_{\ss}$ in Theorem~\ref{thm:sphereGS}. This work was supported by an NSF graduate research fellowship and a William R. and Sara Hart Kimball Stanford graduate fellowship.

\section{Technical Preliminaries}
\label{sec:prelim}

We will use ordinary lower-case letters for scalars $(m,x,\dots,)$ and bold lower-case for vectors $(\bm,\bx)$. Ordinary upper-case letters are used for the state-evolution limits of AMP as in Proposition~\ref{prop:AMP} such as $(X_j^{\delta},Z_j^{\delta},N_j^{\delta})$ as well as for continuous-time stochastic processes such as $(X_t,Z_t,N_t)$. We denote limits in probability as $N\to\infty$ by $\plim_{N\to\infty}(\cdot)$. We write $x\simeq y$ to indicate that $\plim_{N\to\infty}(x-y)=0$ where $x,y$ are random scalars.
 
We will use the ordinary inner product $\langle \bx,\by\rangle=\sum_{i=1}^N x_iy_i$ as well as the normalized inner product $\langle \bx,\by\rangle_N=\frac{\sum_{i=1}^N x_iy_i}{N}.$ Here $\bx=(x_1,\dots,x_N)\in\mathbb R^N$ and similarly for $\by$. Associated with these are the norms $\|\bx\|=\sqrt{\langle \bx,\bx\rangle}$ and $\|\bx\|_N=\sqrt{\langle \bx,\bx\rangle_N}.$ We will also use the notation $\langle \bx\rangle_N=\frac{\sum_{i=1}^N x_i}{N}$. Often, for example in \eqref{eq:ampdef}, we apply a scalar function $f$ to a vector $\bx\in\mathbb R^N$. This will always mean that $f$ is applied entrywise, i.e. $f(x_1,\dots,x_N)=(f(x_1),\dots,f(x_N))$. Similarly for a function $f:\mathbb R^{\ell+1}\to\mathbb R$, we define 
\begin{equation}
\label{eq:coordinatewise-def}
    f(\bx^0,\bx^1,\dots,\bx^{\ell})
    =
    \big(f(x^0_1,x^1_1,\dots,x^{\ell}_1),\,
    f(x^0_2,x^1_2,\dots,x^{\ell}_2),\,\dots\,
    f(x^0_N,x^1_N,\dots,x^{\ell}_N)\big)
    \in \mathbb R^N.
\end{equation}

The following useful a priori estimate shows that all derivatives of {\color{black}$\wt H_N$} have order $1$ in the $\|\cdot\|_N$ norm. Note that we do not apply any non-standard normalization in the definitions of gradients, Hessians, etc.
{\color{black}
We use $\|\cdot\|_{\op}$ to denote the operator norm of a tensor $\bT\in(\mathbb R^N)^{\otimes k}$ of arbitrary order $k$:
\begin{equation}
\label{eq:operator-norm-def}
    \|\bT\|_{\op}
    =
    \sup_{\bx\in\mathbb R^N} 
    \frac{|\langle \bT,\bx^{\otimes k}\rangle|}{\|\bx\|^k}.
\end{equation}
}

\begin{proposition}
[{\cite[Lemma C.1]{arous2020geometry}}]
\label{prop:lip}
Fix a mixture function $\xi$, external field distribution $\mathcal L_h$, $k\in\mathbb Z^+,\eta\in\mathbb R^+$, and assume that the coefficients of $\xi$ decay exponentially. Then for suitable $C=C(\xi,\mathcal L_h,k,\eta)$,  
\[
    \mathbb P\lt[\sup_{\|\bx\|\leq \lt(1+\eta\rt)\sqrt{N}}\|\nabla^k \wt H_N(\bx)\|_{\op}\leq CN^{1-\frac{k}{2}}\rt]\geq 1-e^{-\Omega(N)}.
\] 
\end{proposition}

\subsection{Review of Approximate Message Passing}\label{subsec:amp}

Here we review the general class of approximate message passing (AMP) algorithms. AMP algorithms are a flexible class of efficient algorithms based on a random matrix or, in our setting, mixed tensor. To specify an AMP algorithm, we fix a probability distribution $p_0$ on $\mathbb R$ with finite second moment and a sequence $f_0,f_1,\dots$ of Lipschitz functions $f_{\ell}:\mathbb R^{\ell+1}\to\mathbb R$, with $f_{-1}=0$. The functions $f_{\ell}$ will often be referred to as \emph{non-linearities}. We begin by taking $\bz^0\in\R^{N}$ to have i.i.d. coordinates $(z^0_i)_{i\in [N]}\sim p_0$. Then we recursively define $\bz^{1},\bz^2,\dots$ via
\begin{align}
\label{eq:ampdef}
    \bz^{\ell+1}
    &=
    \nabla \wt H_N(f_{\ell}(\bz^0,\dots,\bz^{\ell}))-\sum_{j=1}^{\ell} d_{\ell,j}f_{j-1}(\bz^0,\dots,\bz^{j-1}),
    \\
\label{eq:onsagerdef}    
    d_{\ell,j}
    &=
    \xi''\lt(\lt\la f_{\ell}(\bz^0,\dots,\bz^{\ell}),f_{j-1}(\bz^0,\dots,\bz^{j-1})\rt\ra_N\rt)
    \cdot 
    \mathbb E\lt[\frac{\partial f_{\ell}}{\partial Z^j}(Z^0,\dots,Z^{\ell})\rt].
\end{align}
Here the non-linearity $f_{\ell}$ is applied coordinate-wise as in \eqref{eq:coordinatewise-def}. Moreover $Z^0\sim p_0$ while $(Z^{\ell})_{\ell\geq 1}$ is an independent centered Gaussian process with covariance $Q_{\ell,j}=\mathbb E[Z^{\ell}Z^j]$ defined recursively by
\begin{equation}
    Q_{\ell+1, j+1}=\xi'\lt(\mathbb{E}\lt[f_{\ell}\lt(Z^{0}, \cdots, Z^{\ell}\rt) f_{j}\lt(Z^{0}, \cdots, Z^{j}\rt)\rt]\rt), \quad \ell,j \geq 0.
\end{equation}

The key property of AMP, stated below in Proposition~\ref{prop:AMP}, is that for any $\ell$ the empirical distribution of the $N$ sequences $(\bz^1_i,\bz^2_i,\dots,\bz^{\ell}_i)_{i\in [N]}$ converges in distribution to the law of the Gaussian process $(Z^1,\dots,Z^{\ell})$ as $N\to\infty$. This is called \emph{state evolution}. 

\begin{definition}

For non-negative integers $n,m$ the function $\psi:\mathbb R^{\ell}\to\R$ is \emph{pseudo-Lipschitz} if for some constant $L$ and any $x,y\in\R^{\ell}$,
\[
    \|\psi(x)-\psi(y)\|\leq L(1+\|x\|+\|y\|) \|x-y\|.
\]
\end{definition}

\begin{proposition}
[{\cite[Proposition 3.1]{ams20}}]
\label{prop:AMP}
For any pseudo-Lipschitz $\psi:\mathbb R^{\ell+1}\to\R$, the AMP iterates satisfy
\[
    \plim_{N\to\infty}\lt\la\psi\lt(\boldsymbol{z}^{0}, \cdots, \boldsymbol{z}^{\ell}\rt)\rt\ra_{N} = \mathbb{E}\lt[\psi\lt(Z^{0}, \cdots, Z^{\ell}\rt)\rt].
\]
\end{proposition}

The first version of state evolution was given for Gaussian random matrices in \cite{bolthausen2014iterative,BM-MPCS-2011}. Since then it has been extended to more general situations in many works including \cite{javanmard2013state,bayati2015universality,berthier2019state,chen2020universality,fan2020approximate}. As state evolution holds for essentially arbitrary non-linearities $f_{\ell}$, it allows a great deal of flexibility in solving problems involving random matrices or tensors. 

We remark that \cite[Proposition 3.1]{ams20} is phrased in terms of a random mixed tensor $\bW$, i.e. a sequence of $p$-tensors $(\bW^{(p)}\in(\mathbb R^N)^{\otimes p})_{p\geq 2}$ - see Equation (3.2) therein. The two descriptions are equivalent because $\bW$ is constructed so that $\sum_{p\geq 2} c_p\langle \bW^{(p)},\bx^{\otimes p}\rangle = \wt H_N(\bx).$ While the tensor language is better suited to proving state evolution, for our purposes it is more convenient to express AMP just in terms of $\wt H_N$ and $\nabla\wt H_N$.

Let us finally discuss the efficiency of our AMP algorithms. The algorithms we give are described by parameters $\ubq$ and $\ul$ and require oracle access to the function $\Phi_{\gamma_*}(t,x)$ and its derivatives. We do not address the complexity of computing $\Phi_{\gamma_*}(t,x)$. However as stated in \cite{mon18,ams20} it seems unlikely to present a major obstacle because solving for $\gamma_*^{\cuU}$ is a convex problem which only must be solved once for each $(\xi,\mathcal L_h)$. Moreover \cite{alaoui2020algorithmic} demonstrates that these algorithms are practical to implement.

In the end, our algorithms output rounded points $\bsigma$ with $\bsigma_i=\sign(f_{\ubl}(\bz^0_i,\dots,\bz^{\ubl}_i))$ for a large value $\ubl=\ubl(\ubq,\ul)$. The outputs satisfy 
\[
    \lim_{\ubq\to 1}\lim_{\ul\to\infty}\plim_{N\to\infty}\frac{H_N(\bsigma)}{N}=H_*
\]
for some asymptotic energy value $H_*$. To achieve an $\eps$-approximation to the value $H_*$, the parameters $\ubq$ and $\ul$ must be sent to $1$ and $\infty$ which requires a diverging number of iterations. In particular let $\chi$ denote the complexity of computing $\nabla \wt H_N$ at a point and let $\chi_1$ denote the complexity of computing a single coordinate of $\nabla \wt H_N$ at a point. Then the total complexity needed to achieve energy $H_*-\eps$ is $C(\eps)(\chi+N)+N\chi_1$. When $\xi$ is a polynomial this complexity is linear in the size of the input specifying $H_N$ - see the comments following \cite[Remark 2.1]{ams20}. In the statements of our results, we refer to such algorithms as ``efficient AMP algorithms".

\subsection{Initializing AMP}
\label{subsec:init}

Here we explain some technical points involved in initializing our AMP algorithms and why they arise. First, we would like to use a random external field $h_i$ which varies from coordinate to coordinate. In the most natural AMP implementation, this requires that the non-linearities $f_{\ell}$ correspondingly depend on the coordinate rather than being fixed, which is not allowed in state evolution. Second we would like to use many i.i.d. Gaussian vectors throughout the branching version of the algorithm. However Proposition~\ref{prop:AMP} allows only a single initial vector $\bz^0$ as a source of external randomness independent of $H_N$. One could prove a suitable generalization of Proposition~\ref{prop:AMP}, but we instead build these additional vectors into the initialization of the AMP algorithm as a sort of preprocessing phase. To indicate that our constructions here are preparation for the ``real algorithm", we reparametrize so the preparatory iterates have negative index.

We begin by taking $p_0=\mathcal L_h$ to be the distribution of the external field itself, and initialize $(\bz^{-K})_i=h_i\sim \mathcal L_h$ for some constant $K\in\mathbb Z^+$. We then set $f_{-K}(\bz^{-K})=\frac{\bz^{-K}}{\sqrt{\mathbb E^{h\sim \mathcal L_h}[h^2]}}$ and $f_{-k}(\bz^{-K},\dots,\bz^{-k})=\bz^{-k}$ for $2\leq k\leq K$. Finally we set $f_{-1}(\bz^{-K},\dots,\bz^{-1})=c\bz^{-1}$ for some constant $c>0$ which the algorithm is free to choose. (Note that the functions $f_{-k}$ correspond to entry-wise applications of the form in \eqref{eq:coordinatewise-def}.) State evolution immediately implies the following

\begin{proposition}
\label{prop:init}
In the state evolution $N\to\infty$ limit, the empirical distribution of $(z^{-K}_i,\dots,z^{0}_i)$ (for $i\in [N]$ uniformly random) converges in probability to the law of an independent $(K+1)$-tuple $(Z^{-K},Z^{-K+1},\dots,Z^{-1},Z^0)$ with $Z^{-K}\sim \mathcal L_h$, $(Z^{-K+1},\dots,Z^{-1})\sim \normal(0,I_{K-1})$ i.i.d. standard Gaussian, and $Z^0\sim\normal (0,c^2)$.
\end{proposition}

In fact taking $K=1$ suffices for the main construction in the paper. In Section~\ref{sec:branch} we require larger values of $K$ for branching IAMP, where the iterates $(\bz^{-K+1},\dots,\bz^{-1})$ serve as proxies for i.i.d. Gaussian vectors.

\begin{remark} 
\label{rem:h} 
    Because the sum defining the Onsager correction term in \eqref{eq:ampdef} starts at $j=1$, the effect of the external field $h_i$ on future AMP iterates does not enter into any Onsager correction terms in this paper.
\end{remark}

\subsection{Properties of the Parisi PDE and SDE}\label{sub:parisi}

Quite a lot is known about the solution $\Phi_{\gamma}$ to the Parisi PDE. The next results hold for any $\gamma\in\cuL$ and are taken from \cite{ams20}. Similar results for $\gamma\in\cuU$ appear in \cite{auffinger2015parisi,jagannath2016dynamic}.

\begin{proposition}
[{\cite[Lemmas 6.2, 6.4]{ams20}}]
\label{prop:phireg}

For any $\gamma\in \cuL$, the solution $\Phi_{\gamma}(t,x)$ to the Parisi PDE is continuous on $[0,1]\times \reals$, convex in $x$, and further satisfies the following regularity properties for any $\eps>0$.
\begin{itemize}
\item[$(a)$] 
    $\partial_x^j\Phi \in L^{\infty}\big([0,1-\eps];L^2(\reals)\cap L^{\infty}(\reals)\big)$ for $j\ge 2$.
\item[$(b)$] 
    $\partial_t\Phi \in L^{\infty}([0,1]\times \reals)$ and $\partial_t\partial_x^j\Phi\in L^{\infty}\big([0,1-\eps];L^2(\reals)\cap L^{\infty}(\reals)\big)$ for $j\ge 1$.
\end{itemize}
\end{proposition}

\begin{proposition}
[{\cite[Lemmas 6.2, 6.4]{ams20}}]
\label{prop:1lip}
    For any $\gamma\in \cuL$, $\Phi_{\gamma}$ satisfies 
    \[
    % \label{eq:philip}
    |\partial_x\Phi_{\gamma}(t,x)|\leq 1
    \]
    for all $(t,x)\in [0,1]\times\mathbb R$. 
\end{proposition}

\begin{proposition}[{\cite[Lemma 6.5]{ams20}}]
\label{prop:Xt}
For any $\gamma\in \cuL$, the Parisi SDE~\eqref{eq:parisiSDE} has unique strong solution $(X_t)_{t\in [0,1]}$ which is a.s. continuous and satisfies 
\begin{equation}
\label{eq:DxIntegral}
    \partial_x\Phi_{\gamma}(t,X_t) = \int_0^t \sqrt{\xi''(s)}\, \partial_{xx}\Phi_{\gamma}(s,X_s)\, \de B_s\, .
\end{equation}
\end{proposition}

Finally we give two additional properties for optimizable $\gamma_*$, which are proved in Section~\ref{sec:identity}.

\begin{restatable}{lemma}{lemidentitytwo}
\label{lem:identity2}

If $\gamma_*\in\cuL$ is $\lbq$-optimizable then it satisfies:
\begin{align}
\label{eq:id2}
    \mathbb E[\partial_{xx}\Phi_{\gamma_*}(t,X_t)^2]
    &= 
    \frac{1}{\xi''(t)},\quad t\geq \lbq,
    \\
\label{eq:id3}
    \mathbb E[\partial_{xx}\Phi_{\gamma_*}(t,X_t)]
    &{\color{black}\geq }
    \int_t^1 \gamma_*(s)\de s,\quad t\in [0,1]
    .
\end{align}
\end{restatable}

\begin{remark}
{\color{black}
    We expect \eqref{eq:id3} to hold with equality; if this is true, then our analysis in Subsection~\ref{subsec:final-energy}
    shows that Theorem~\ref{thm:main} holds as a two-sided estimate. Conversely, the main result of \cite{huang2021tight} implies such a two-sided estimate when $\xi$ is even; retracing Subsection~\ref{subsec:final-energy} then implies \eqref{eq:id3} is indeed an equality in such cases. However this is unsatisfyingly indirect and it would be interesting to give a direct proof.
}
\end{remark}

\section{The Main Algorithm}
\label{sec:algo}

In this section we explain our main AMP algorithm and prove Theorem~\ref{thm:main}. Throughout we take $\gamma_*\in\cuL$ to be $\lbq$-optimizable for $\lbq=\inf(\supp(\gamma_*))\in [0,1)$.

\subsection{Phase 1: Finding the Root}\label{subsec:root}

Here we give the first phase of the algorithm, which proceeds for a large constant number $\ul$ of iterations after initialization and approximately converges to a fixed point. The AMP iterates during this first phase are denoted by $(\bw^{k})_{-K\leq k\leq \ul}$. We rely on the function 
\[
    f(x)=\partial_x\Phi_{\gamma_*}(\lbq,x)
\] 
and use non-linearities 
\[
    f_k(\bh,\bw^{-K+1},\dots,\bw^0,\bw^1,\dots,\bw^k)=f(\bh+\bw^k)
\] 
for all $k\geq 1$. (As a reminder, if $f$ is a scalar function, $f(\bx^k)$ is evaluated entrywise as explained in \eqref{eq:coordinatewise-def}.) Proposition~\ref{prop:phireg} implies that each $f_k$ is Lipschitz, so that state evolution applies to the AMP iterates. In the initialization phase we take $c=\sqrt{\xi'(\lbq)}$ as described in Subsection~\ref{subsec:init}, so that the coordinates $w^0_i$ are asymptotically distributed as centered Gaussians with variance $\xi'(\lbq)$ in the $N\to\infty$ limit. Moreover we set $\bm^k=f(\bx^k)$ where $\bx^k=\bw^k+\bh$. This yields the following iteration.
\begin{align}
\label{eq:RSAMP}
  \bw^{k+1}
  &=
  \nabla \wt H_N(f(\bx^k))-f(\bx^{k-1})\xi''\lt(\langle f(\bx^k),f(\bx^{k-1})\rangle_N\rt)\langle f'(\bx^k)\rangle_N
  \\
\nonumber
  &=\nabla \wt H_N(\bm^k)
  -\bm^{k-1}\xi''\lt(\langle \bm^k,\bm^{k-1}\rangle_N\rt)\langle \partial_{xx}\Phi_{\gamma_*}(\lbq,\bx^k)\rangle_N,
  \\
\nonumber
  \bx^{k+1}
  &=
  \bw^{k+1}+\bh
  \\
\nonumber
  \bm^k
  &=
  f(\bx^k)=f_k(\bw^k).
\end{align}

\begin{lemma}\label{lem:1}

For $f$ as defined above, $h\sim\mathcal L_h$ and $Z\sim \normal(0,1)$ an independent standard Gaussian,
\begin{align}
\label{eq:fsquared}
    \mathbb E^{h,Z}\lt[f\lt(h+Z\sqrt{\xi'(\lbq)}\rt)^2\rt]
    &=\lbq
    \\
\label{eq:fprimesquare}
    \mathbb E^{h,Z}\lt[f'\lt(h+Z\sqrt{\xi'(\lbq)}\rt)^2\rt]
    &= 
    \frac{1}{\xi''(\lbq)}.
\end{align}
\end{lemma}

\begin{proof}

The identities follow by taking $t=\lbq$ in the definition of optimizability as well as Lemma~\ref{lem:identity2}. Here we use the fact that $X_t\stackrel{d}{=}X_0+ Z\sqrt{\xi'(t)}$ is a time-changed Brownian motion started from $X_0$ for $t\leq \lbq.$
\end{proof}

Next with $(Z,Z',Z'')\sim \normal(0,I_3)$ independent of $h\sim\mathcal L_h$, define for $t\leq \xi'(\lbq)$ the function

\begin{equation}\label{eq:phi}\phi(t)=\mathbb E^{h,Z,Z',Z''}\lt[f\lt(h+Z\sqrt{t}+Z'\sqrt{\xi'(\lbq)-t}\rt)f\lt(h+Z\sqrt{t}+Z''\sqrt{\xi'(\lbq)-t}\rt)\rt].\end{equation} Define also $\psi(t)=\xi'(\phi(t))$.  It follows from~\eqref{eq:fsquared} that 

\begin{equation}\label{eq:phiid}\phi(\xi'(\lbq))=\lbq.\end{equation}

\begin{lemma}\label{lem:psi}
The function $\psi$ is strictly increasing and strictly convex on $[0,\xi'(\lbq)]$. Moreover 
\[
\psi(\xi'(\lbq))=\xi'(\lbq),\quad \psi'(\xi'(\lbq))=1.
\] 
Finally $\psi(t)>t$ for all $t<\xi'(\lbq)$. 
\end{lemma}
\begin{proof}
Using Gaussian integration by parts as in~\cite[Lemma 2.2]{bolthausen2014iterative}, we find
\begin{align*}
    \phi'(t)&=\E^{h,Z,Z',Z''}\lt[
    f'\lt(h+\sqrt{t}Z+\sqrt{\xi'(\lbq)-t}Z'\rt) 
    f'\lt(h+\sqrt{t}Z+\sqrt{\xi'(\lbq)-t}Z''\rt)
    \rt]
    \\
    &=
    \E^{h,Z}\lt[
    \E^{Z'}
    \lt[f'\lt(h+\sqrt{t}Z+\sqrt{\xi'(\lbq)-t}Z'\rt)\rt]^2 
    \rt],
    \\
    \phi''(t)
    &=
    \E\lt[
    f''\lt(h+\sqrt{t}Z+\sqrt{\xi'(\lbq)-t}Z'\rt) 
    f''\lt(h+\sqrt{t}Z+\sqrt{\xi'(\lbq)-t}Z''\rt)
    \rt]
    \\
    &=
    \E^{h,Z}\lt[\E^{Z'}
    \lt[f''\lt(h+\sqrt{t}Z+\sqrt{\xi'(\lbq)-t}Z'\rt)\rt]^2 \rt].
\end{align*}
These expressions are each strictly positive, as the optimizability of $\gamma_*$ implies that $f',f''$ are not identically zero. Therefore $\phi$ is increasing and convex. Since $\xi'$ is also increasing and convex (being a power series with non-negative coefficients) we conclude the same about their composition $\psi$. The values $\psi(\xi'(\lbq))=\xi'(\lbq)$ and $\psi'(\xi'(\lbq))=1$ follow from Lemma~\ref{lem:1} and the chain rule. Finally the last claim follows by strict convexity of $\psi$ and $\psi'(\xi'(\lbq))=1$.
\end{proof}

Next, let $h,W^{-1},(W^j,X^j,M^j)_{j\geq 0}$ be the state evolution limit of the coordinates of \[(\bh,\bw^{-1},\bw^0,\bx^{0},\bm^{0},\dots,\bw^k,\bx^k,\bm^k)\] as $N\to\infty$. Hence each $W^j$ is a centered Gaussian and $X^j=W^j+h$, $M^{j+1}=f(X^j)$ hold for $j\geq 0$. Define the sequence $(a_0,a_1,\dots)$ recursively by $a_0=0$ and $a_{k+1}=\psi(a_k)$. 

\begin{lemma}\label{lem:RSconverge}

For all non-negative integers $0\leq j<k$, the following equalities hold:
\begin{align} \label{eq:id1.0}\mathbb E[(W^j)^2]&=\xi'(\lbq)\\
\label{eq:id2.0}\mathbb E[W^jW^k]&=a_j\\
\label{eq:id3.0}\mathbb E[(M^j)^2]&=\lbq\\
\label{eq:id4.0}\mathbb E[M^jM^k]&=\phi(a_j).
\end{align}

Moreover $(W^j)_{j\geq 0}$ is independent of $h$.

\end{lemma}

\begin{proof}

We proceed by induction on $j$, first showing \eqref{eq:id1.0} and \eqref{eq:id3.0} together. As a base case, \eqref{eq:id1.0} holds for $j=0$ by initialization. For the inductive step, assume first that \eqref{eq:id1.0} holds for $j$. Then state evolution and \eqref{eq:phiid} yield 
\[
  \mathbb E\lt[(M^j)^2\rt]=\phi\lt(\xi'(\lbq)\rt)=\lbq
\] 
so that \eqref{eq:id1.0} implies \eqref{eq:id3.0} for each $j\geq 0$. On the other hand, state evolution directly implies that if \eqref{eq:id3.0} holds for $j$ then \eqref{eq:id1.0} holds for $j+1$. This establishes \eqref{eq:id1.0} and \eqref{eq:id3.0} for all $j\geq 0$.

We similarly show \eqref{eq:id2.0} and \eqref{eq:id4.0} together by induction, beginning with \eqref{eq:id2.0} when $j=0$. By the initialization of Subsection~\ref{subsec:init} it follows that the random variables $h,W^{-1},W^0$ are jointly independent. State evolution implies that $W^{k-1}$ is independent of $W^{-1}$ for any $k\geq 0$. Then state evolution yields for any $k\geq 1$:

\begin{align*}
    \mathbb E[W^0W^k]&=\xi'(\mathbb E[M^{-1}M^{k-1}])
    \\
    &=
    \xi'\lt(\mathbb E\lt[W^{-1}f(W^{k-1}\rt]\rt)\\
    &=
    \xi'(0)
    \\
    &=
    0.
\end{align*}

Just as above, it follows from state evolution that \eqref{eq:id2.0} for $(j,k)$ implies \eqref{eq:id4.0} for $(j,k)$ which in turn implies \eqref{eq:id2.0} for $(j+1,k+1)$. Hence induction on $j$ proves \eqref{eq:id2.0} and \eqref{eq:id4.0} for all $(j,k)$. Finally the last independence assertion is immediate from state evolution just because $h$ is the first step in the AMP iteration.
\end{proof}

\begin{lemma}\label{lem:RSlimit}

\[
\lim_{k\to\infty} a_k=\xi'(\lbq).
\]

\end{lemma}

\begin{proof}

Since $\psi$ is strictly increasing and maps $[0,\xi'(\lbq)]\to[0,\xi'(\lbq)]$, it follows that $(a_k)_{k\geq 0}$ is a strictly increasing sequence with limiting value in $[0,\xi'(\lbq)]$. Let $a_{\infty}=\lim_{k\to\infty} a_k$ be this limit. Then continuity implies $\psi(a_{\infty})=a_{\infty}$ which by the last part of Lemma~\ref{lem:psi} implies $a_{\infty}=\xi'(\lbq)$. This concludes the proof.
\end{proof}

We now compute the limiting energy 
\[
    \lim_{k\to\infty}\plim_{N\to\infty}\frac{H_N(\bm_k)}{N}
\] 
from the first phase. Since the first phase is similar to many ``standard" AMP algorithms, this step is comparable to the computation of their final objective value, for example~\cite[Lemma 6.3]{deshpande2017asymptotic}. This computation is straightforward when $\wt{H}_N$ is a homogeneous polynomial of degree $p$, because one can just rearrange the equation for an AMP iteration to solve for 
\[
    \wt H_N(\bm^k)=p^{-1}\langle \bm^k,\nabla\wt H_N(\bm^k)\rangle.
\] 
However it requires more work in our setting because $\nabla \wt{H}_N$ acts differently on terms of different degrees. We circumvent this mismatch by applying state evolution to a $t$-dependent auxiliary AMP step and integrating in $t$.

\begin{lemma}
\label{lem:RSenergy}
With $X_t$ the Parisi SDE~\eqref{eq:parisiSDE},
\begin{align*}
    \lim_{k\to\infty} \plim_{N\to\infty}\frac{H_N(\bm^k)}{N}
    &= \xi'(\lbq) \cdot \mathbb E
        \lt[
            \partial_{xx}\Phi_{\gamma_*}
            \lt(
                \lbq, h + Z\sqrt{\xi'(\lbq)}
            \rt)
        \rt]
        + \mathbb E
        \lt[
            h \cdot \partial_{x}\Phi_{\gamma_*}
            \lt(
                \lbq, h + Z\sqrt{\xi'(\lbq)}
            \rt)
        \rt] \\
    &= \xi'(\lbq) \cdot \mathbb E
        \lt[
            \partial_{xx}\Phi_{\gamma_*}
            \lt(
                \lbq, X_{\lbq}
            \rt)
        \rt]
        + \mathbb E
        \lt[
            h \cdot \partial_{x}\Phi_{\gamma_*}
            \lt(
                \lbq, X_{\lbq}
            \rt)
        \rt].
\end{align*}
\end{lemma}

\begin{proof}

The equivalence of the latter two expressions follows from the fact that $X_{\lbq}\sim X_0+ \normal(0,\xi'(\lbq))$ so we focus on the first equality. Observe that
\begin{equation}
\label{eq:outofnames}
    \frac{H_N(\bm^k)}{N}=\langle \bh, \bm^k\rangle_N+\int_0^1 \langle \bm^k,\nabla \wt H_N(t\bm^k)\rangle_N \de t
\end{equation} 
holds for any vector $\bm^k$ by considering each monomial term of $H_N$. Our main task now reduces to computing the in-probability limit of the integrand as a function of $t$. Proposition~\ref{prop:lip} ensures that $t\to \langle \bm^k,\nabla \wt H_N(t\bm^k)\rangle_N$ is Lipschitz assuming $\|\bm^k\|_N\leq 1+o(1)$. This holds with high probability for each $k$ as $N\to\infty$ by state evolution and Proposition~\ref{prop:1lip}, so we may freely interchange the limit in probability with the integral.

To compute the integrand $\langle \bm^k,\nabla \wt H_N(t\bm^k)\rangle_N$ we analyze a modified AMP which agrees with the AMP we have considered so far up to step $k$, whereupon we replace the non-linearity $f_k(\bh,\bw^k)=f(\bw^k+\bh)=f(\bx^k)$ by 
\[
    \tilde f_k(\bh,\bw^k)\equiv t\cdot f(\bx^k)
\] 
for arbitrary $t\in (0,1)$. We obtain the new iterate
\[
    \by^{k+1}(t)\equiv\nabla \wt H_N(t\bm^k)-t\bm^{k-1} \xi''(t\langle \bm^k,\bm^{k-1}\rangle_N )\langle f'(\bx^{k})\rangle_N. 
\] 
Rearranging yields
\begin{align*}
  \langle \bm^k,\nabla \wt H_N(t\bm^k)\rangle_N &= \langle \bm^k,\by^{k+1}(t)\rangle_N +t\langle \bm^k,\bm^{k-1} \rangle_N \xi''(t\langle \bm^k,\bm^{k-1}\rangle_N )\langle \bf'(\bx^{k})\rangle_N  \\
  &\simeq   \langle \bm^k,\by^{k+1}(t)\rangle_N +ta_{k-1} \xi''(t\phi(a_{k-1}) )\langle \bf'(\bx^{k})\rangle_N.
\end{align*}
We evaluate the $N\to\infty$ limit in probability of the first term, via the state evolution limits $W^k,X^k,Y^{k+1}(t)$. State evolution directly implies 
\[
    \mathbb E[W^kY^{k+1}(t)]=\xi'(t\cdot\mathbb E[M^{k-1}M^k])=\xi'(t\phi(a_{k-1})).
\] 
Since $h$ is independent of $(W^k,Y^{k+1})$ we use Gaussian integration by parts to derive 
\begin{align*}
  \mathbb E[f(X^k)Y^{k+1}(t)]
  &=\mathbb E[f(h+W^k)Y^{k+1}(t)]\\
  &=\mathbb E[f'(h+W^k)]\cdot\mathbb E[W^kY^{k+1}(t)]\\
  &=\mathbb E\lt[f'\lt(h+Z\sqrt{\xi'(\lbq)}\rt)\rt]\cdot\xi'(t\phi(a_{k-1})).
\end{align*} 
Integrating with respect to $t$ yields 
\begin{align*} 
  &\int_0^1 \langle \bm^k,\nabla \wt H_N(t\bm^k)\rangle_N \de t
  \\
  &\simeq 
  \mathbb E\lt[f'\lt(h+Z\sqrt{\xi'(\lbq)}\rt)\rt]\cdot 
  \int_0^1 \xi'(t\phi(a_{k-1}))+t \phi(a_{k-1})\xi''(t\phi(a_{k-1}))
  ~
  \de t
  \\
  &=
  \mathbb E\lt[\partial_{xx}\Phi_{\gamma_*}\lt(\lbq,h+Z\sqrt{\xi'(\lbq)}\rt)\rt]\cdot \lt[t\xi'(t\phi(a_{k-1}))\rt]|^{t=1}_{t=0}
\end{align*} 
Finally the first term in \eqref{eq:outofnames} gives energy contribution 
\begin{align*}
    \langle \bh,\bm^k\rangle_N 
    &\simeq
    \mathbb E[h\cdot f(Z\sqrt{\xi'(\lbq)})]
    \\
    &=
    \mathbb E\lt[h\cdot \partial_{x}\Phi_{\gamma_*}\lt(\lbq,h+Z\sqrt{\xi'(\lbq)}\rt)\rt] .
\end{align*} 
Since $\lim_{k\to\infty} a_{k-1}=\xi'(\lbq)$ and $\psi(\xi'(\lbq))=\xi'(\lbq)$ combining concludes the proof.
\end{proof}

\subsection{Phase 2: Incremental AMP}\label{subsec:IAMP}

We now switch to IAMP, which has a more complicated definition. We will begin from the iterates $\bx^{\ul},\bm^{\ul}$ from phase $1$ for a large $\ul\in\mathbb Z^+$. Our implementation follows that of \cite{ams20,AS20} and we relegate several proofs to Section~\ref{ap:iamp}. First define the functions
\[
    u(t,x)=\partial_{xx}\Phi_{\gamma_*}(t,x),\quad v(t,x)=\xi''(t)\gamma_{*}(t)\partial_x\Phi_{\gamma_*}(t,x).
\]

Set $\eps_0=\frac{\lbq}{\phi(a_{\ul-1})}-1$ and 
% $\delta=\xi'(\lbq(1+\eps_0)^2)-\xi'(\lbq)$; 
{\color{black}$\delta=\lbq\big((1+\eps_0)^2)-1\big)$}
;
observe that $\eps_0,\delta\to 0$ as $\ul\to\infty$.\footnote{When $\lbq=0$, $\eps_0$ is not defined. In this case we take $\delta>0$ small and begin IAMP at $\bn^{\ul}=0$ as in \cite{ams20}.} 
Define the sequence $(q_{\ell}^{\delta})_{\ell\geq \ul}$ by 
\[
q_{\ell}^{\delta}=\lbq+({\ell}-\ul)\delta.
\]
Fix $\ubq\in (\lbq,1)$; the value $\ubq$ will be taken close to $1$ after sending $\ul\to\infty$. In particular we will assume $\delta<1-\ubq$ holds and set $\ubl=\min\{\ell\in\mathbb Z^+:q_{\ell}^{\delta}\geq\ubq\}$. Also define 
\[
    \bn^{\ul}\equiv (1+\eps_0)\bm^{\ul}.
\]

Set $\bz^{\ul}=\bw^{\ul}$. So far, we have defined $(\bx^{\ul},\bz^{\ul},\bn^{\ul})$. We turn to inductively defining the triples $(\bx^{\ell},\bz^{\ell},\bn^{\ell})$ for $\ul\leq\ell\leq\ubl$. First, the values $(\bz^{\ell})_{\ell\geq \ul}$ are defined as AMP iterates via

\begin{align}
\label{eq:general_amp}
    \begin{split}
    \bz^{\ell+1} 
    &= 
    \nabla \wt H_N(f_{\ell}(\bz^0,\cdots,\bz^\ell)) - \sum_{j=0}^\ell d_{\ell, j} f_{j-1}(\bz^0,\cdots,\bz^{j-1}),
    \\  
    d_{\ell,j} 
    &= 
    \xi''\lt( \E\lt[f_{\ell}(Z^0,\dots,Z^{\ell})f_{j-1}(Z^0,\dots,Z^{j-1}\rt]\rt) \cdot \E\lt[\frac{\partial f_{\ell}}{\partial z^j}(Z^0,\cdots,Z^\ell)\rt]\, .
    \end{split}
\end{align}  
(The non-linearities $f_{\ell}$ will be specified below). The sequence $(\bx^{\ell+1})_{\ell\geq \ul}$ is defined by 
\begin{align*}
    \bx^{\ell+1}
    &\equiv
    \bx^{\ul}+
    \sum_{j=\ul}^{\ell} 
    v\lt(q_{j}^{\delta}, \bx^{j}\rt) \delta
    + 
    \sum_{j=\ul}^{\ell}
    \lt(\bz^{j+1}-\bz^{j}\rt)
    \\
    &=
    \bx^{\ell}+
    v\lt(q_{\ell}^{\delta}, \bx^{\ell}\rt) \delta
    +
    \lt(\bz^{{\ell}+1}-\bz^{\ell}\rt)
    , 
    \quad 
    \ul \leq {\ell} \leq \ubl-1.
\end{align*} 
As usual, $v(q_j^{\delta},\cdot )$ is applied component-wise so that $v(q_j^{\delta},\bx^j)_i=v(q_j^{\delta},x^j_i)$. Next define the scalar function
\[
    u_{\ell}^{\delta}(x)
    =
    \frac{\delta u(q^{\delta}_{\ell},x)}
    {
        \lt(\xi'(q^{\delta}_{\ell}) - \xi'(q^{\delta}_{\ell-1})\rt) 
        \E\lt[u(q^{\delta}_{\ell}; X^{\delta}_{\ell})^2\rt]
    }
\]
and consider for $\ell\geq \ul$ the recursive definition 
\begin{align}
\label{eq:IAMP}
    \bn^{\ell+1} 
    &\equiv 
    \bn^{\ul}+
    \sum_{j=\ul}^{\ell} 
    u_{j}^{\delta}
    \lt(
        \bx^{j}\rt)\lt(\bz^{j+1}-\bz^{j}
    \rt)
    \\
\nonumber
    &=
    \bn^{\ell}+ 
    u_{\ell}^{\delta}
    \lt(
        \bx^{\ell}\rt)\lt(\bz^{\ell+1}-\bz^{\ell}
    \rt).
\end{align}
We define the non-linearity $f_{\ell}:\mathbb R^{\ell+1}\to\mathbb R$ to recursively satisfy
\[
    f_{\ell}(\bz^0,\dots,\bz^{\ell})= \bn^{\ell},\quad \ell>\ul.
\]

It is not difficult to verify that the equations above form a ``closed loop" uniquely determining the sequence $(\bx^{\ell},\bz^{\ell},\bn^{\ell})_{\ell\geq \ul}$. Since $(x^{\ell}_i,n^{\ell}_i)$ is determined by the sequence $(z^{\ul}_i,\dots,z^{\ell}_i)$ we may define the state evolution limiting random variables $(X_{\ell}^{\delta},N_{\ell}^{\delta},Z_{\ell}^{\delta})_{\ell\geq\ul}.$ We emphasize that the IAMP just defined is part of the same $\lbq$-AMP algorithm as the first phase defined in the previous subsection. However the variable naming has changed so that the main iterates are $\bz^{\ell}$ for $\ell\geq \ul$ rather than $\bw^{\ell}$ for $\ell\leq \ul$. In particular there is no problem in applying state evolution even though the two AMP phases take different forms.

To complete the algorithm, we output the coordinate-wise sign $\bsigma=\sign(\bn^{\ubl})$ where
\[
    \sign(x)=
    \begin{cases}
        ~~1,\quad x\geq 0\\
        -1,\quad x\leq 0.
    \end{cases}
\]

The key to analyzing the AMP algorithm above is an SDE description in the $\delta\to 0$ limit. Define the filtration
\begin{equation}
\label{eq:Felldelta}
    \mathcal F_{\ell}^{\delta}=\sigma\lt((Z^{\delta}_{k},N^{\delta}_{k})_{0\leq k\leq \ell}\rt)
\end{equation}
for the state evolution limiting process.

\begin{lemma}
\label{lem:BMlimit}
The sequences $(Z^{\delta}_{\ul},Z^{\delta}_{\ul+1},\dots)$ and $(N^{\delta}_{\ul},N^{\delta}_{\ul+1},\dots)$ satisfy for each $\ell\geq \ul$:
\begin{align*} 
    \mathbb E[(Z^{\delta}_{\ell+1}-Z^{\delta}_{\ell})Z_{j}^{\delta}]
    &=0,
    \quad \text{for all }\ul+1\leq j\leq \ell
    \\
    \mathbb E[(Z^{\delta}_{\ell+1}-Z^{\delta}_{\ell})^2|\mathcal F_{\ell}^{\delta}]
    &=
    \xi'(q^{\delta}_{\ell+1})-\xi'(q^{\delta}_{\ell})
    \\
    \mathbb E[(Z^{\delta}_{\ell})^2]
    &=
    \xi'(q_{\ell}^{\delta})]
    \\
    \mathbb E[(N^{\delta}_{\ell+1}-N^{\delta}_{\ell})|\mathcal F_{\ell}^{\delta}]
    &=0
    \\
    \mathbb E[(N^{\delta}_{\ell+1}-N^{\delta}_{\ell})^2]
    &=\delta
    \\
    \mathbb E[(N^{\delta}_{\ell})^2]
    &=
    q_{\ell+1}^{\delta}.
\end{align*}
\end{lemma}

 From Lemma~\ref{lem:BMlimit} and the fact that $(Z^{\delta}_{\ul},Z^{\delta}_{\ul+1},\dots)$ form a Gaussian process, it follows that there is a coupling with a standard Brownian motion $(B_t)_{t\in [0,1]}$ such that $\int_0^{q_{\ell}^{\delta}} \sqrt{\xi''(t)}\de B_{t}=Z_{\ell}^{\delta}$ for each $\ell$. Denote by $(\mathcal F_t)_{t\in [0,1]}$ the associated natural filtration. Recall that $X_t$ is defined as the solution to the SDE
 \[
    dX_t=\gamma_*(t)\partial_x\Phi_{\gamma_*}(t,X_t)\de t+\sqrt{\xi''(t)}\de B_t
 \] 
 with initialization $X_0\sim \mathcal L_h$. Recalling Proposition~\ref{prop:Xt}, define processes $(N_t,Z_t)_{t\in [0,1]}$ by
 \begin{align*}
    N_t
    &\equiv
    \partial_x\Phi_{\gamma_*}(t,X_t)\\
    &=
    \partial_x\Phi_{\gamma_*}(\lbq,X_{\lbq})+\int_{\lbq}^t \sqrt{\xi''(s)}u(s,X_s)dB_s,\\ 
    Z_t
    &\equiv
    \int_0^t \sqrt{\xi''(s)}dB_s.
 \end{align*} The next lemma states that these continuous-time processes are the $\delta\to 0$ limit of $(X_{\ell}^{\delta},N_{\ell}^{\delta},Z_{\ell}^{\delta})_{\ell\geq\ul}$. 

\begin{restatable}{lemma}{SDE}
\label{lem:SDE}
Fix $\ubq\in (\lbq,1)$. There exists a coupling between the families of triples $\{(Z^{\delta}_{\ell},X^{\delta}_{\ell},N^{\delta}_{\ell})\}_{\ell\geq 0}$ and $\{(Z_t,X_t,N_t)\}_{t\geq 0}$ such that the following holds for a constant $C>0$. 
{\color{black}
For large enough $\ul$, and every $\ell\geq \ul$ with $q_{\ell}\leq \ubq$,
}
\begin{align*}
    \max _{\ul \leq j \leq \ell} \mathbb{E}\lt[\lt(X_{j}^{\delta}-X_{q_{j}}\rt)^{2}\rt] 
    &\leq 
    C \delta,
    \\
    \max _{\ul \leq j \leq \ell} \mathbb{E}\lt[\lt(N_{j}^{\delta}-N_{q_{j}}\rt)^{2}\rt] 
    &\leq 
    C \delta.
\end{align*}

\end{restatable}

Lemmas~\ref{lem:BMlimit} and~\ref{lem:SDE} are proved in Section~\ref{ap:iamp}.

\subsection{Computing the Final Energy}
\label{subsec:final-energy}

In this subsection we establish Theorem~\ref{thm:main} by showing $\lim_{\ubq\to 1}\lim_{\ul\to\infty}\plim_{N \rightarrow \infty}\frac{H_{N}(\bsigma)}{N}=\Par(\gamma_*).$ First we show that the replacements $\bm^{\ul}\to\bn^{\ul}$ and $\bn^{\ubl}\to\bsigma$ have negligible effect on the asymptotic value attained.

\begin{lemma} \label{lem:noroundingerror}

\begin{align}
\label{eq:approx1}
    \lim_{\ubq\to 1}
    \lim_{\ul\to\infty}
    \plim_{N \rightarrow \infty} 
    \lt|\frac{H_{N}(\bsigma)-H_{N}\lt(\bn^{\ubl}\rt)}{N}\rt|
    &=
    0
    ,
    \\
\label{eq:approx2}
    \lim_{\ul\to\infty}
    \plim_{N \rightarrow \infty}
    \lt| 
    \frac{H_N(\bm^{\ul})-H_N(\bn^{\ul})}{N}
    \rt|
    &=0.
\end{align}
\end{lemma}

\begin{proof}
Proposition~\ref{prop:1lip} implies that $N_t\in [-1,1]$ almost surely, while optimizability of $\gamma_*$ implies that $\mathbb E[(N_t)^2]=t$ for $t\in [\lbq,\ubq]$. It follows that
\begin{align*}
    \lim_{\ubq\to 1}\lim_{\ul\to\infty}\plim_{N \rightarrow \infty} \|\bn^{\ubl}-\sign(\bn^{\ubl})\|_N
    &=
    \lim_{\ubq\to 1} \sqrt{\mathbb E\lt[\lt(N_{\ubq}-\sign(N_{\ubq})\rt)^2\rt]}
    \\
    &=0.
\end{align*} The limit \eqref{eq:approx1} follows from Proposition~\ref{prop:lip} with $k=1$. \eqref{eq:approx2} follows similarly as \[\bn^{\ul}-\bm^{\ul}=\eps_0\bm^{\ul}\] and $\eps_0\to 0$ as $\ul\to\infty$ while $\plim_{N\to\infty}\|\bm^{\ul}\|_N\leq 1$ thanks to Proposition~\ref{prop:1lip}.
\end{proof}

In the next lemma, proved in Section~\ref{ap:iamp}, we compute the energy gain during IAMP.

\begin{lemma}
\label{lem:iampenergy}
\begin{equation}
\label{eq:iampenergy}
    \lim_{\ubq\to 1}\lim_{\ul\to\infty}
    \plim_{N \rightarrow \infty}
    \frac{H_{N}\big(\bn^{\ubl}\big)-H_{N}\big(\bn^{\ul}\big)}{N} 
    = 
    \int_{\lbq}^{1} 
    \xi^{\prime \prime}(t) 
    \mathbb{E}\lt[
    u\lt(t, X_{t}\rt)
    \rt] 
    \mathrm{d} t.
\end{equation}
\end{lemma}

We now put everything together. Recall from Lemma~\ref{lem:RSenergy} that 
\begin{align*}
    \lim_{\ul\to\infty} 
    \plim_{N\to\infty}
    \frac{H_N(\bm^{\ul})}{N}
    =
    \xi'(\lbq)\cdot
    \mathbb E\lt[
    \partial_{xx}\Phi_{\gamma_*}
    \lt(\lbq,X_{\lbq}\rt)
    \rt]
    +
    \mathbb E\lt[
    h\cdot \partial_{x}\Phi_{\gamma_*}
    \lt(\lbq,X_{\lbq}\rt)
    \rt].
\end{align*} 
Proposition~\ref{prop:Xt} implies that the process $\partial_x\Phi_{\gamma_*}(t,X_{t})$ is a martingale, while Lemma~\ref{lem:identity2} states that $\mathbb E[u(t,X_t)]=\mathbb E[\partial_{xx}\Phi_{\gamma_*}(t,X_t)]{\color{black}\geq}\int_t^1 \gamma_*(s)\de s$. Substituting, we find
\[
\lim_{\ul\to\infty}\plim_{N\to\infty} \frac{H_N(\bm^{\ul})}{N}\geq  \xi'(\lbq)\int_{\lbq}^1\gamma_*(s)\de s + \mathbb E[h\cdot \partial_x\Phi_{\gamma_*}(0,h)].
\] 
Using again that $\mathbb E[u(t,X_t)]=\int_t^1 \gamma_*(s)\de s$, the right-hand side of \eqref{eq:iampenergy} is 
 \begin{align*}
    \lim_{\ubq\to 1}\lim_{\ul\to\infty}\plim_{N \rightarrow \infty}\frac{H_{N}\lt(\bn^{\ubl}\rt)-H_{N}\lt(\bn^{\ul}\rt)}{N} 
    &{\color{black}\geq} 
    \int_{\lbq}^1 \xi''(t)\int_t^1 \gamma_*(s)\de s \de t
    \\
    &=
    \int_{\lbq}^1\int_{\lbq}^s \xi''(t)\gamma_*(s)\de t \de s
    \\
    &=\int_{\lbq}^1 (\xi'(s)-\xi'(\lbq))\gamma_*(s)\de s
    \\
    &=
    \int_0^1 \xi'(s)\gamma_*(s)\de s -\xi'(\lbq)\int_{\lbq}^1\gamma_*(s)\de s.
\end{align*}
Combining with Lemma~\ref{lem:noroundingerror} yields 
\begin{align}
\nonumber 
    \lim_{\ubq\to 1}\lim_{\ul\to\infty}\plim_{N \rightarrow \infty}\frac{H_{N}(\bsigma)}{N}
    &{\color{black}\geq}
    \lim_{\ubq\to 1}\lim_{\ul\to\infty}\plim_{N \rightarrow \infty}\frac{1}{N}\cdot \Bigg(H_{N}\lt(\sign\big(\bn^{\ubl}\big)\rt)-H_N\big(\bn^{\ubl}\big)
\\
\nonumber
    &\quad\quad 
    + H_N\big(\bn^{\ubl}\big)-H_N\big(\bn^{\ul}\big)+H_N\big(\bn^{\ul}\big)-H_N\big(\bm^{\ul}\big)+H_N\big(\bm^{\ul}\big) \Bigg) 
    \\
\label{eq:finalvalue}
    &{\color{black}\geq}
    \mathbb E^{h\sim\mathcal L_h}[h\cdot \partial_x\Phi_{\gamma_*}(0,h)]+\int_0^1 \xi'(s)\gamma_*(s)\de s .
\end{align}
Having estimated the limiting energy achieved by our $\lbq$-AMP algorithm, it remains to verify that the value in Equation~\eqref{eq:finalvalue} is $\Par_{\xi,h}(\gamma_*)$. Define 
\[
    \Psi_{\gamma_*}(t,x)=\Phi_{\gamma_*}(t,x)-x\partial_x \Phi_{\gamma_*}(t,x)
\] 
for $(t,x)\in [0,1]\times \mathbb R.$ We also use some identities computed in \cite{ams20}.

\begin{proposition}[{\cite[Lemma 6.13]{ams20}}]
\label{prop:DPsi}
For any $\gamma\in\cuL$, the following identities hold:
\begin{align*}
% \label{eq:DEPhi}
    \frac{\de\phantom{t}}{\de t}\E\lt[\Phi_{\gamma}(t,X_{t})\rt]
    &=
    \frac{1}{2}\, \xi''(t)\gamma(t)\E\lt[\partial_{x}\Phi_{\gamma}(t,X_t)^2\rt]\, 
    \\
% \label{eq:DEPhiX}
    \frac{\de\phantom{t}}{\de t}\E\lt[X_{t}\partial_x\Phi_{\gamma}(t,X_{t})\rt]
    &=
    \xi''(t)\gamma(t)\E\lt[\partial_{x}\Phi_{\gamma}(t,X_t)^2\rt]+\xi''(t) \E\lt[\partial_{xx}\Phi_{\gamma}(t,X_t)\rt]\, . 
\end{align*}
\end{proposition}

\begin{lemma} 
For $h\sim\mathcal L_h$, $\ubq\geq\lbq$, and $X_t$ as in \eqref{eq:parisiSDE},
\begin{align*}%\label{eq:PsiPhi}
    \mathbb E[\Phi_{\gamma_*}(0,h)] &=\mathbb E[h\cdot\partial_x\Phi_{\gamma_*}(0,h)]+\E\lt[\Psi_{\gamma_*}(\ubq,X_{\ubq})\rt]\\
    +&
    \frac{1}{2}\int_0^{\ubq}\xi''(t) \gamma_*(t) \E\lt[\partial_x\Phi_{\gamma_*}(t,X_t)^2\rt]\, \de t
    +\int_0^{\ubq} \xi''(t) \E\lt[\partial_{xx}\Phi_{\gamma_*}(t,X_t)\rt]\, \de t.
\end{align*}
\end{lemma}

\begin{proof}

We write
\begin{align*}
    \mathbb E[\Psi_{\gamma_*}(\ubq,X_{\ubq})-\Psi_{\gamma_*}(0,X_{0})]&=\int_{0}^{\ubq}\frac{\de\phantom{s}}{\de s}\mathbb E\lt[\Psi_{\gamma_*}(s,X_s)\rt]|_{s=t} \de t
    \\
    &=
    \int_{0}^{\ubq}\frac{\de\phantom{s}}{\de s} \mathbb E\lt[\Phi_{\gamma_*}(s,X_s)-X_s\partial_x \Phi_{\gamma_*}(s,X_s)\rt]|_{s=t}\de t
    \\
    &=
    - \frac{1}{2}\int_{0}^{\ubq}\xi''(t) \gamma_*(t) \E\lt[\partial_x\Phi_{\gamma_*}(t,X_t)^2\rt] \de t
    \\
    &~~~
    -\int_{0}^{\ubq} \xi''(t) \E\lt[\partial_{xx}\Phi_{\gamma_*}(t,X_t)\rt]\de t.
\end{align*} 
Rearranging shows:
\begin{align*}
    \mathbb E[\Phi_{\gamma_*}(0,X_{0})]&=\mathbb E[X_{0}\partial_x\Phi_{\gamma_*}(0,X_{0})]+\mathbb E[\Psi_{\gamma_*}(\ubq,X_{\ubq})]
    \\
    &+
    \frac{1}{2}\int_{0}^{\ubq}\xi''(t) \gamma_*(t) \E\lt[\partial_x\Phi_{\gamma_*}(t,X_t)^2\rt] \de t+\int_{0}^{\ubq} \xi''(t) \E\lt[\partial_{xx}\Phi_{\gamma_*}(t,X_t)\rt]\de t.
\end{align*}
As $X_0=h$ this concludes the proof.
\end{proof}

\begin{proof}[Proof of Theorem~\ref{thm:main}]
Note that $\gamma_*(t)>0$ implies $t\geq\lbq$ and hence by optimizability 
\[
    \mathbb E[(\partial_{x}\Phi_{\gamma_*}(t,X_t))^2]=t.
\] 
Meanwhile for any $t\in [0,1]$, 
\[
\mathbb E[\partial_{xx}\Phi_{\gamma_*}(t,X_t)]=\int_t^1 \gamma_*(t)\de t.
\]
Therefore 
\[
    \Phi_{\gamma_*}(0,h) = h\cdot \partial_x\Phi_{\gamma_*}(0,h)+\E\lt[\Psi_{\gamma_*}(\ubq,X_{\ubq})\rt] +\frac{1}{2}\int_0^{\ubq}\xi''(t) \gamma_*(t)t \, \de t +\int_0^{\ubq} \xi''(t) \int_t^1 \gamma_*(s) \de s\, \de t\, .
\]
Recalling \eqref{eq:Par}, we find
\begin{align*}
    \Par(\gamma_*)
    &=
    \mathbb E^h[\Phi_{\gamma_*}(0,h)]-\frac{1}{2}\int_0^1\xi''(t)\gamma_*(t)t\de t 
    \\
    &=
    \mathbb E^h[h\cdot \partial_x\Phi_{\gamma_*}(0,h)]
    +
    \E\lt[\Psi_{\gamma_*}(\ubq,X_{\ubq})\rt] 
    +
    \int_0^{\ubq} \xi''(t) 
    \int_t^1 
    \gamma_*(s) \de s\, \de t\, 
    +
    o_{\ubq\to 1}(1).
\end{align*}
As in \cite[Proof of Theorem 3.2]{ams20} $\lim_{\ubq\to 1}\Psi_{\gamma_*}(\ubq,x)=0$ holds uniformly in $x$. Moreover
\begin{align*}
    \lim_{\ubq\to 1}\int_0^{\ubq} \xi''(t) \int_t^1 \gamma_*(s) \de s\, \de t
    &=
    \int_0^{1} \xi''(t) \int_t^1 \gamma_*(s) \de s\, \de t
    \\
    &=
    \int_0^1 
    \int_0^s \xi''(t) \gamma_*(s) 
    \de t\, \de s
    \\ 
    &=
    \int_0^1 \xi'(s)\gamma_*(s)\de s.
\end{align*}
Combining the above and comparing with \eqref{eq:finalvalue} yields
\begin{align*}
    \Par(\gamma_*)
    &=
    \mathbb E^h\lt[h\cdot \partial_x\Phi_{\gamma_*}(0,h)\rt]
    +\int_0^1 \xi'(s)\gamma_*(s)\de s
    \\
    &{\color{black}\leq}
    \lim_{\ubq\to 1}
    \lim_{\ul\to\infty}
    \plim_{N \rightarrow \infty}\frac{H_{N}(\bsigma)}{N}.
\end{align*}
This completes the proof of Theorem~\ref{thm:main}.
\end{proof}

\section{Constructing Many Approximate Maximizers}\label{sec:branch}

Here we explain the modifications needed for branching IAMP and Theorem~\ref{thm:branch}. The proofs are a slight extension of those for the main algorithm, and in fact we give many proofs for IAMP directly in this more general setting in Section~\ref{ap:iamp}. Let us fix values $Q=(q_1,\dots,q_m)$ with $\lbq\leq  q_1< \dots< q_m<1$ and an index $B\in [m]$. To construct a pair of approximate maximizers with overlap $q_B$ we first construct $\bn^{\ul}$ exactly as in Subsection~\ref{subsec:root}. For each $i<B$, set $\bg^{(q_i,1)}=\bg^{(q_i,2)}=\bz^{-k_{i,1}}=\bz^{-k_{i,2}}\in\mathbb R^N$ for some $k_{i,1}=k_{i,2}\leq K$ as in Subsection~\ref{subsec:init}. For each $B\leq i\leq m$, set $\bg^{(q_i,1)}=\bz^{-k_{i,1}}$ and $\bg^{(q_i,2)}=\bz^{-k_{i,2}}$ where $k_{i,1}\neq k_{i,2}$. Because the vectors $\bg^{(q_i,a)}$ are constructed using AMP, we require some additional conditions. 
{\color{black}
We impose the separation condition
\begin{equation}
\label{eq:kiu-separated}
k_{i,a'}-\ell^{\delta}_{q_i}>k_{j,a}-\ell^{\delta}_{q_j}>0
\end{equation}
for all $i>j$ and $a,a'\in \{1,2\}$. (In particular, it implies that $k_{i,a'}\neq k_{j,a}$ for $i\neq j$.) It is not hard to satisfy \eqref{eq:kiu-separated} by choosing the values $k_{i,a}$ sequentially in increasing order of $i$.
} Finally we insist that $\max_{i,a}(k_{i,a})+\ubl+1<K$, where $h=\bz^{-K}$ was the AMP initialization, which is satisfied by choosing $K$ large at the end.

Having fixed this setup, we proceed by defining $\bm^{k,1}=\bm^{k,2}=\bm^{k}$ for $k\geq 0$ exactly as in the original first phase. 
Next we generate two sequences of IAMP iterates using \eqref{eq:IAMP} except at times corresponding to $q_i\in Q$, altogether generating $\bn^{\ell,a}$ for $\ell> \ul$ and $a\in\{1,2\}$ via:

\begin{equation}
\label{eq:branchIAMP}
\bn^{\ell,a} =\begin{cases}
\bn^{\ell-1,a}+\sqrt{\delta}\bg^{(q_i,a)},\quad \quad\quad \quad\quad\quad\quad\ell=\ell^{\delta}_{q_i}\equiv\ul+\lceil (q_i-\lbq)\delta\rceil+1~~\text{for some } i\in [m]
\\
\bn^{\ell-1,a}+ u_{\ell-1}^{\delta}\lt(\bx^{\ell-1,a}\rt)\lt(\bz^{\ell,a}-\bz^{\ell-1,a}\rt),\quad \text{ else}.
\end{cases}
\end{equation}

Recalling Subsection~\ref{subsec:init}, this is an AMP algorithm of the required form. 
The definitions of $\bx^{\ell,a},\bz^{\ell,a}$ are the same as in e.g. \eqref{eq:general_amp} with superscript $a$ everywhere, {\color{black}though note that now the definition
\[
f_{j-1}(\bz^{-K},\dots, \bz^{0}, \bz^{1,a},\dots,\bz^{j-1,a})=\bn^{j-1,a}
\] 
of $f_{j-1}$ has explicit dependence on the negatively indexed variables through $\bg^{(q_i,a)}$.}
The following result follows immediately from Lemmas~\ref{lem:SDE2}, \ref{lem:iampenergy2} and readily implies Theorem~\ref{thm:branch}.

\begin{lemma}
\label{lem:branch}
For optimizable $\gamma_*$,
\begin{align*}
    \lim_{\ubq\to 1}\lim_{\ul\to \infty}\plim_{N\to\infty}\frac{H_N(\bn^{\ubl,a})}{N}
    &{\color{black}\geq }
    \Par(\gamma_*),\quad a\in\{1,2\}
    \\
    \lim_{\ubq\to 1}\lim_{\ul\to \infty}\plim_{N\to\infty}\frac{\langle\bn^{\ubl,1},\bn^{\ubl,2}\rangle}{N}
    &=
    q_B.
\end{align*}
\end{lemma}

{\color{black}
Next we extend this construction to general finite ultrametric spaces $X$. 
Recall that any finite ultrametric space $X$ with all pairwise distances in the set $\{\sqrt{2(1-q_i)}\}_{i\in [m]}$ can be identified with a rooted tree $\mathcal T$ whose leaves $\partial \mathcal T$ are in bijection with $X$, and so that $d_X(x_i,x_j)=\sqrt{2(1-q_k)}$ is equivalent to leaves $i,j$ having least common ancestor at depth $k$.

Given such $\mathcal T$, we may assign to each non-root vertex $u\in\mathcal T$ a distinct initialization iterate $\bg^{(u)}=\bz^{-k_u}$. We require that
\begin{enumerate}
    \item The $k_u$ are pairwise distinct.
    \item Analogously to \eqref{eq:kiu-separated},
    \[
    k_{u'}-\ell^{\delta}_{q_i}>k_u-\ell^{\delta}_{q_j}
    \] 
    if $q_j=\text{depth}(u)<\text{depth}(u')=q_i$.
    \item $\max_{u}(k_{u})+\ubl+1<K$. 
\end{enumerate}
Again, these conditions are easy to satisfy by choosing $k_u$ in increasing order of $\text{depth}(u)$ with ties broken arbitrarily.

Then for each $x\in X$, we first construct $\bm^{k,x}=\bm^{k}$ for $k\geq 0$ exactly as in the original first phase, which does not depend on $x$. 
Next, denoting by
\[
\text{root}=v_0,v_1,\dots,v_m=x\in\partial\mathcal T=X
\]
the root-to-leaf path for $x$ within $\mathcal T$,
we compute the analog of ~\eqref{eq:branchIAMP}:
\begin{equation}
\label{eq:branchIAMP-general}
\bn^{\ell,x} =
\begin{cases}
\bn^{\ell-1,x}+\sqrt{\delta}\bg^{(k_{v_i})},
\quad \quad\quad \quad\quad\quad\quad
\ell=\ell^{\delta}_{q_i}
% \equiv
% \ul+\lceil (q_i-\lbq)\delta\rceil+1
~~\text{for some } i\in [m]
\\
\bn^{\ell-1,x}+ u_{\ell-1}^{\delta}\lt(\bx^{\ell-1,x}\rt)\lt(\bz^{\ell,x}-\bz^{\ell-1,x}\rt),\quad \text{ else}.
\end{cases}
\end{equation}
Again, $\bx^{\ell,x},\bz^{\ell,x}$ are defined using the same recursions as before.
}

\branch*

\begin{proof}
{\color{black}
It is easy to see that for each distinct $x,y\in X$, the behavior of the pair $\bn^{\ell,x},\bn^{\ell,y}$ in \eqref{eq:branchIAMP-general} is identical to $\bn^{\ell,1},\bn^{\ell,2}$ in \eqref{eq:branchIAMP} (e.g. both pairs have the same joint law with $H_N$). Applying Lemma~\ref{lem:branch} to all such pairs, we find that the iterates $\bn^{\ubl,x}$ satisfy $\frac{H_N(\bn^{\ubl,x})}{N}\geq \Par(\gamma_*)-\eps$ and $\langle\bn^{\ubl,x},\bn^{\ubl,y}\rangle_N\simeq q_j$ if $d_X(x,y)=\sqrt{2(1-q_j)}$. 
The conclusion follows by rounding $\bn^{\ubl,x}\to \bsigma_x\in\bSigma_N$ for each $x\in X$ as in the main algorithm.
}
\end{proof}

We remark that our construction differs from the one proposed in \cite{alaoui2020algorithmic} only because we construct the vectors $\bg^{(u)}$ using AMP rather than taking them to be literally independent Gaussian vectors. While the latter construction almost certainly works as well, the analysis seems to require a more general version of state evolution.

\section{Spherical Models}\label{sec:sphere}

We now consider the case of spherical spin glasses with external field. The law of the Hamiltonian $H_N$ is specified according to the same formula as before depending on $(\xi,\mathcal L)$, however the state space is the sphere $\mathbb S^{N-1}(\sqrt{N})$ instead of the hypercube. The free energy in this case is given by a similar Parisi-type formula, however it turns out to dramatically simplify under no overlap gap so we do not give the general formula. At positive temperature the spherical free energy was computed non-rigorously in \cite{crisanti1992sphericalp} and rigorously in \cite{talagrand2006free,chen2013aizenman}, but we rely on \cite{chen2017parisi} which directly treats the zero-temperature setting.

\begin{remark} Due to rotational invariance, for spherical models all that matters about $\mathcal L_h$ is the squared $L^2$ norm $\mathbb E^{h\sim \mathcal L_h}[h^2].$ In particular unlike the Ising case there is no loss of generality in assuming $h$ is constant. We continue to work with coordinates $h_i$ sampled i.i.d. from $\mathcal L_h$ and implicitly use this observation when interpreting the results of \cite{chen2017parisi}.

\end{remark}

Our treatment of spherical models is less detailed and we simply show how to obtain the energy value in Theorem~\ref{thm:sphereGS} which is the ground state in models with $\supp(\gamma_*)=[q,1)$. In the case that $\E[h^2]+\xi'(1)<\xi''(1)$, we let $\lbq_{\ss}\in [0,1]$ be the unique solution to
\[
  \lbq_{\ss}\xi''(\lbq_{\ss})=\E[h^2]+\xi'(\lbq_{\ss}).
\] 
When $\E[h^2]+\xi'(1)\geq\xi''(1)$, we simply set $\lbq_{\ss}=1$. 

Note that when $h=0$ almost surely it follows that $\lbq_{\ss}=0$, which is the setting of \cite{subag2018following}. Generate initial iterates 
\[
  (\bw^{-K}_{\ss},\dots,\bw^0_{\ss})
\]
as in Subsection~\ref{subsec:init}. For non-zero $h$ we take $c=\sqrt{\E[h^2]+\xi'(\lbq_{\ss})}$ so that 
\[
  \|\bw^0_{\ss}\|_N\simeq \sqrt{\E[h^2]+\xi'(\lbq_{\ss})}.
\]
Generate further iterates via the following AMP iteration.
\begin{align}
\label{eq:RSsphere}
  \bw_{\ss}^{k+1}
  &=
  \nabla \wt H_N(\bm_{\ss}^k)
  -\bm_{\ss}^{k-1}\xi''\lt(\langle \bm^k_{\ss},\bm^{k-1}_{\ss}\rangle\rt)
  \sqrt{\frac{\lbq_{\ss}}{\E[h^2]+\xi'(\lbq_{\ss})}}
  \\
\nonumber
  \bx^k_{\ss}
  &=
  \bw^k_{\ss}+\bh
  \\
\nonumber
  \bm^k_{\ss}
  &=
  \bx^k_{\ss}\cdot \sqrt{\frac{\lbq_{\ss}}{\E[h^2]+\xi'(\lbq_{\ss})}}.
\end{align}

The next lemma is the spherical analog of Lemmas~\ref{lem:RSconverge},~\ref{lem:RSlimit},~\ref{lem:RSenergy} - the proof is similar to the Ising case and is given in the next subsection.

\begin{restatable}{lemma}{RSsphere}\label{lem:RSsphere}

Using the AMP of \eqref{eq:RSsphere}, the asymptotic overlaps and energies satisfy
\begin{align}
\nonumber
    \lim_{k,\ell\to\infty}\plim_{N\to\infty} \frac{\langle \bw^k_{\ss},\bw^{\ell}_{\ss}\rangle}{N} 
    &= 
    \xi'(\lbq_{\ss}),
    \\
\nonumber
    \lim_{k,\ell\to\infty}\plim_{N\to\infty} \frac{\langle \bx^k_{\ss},\bx^{\ell}_{\ss}\rangle}{N} &= 
    \E[h^2]+\xi'(\lbq_{\ss}),
    \\
\nonumber
    \lim_{k,\ell\to\infty}\plim_{N\to\infty} \frac{\langle \bm^k_{\ss},\bm^{\ell}_{\ss}\rangle}{N} &= \lbq_{\ss},
    \\
\label{eq:sphereenergy1}
    \lim_{k,\ell\to\infty}\plim_{N\to\infty} \frac{H_N(\bm^k_{\ss})}{N}&= \sqrt{\lbq_{\ss}(\E[h^2]+\xi'(\lbq_{\ss}))}.
\end{align}

\end{restatable}

\begin{proof}[Proof of Theorem~\ref{thm:sphereopt}]

The latter two parts of Lemma~\ref{lem:RSsphere} directly imply Theorem~\ref{thm:sphereopt} in the case that $\E[h^2]+\xi'(1)\geq\xi''(1)$ (recall $\lbq_{\ss}=1$ in this case). Indeed, it suffices to take 
\begin{equation}\label{eq:sphere-normalize}
  \bsigma_{\ss}=\frac{\bm^{\ubl}_{\ss}}{\|\bm^{\ubl}_{\ss}\|_N}\in \mathbb S^{N-1}(\sqrt{N})
\end{equation}
for a sufficiently large constant $\ubl$.
When $\E[h^2]+\xi'(1)<\xi''(1)$, we can conclude by mimicking the IAMP phase using the simple non-linearities $u(t,x)=u(t)=\xi''(t)^{-1/2}$ and $v(t,x)=0$ - see also \cite[Remark 2.2]{ams20}. Lemma~\ref{lem:iampenergy} then shows the energy gain from IAMP is 
\[
  \int_{\lbq}^1\xi''(t)u(t)\de t=\int_{\lbq}^1\xi''(t)^{1/2}.
\]
As in the Ising case, we may start IAMP from $\bm=\bm^k$ for a large constant $k$. Combining with~\eqref{eq:sphereenergy1} and defining $\bsigma_{\ss}$ via \eqref{eq:sphere-normalize} with $\bm^{\ubl}$ an IAMP iterate, we obtain
\[
  \plim_{N\to\infty}\frac{H_N(\bsigma_{\ss})}{N}\geq\lbq_{\ss}\xi''(\lbq_{\ss})^{1/2}+\int_{\lbq_{\ss}}^1 \xi''(\lbq_{\ss})^{1/2}\de q.
\]
Alternatively to IAMP, in the spherical setting it is possible to use the approach of \cite{subag2018following}. Indeed \cite[Theorem 4]{subag2018following} immediately extends to an algorithm taking in an arbitrary point $\bm$ with $\|\bm\|_N\leq 1$ and outputting a point $\bm_*\in \mathbb S^{N-1}(\sqrt{N})$ (which may depend on $H_N$) satisfying 
\[
  \frac{H_N(\bm_*)-H_N(\bm)}{N} \geq \int_{\|\bm\|_N^2}^1 \sqrt{\xi''(q)}\de q-\eps
\]
with probability $1-o_{N\to\infty}(1)$ for any desired $\eps>0$. Either approach completes the proof of Theorem~\ref{thm:sphereopt}.
\end{proof}

% \begin{theorem}

% Suppose $\xi''(s)^{-1/2}$ is concave on $[0,1]$. Then for every $\eps>0$ there exists an efficient AMP algorithm (in the sense of Theorem~\ref{thm:amsmain}) which outputs $\bsigma\in \mathbb S^{N-1}(\sqrt{N})$ such that

% \begin{align}
% \frac{H_N(\bsigma)}{N} \ge \Par(\gamma_*) - \eps\, ,
% \end{align}
% %
% with probability converging to one as $N\to\infty$.

% \end{theorem}

\subsection{Proof of Lemma~\ref{lem:RSsphere}}\label{ap:sphere}

For $t\in [0,\xi'(\lbq_{\ss})]$, take $h\sim \mathcal L_h$ and $(Z,Z',Z'')\sim\normal(0,I_3)$ and define the function 
\begin{align*}
  \phi_{\ss}(t)
  &=\frac{\lbq_{\ss}}{\E[h^2]+\xi'(\lbq_{\ss})}\cdot\mathbb E^{h,Z,Z',Z''}\lt[\lt(h+Z\sqrt{t}+Z'\sqrt{\xi'(\lbq_{\ss})-t}\rt)\lt(h+Z\sqrt{t}+Z''\sqrt{\xi'(\lbq_{\ss})-t}\rt)\rt]\\
  &=\frac{\lbq_{\ss}(\E[h^2]+t)}{\E[h^2]+\xi'(\lbq_{\ss})}.
\end{align*}
so that $\phi_{\ss}(\xi'(\lbq_{\ss}))=\lbq_{\ss}$. Define $\psi_{\ss}(t)=\xi'(\phi_{\ss}(t))$.

\begin{lemma}\label{lem:spherepsi}

$\psi_{\ss}$ is strictly increasing and convex on $[0,\xi'(\lbq_{\ss})]$ and 
\begin{align}
\label{eq:spherepsi1}
    \psi_{\ss}(\xi'(\lbq_{\ss}))&=\xi'(\lbq_{\ss}),
    \\
\label{eq:spherepsi2}
    \psi'_{\ss}(\xi'(\lbq_{\ss}))&=1,
    \\
\label{eq:spherepsi3}
    \psi_{\ss}(t)
    &>t,
    \quad\quad \forall t<\xi'(\lbq_{\ss}).
 \end{align} 

\end{lemma}

\begin{proof}

Since $\xi'$ is strictly increasing and convex and $\phi_{\ss}$ is affine and increasing, it follows that $\psi_{\ss}$ is strictly increasing and convex. \eqref{eq:spherepsi1} is equivalent to the equation $\lbq_{\ss}\xi''(\lbq_{\ss})=\E[h^2]+\xi'(\lbq_{\ss})$ defining $\lbq_{\ss}.$ To show \eqref{eq:spherepsi2} we use the chain rule to write
\[
  \psi'_{\ss}(\xi'(\lbq_{\ss})) = \xi''(\phi_{\ss}(\xi'(\lbq_{\ss})))\cdot \phi'_{\ss}(\xi'(\lbq_{\ss}))=\xi''(\lbq_{\ss})\cdot (\xi''(\lbq_{\ss}))^{-1}=1.
\] 
Equations~\eqref{eq:spherepsi1} and \eqref{eq:spherepsi2} and the convexity of $\psi_{\ss}$ just shown imply \eqref{eq:spherepsi3}
\end{proof}

Let $h,W^{-1}_{\ss},(W^j_{\ss},X^j_{\ss},M^j_{\ss})_{j=0}^{k}$ be the state evolution limit of the coordinates of 
\[
    \big(\bh,\bw^{-1}_{\ss},\bw^0_{\ss},\bx^{0}_{\ss},\bm^{0}_{\ss},\dots,\bw^k_{\ss},\bx^k_{\ss},\bm^k_{\ss}\big)
\] 
as $N\to\infty$. Define the sequence $(b_0,b_1,\dots)$ recursively by $b_0=0$ and $b_{k+1}=\psi_{\ss}(b_k)$. 

\begin{lemma}
\label{lem:sphereconverge}
For all non-negative integers $0\leq j<k$ the following equalities hold:
\begin{align*}
    \mathbb E[(W_{\ss}^j)^2]
    &=
    \xi'(\lbq_{\ss})
    \\
    \mathbb E[W_{\ss}^jW_{\ss}^k]
    &=
    b_j
    \\
    \mathbb E[(M_{\ss}^j)^2]
    &=
    \lbq_{\ss}
    \\
    \mathbb E[M_{\ss}^jM_{\ss}^k]
    &=
    \phi_{\ss}(b_j).
\end{align*}

\end{lemma}

\begin{proof}

Follows from state evolution and induction exactly as in Lemma~\ref{lem:RSconverge}.
\end{proof}

\begin{lemma}\label{lem:spherelimit}

\begin{align*}
    \lim_{k\to\infty}b_k
    &=
    \xi'(\lbq_{\ss}),
    \\
    \lim_{k\to\infty}\phi_{\ss}(b_k)
    &=
    \lbq_{\ss}.
\end{align*}
\end{lemma}

\begin{proof}

As in the proof of Lemma~\ref{lem:RSlimit}, the sequence $b_1,b_2,\dots,$ must converge up to a limit, and this limit must be a fixed point for $\psi_{\ss}$, implying the first claim. The second claim follows by continuity of $\phi_{\ss}$.
\end{proof}

\begin{lemma}
\label{lem:sphereenergy}
\[
  \lim_{k\to\infty} \plim_{N\to\infty}\frac{H_N(\bm_{\ss}^k)}{N}
  = \sqrt{\lbq_{\ss}(\E[h^2]+\xi'(\lbq_{\ss}))}.
\]
\end{lemma}

\begin{proof}

We use again the identity
\[
  \frac{H_N(\bm_{\ss}^k)}{N}=\big\langle \bh,\bm_{\ss}^k\rangle_N+\int_0^1 \langle \bm_{\ss}^k,\nabla \wt H_N(t\bm_{\ss}^k)\big\rangle_N \de t
\]
and interchange the limit in probability with the integral. To compute the main term $\plim_{N\to\infty}\langle \bm_{\ss}^k,\nabla \wt H_N((t\bm_{\ss}^k)\rangle$ we introduce an auxiliary AMP step 
\[
    \by_{\ss}^{k+1}=\nabla \wt H_N(t\bm_{\ss}^k)-t\bm_{\ss}^{k-1} \xi''(t\langle \bm_{\ss}^k,\bm_{\ss}^{k-1}\rangle )\sqrt{\frac{\lbq_{\ss}}{\E[h^2]+\xi'(\lbq_{\ss})}}\,. 
\] 
Rearranging yields
\begin{align*}
  \langle \bm_{\ss}^k,\nabla \wt H_N(t\bm_{\ss}^k)\rangle_N 
  &= \langle \bm_{\ss}^k,\by_{\ss}^{k+1}\rangle_N 
  +t\langle \bm_{\ss}^k,\bm_{\ss}^{k-1} \rangle_N 
  \xi''(t\langle \bm_{\ss}^k,\bm_{\ss}^{k-1}\rangle_N )\sqrt{\frac{\lbq_{\ss}}{\E[h^2]+\xi'(\lbq_{\ss})}}  \\
  &\simeq   \langle \bm_{\ss}^k,\by_{\ss}^{k+1}\rangle_N +tb_{k-1} \xi''(t\phi(b_{k-1}) )\sqrt{\frac{\lbq_{\ss}}{\E[h^2]+\xi'(\lbq_{\ss})}}.
\end{align*}

For the first term, Gaussian integration by parts with 
\[
  g(x)=(x+h)\cdot\sqrt{\frac{\lbq_{\ss}}{\E[h^2]+\xi'(\lbq_{\ss})}}
\]
yields
\[
  \mathbb E[g(X_{\ss}^k)Y^{k+1}]
  =\mathbb E[g'(X_{\ss}^k)]\cdot\mathbb E[X_{\ss}^kY^{k+1}_{\ss}]
  = \xi'(t\phi_{\ss}(b_{k-1}))
  \sqrt{\frac{\lbq_{\ss}}{\E[h^2]+\xi'(\lbq_{\ss})}}.
\]
Integrating with respect to $t$, we find
\begin{align*} 
  \int_0^1 \langle \bm_{\ss}^k,\nabla \wt H_N(t\bm_{\ss}^k)\rangle_N \de t&\simeq \mathbb E\lt[g'\lt(Z\sqrt{\xi'(\lbq_{\ss})}\rt)\rt]\cdot \int_0^1 \xi'(t\phi_{\ss}(b_{k-1}))+t \phi_{\ss}(b_{k-1})\xi''(t\phi_{\ss}(b_{k-1}))\\
  &=\lt[t\xi'(t\phi_{\ss}(b_{k-1}))\rt]|^{t=1}_{t=0}
  \cdot \sqrt{\frac{\lbq_{\ss}}{\E[h^2]+\xi'(\lbq_{\ss})}}\\
  &=\psi_{\ss}(b_{k-1})\sqrt{\frac{\lbq_{\ss}}{\E[h^2]+\xi'(\lbq_{\ss})}}.
\end{align*}
Finally the first term gives energy contribution
\begin{align*}
  h\langle \bm_{\ss}^k\rangle_N 
  &\simeq \mathbb E\lt[h\lt(h+Z\sqrt{\xi'(b_{k-1})}\rt)\rt]\sqrt{\frac{\lbq_{\ss}}{\E[h^2]+\xi'(\lbq_{\ss})}}\\
  &=\E[h^2]\sqrt{\frac{\lbq_{\ss}}{\E[h^2]+\xi'(\lbq_{\ss})}}.
\end{align*}
Since $\lim_{k\to\infty} b_{k-1}=\xi'(\lbq_{\ss})$ and $\psi_{\ss}(\xi'(\lbq_{\ss}))=\xi'(\lbq_{\ss})$ we conclude
\begin{align*}
  \lim_{k\to\infty} \plim_{N\to\infty}\frac{H_N(\bm_{\ss}^k)}{N}&= \sqrt{\lbq_{\ss}(\E[h^2]+\xi'(\lbq_{\ss}))}.
\end{align*}
\end{proof}

\begin{proof}[Proof of Lemma~\ref{lem:RSsphere}]
The result follows from the preceding lemmas.
\end{proof}

\subsection{Proof of Theorem~\ref{thm:sphereGS}}

It follows from our algorithm that $GS_{\ss}(\xi,\mathcal L_h)\geq \lbq_{\ss}\xi''(\lbq_{\ss})^{1/2}+\int_{\lbq_{\ss}}^1 \xi''(\lbq_{\ss})^{1/2}\de q$. We now characterize the models in which equality holds, which coincide with those exhibiting no overlap gap. Moreover we give an alternate proof of the lower bound for $GS(\xi,\mathcal L_h)_{\ss}$ which shows that equality holds exactly in no overlap gap models. We thank Wei-Kuo Chen for communicating the latter proof.

\sphereGS*

% \begin{lemma}

% The spherical spin glass $(\xi,h)$ is full RSB with support $[\lbq_{\ss},1]$ at zero temperature if and only if $h^2+\xi'(\lbq_{\ss})=\lbq_{\ss}\xi''(\lbq_{\ss})$ and $\xi''(t)^{-1/2}$ is concave on $t\in [\lbq_{\ss},1].$

% \end{lemma}

\begin{proof}

We use the results and notation of \cite{chen2017parisi}. If $\xi''(q)^{-1/2}$ is concave on $[\lbq_{\ss},1]$ then the proof of Proposition 2 in \cite{chen2017parisi} applies verbatim to show that the support of $\alpha$ is $[\lbq_{\ss},1]$. In fact it explicitly shows $\alpha(s)=\frac{\xi'''(s)}{2\xi''(s)^{3/2}}$ for $s\in [\lbq_{\ss},1]$).

In the other direction, we show that if no overlap gap holds and $\E[h^2]+\xi'(1)<\xi''(1)$, then $\xi''(q)^{-1/2}$ is concave on $[\lbq_{\ss},1]$. we use the statement and notation of \cite[Theorem 2]{chen2017parisi}. Assume $\alpha$ is supported on the interval $[\lbq_{\ss},1]$. The last condition in \cite[Theorem 2]{chen2017parisi} states that $g(u)=\int_u^1 \bar{g}(s)\de s=0$ for all $u\in [\lbq_{\ss},1]$, and therefore $\bar{g}(s)=0$ for $s\in [\lbq_{\ss},1]$, where

\[\bar{g}(s)\equiv\xi^{\prime}(s)+h^{2}-\int_{0}^{s} \frac{d q}{\lt(L-\int_{0}^{q} \alpha(r) d r\rt)^{2}}.\]

Setting $s=\lbq_{\ss}$ yields 
$\E[h^2]+\xi'(\lbq_{\ss})=\lbq_{\ss} L^{-2}$, i.e. 
$L=\sqrt{\frac{\lbq_{\ss}}{\E[h^2]+\xi'(\lbq_{\ss})}}$. Differentiating, all $s\geq \lbq_{\ss}$ satisfy

\begin{equation}\label{eq:0}\xi''(s)=\frac{1}{(L-\int_{\lbq_{\ss}}^{s}\alpha(r)dr)^2}.\end{equation}

Taking $s=\lbq_{\ss}$ in \eqref{eq:0} shows $L=\xi''(\lbq_{\ss})^{-1/2}$, hence $\E[h^2]+\xi'(\lbq_{\ss})=\lbq_{\ss}\xi''(\lbq_{\ss})$. Rearranging~\eqref{eq:0} yields 
\[
    L-\int_{\lbq_{\ss}}^{s}\alpha(r)dr=\xi''(s)^{-1/2},\quad\quad s\geq \lbq_{\ss}.
\]
As $\alpha$ must be non-decreasing based on \cite[Equation (9)]{chen2017parisi} it follows that $\xi''(s)^{-1/2}$ is concave on $s\in [\lbq_{\ss},1]$. This completes the proof of the first equivalence. We turn to the value of $GS_{\ss}(\xi,\mathcal L_h)$, first computing
\begin{align*}
    \E[h^2]+\xi'(1)
    &=
    \lbq \xi''(\lbq_{\ss})+\int_{\lbq_{\ss}}^1 \xi''(\lbq_{\ss})\de q
    \\
    &= 
    \int_{0}^{\lbq_{\ss}} \xi''(\lbq_{\ss})\de q+\int_{\lbq_{\ss}}^1 \xi''(\lbq_{\ss})\de q .
\end{align*}
Letting $L>\int_0^1\alpha(s)\de s$ and let $a(q)=\int_0^q \alpha(s)\de s$, we find
\begin{align*}
    2\mathcal{Q}(L,\alpha)&=(\E[h^2]+\xi'(1))L-\int_0^1\xi''(q)a(q)\de q+\int_0^1\frac{\de q}{L-a(q)}\\
    &=\int_0^{\lbq_{\ss}}\bigl(\xi''(\lbq_{\ss})L-\xi''(q)a(q)\bigr)\de q+\int_0^{\lbq_{\ss}}\frac{\de q}{L-a(q)}\\
    &+\int_{\lbq_{\ss}}^1\Bigl(\xi''(q)\bigl(L-a(q)\bigr)+\frac{1}{L-a(q)}\Bigr)\de q.
\end{align*}
Since $\xi''$ is increasing, AM-GM shows the second-to-last line is at most 
\begin{align}
\label{eq:1}
    \int_0^{\lbq_{\ss}}\xi''(\lbq)\bigl(L-a(q)\bigr)\de q+\int_0^{\lbq_{\ss}}\frac{\de q}{L-a(q)}\geq 2 \int_0^{\lbq_{\ss}}\sqrt{\xi''(\lbq_{\ss})}\de q
    =
    2\lbq\sqrt{\xi''(\lbq_{\ss})},
\end{align}
and similarly 
\begin{align}
\label{eq:2}
    \int_{\lbq_{\ss}}^1\Bigl(\xi''(q)\bigl(L-a(q)\bigr)+\frac{1}{L-a(q)}\Bigr)\de q
    \geq 
    2\int_{\lbq_{\ss}}^1\sqrt{\xi''(q)}\de q.
\end{align}
Combining, we conclude the lower bound on $GS_{\ss}(\xi,\mathcal L_h)$. Moreover for equality to hold in \eqref{eq:1} and \eqref{eq:2} we must have 
\[
    \xi''(q)^{-1/2}
    =
    \begin{cases} 
    L-a(q),\,\,\forall q\in [0,\lbq_{\ss}],
    \\
    L-a(q),\,\,\forall q\in [\lbq_{\ss},1).
    \end{cases}
\]
The first equality forces $\alpha(s)=0$ on $[0,\lbq_{\ss})$ and $L=\xi''(q_0)^{-1/2}$, while the second equality implies $\alpha(q)=-\frac{\de}{\de s}\xi''(q)^{-1/2}$ for all $q\in [\lbq_{\ss},1]$. Taken together this means that equality in the $GS_{\ss}$ lower bound implies no overlap gap, completing the proof.
\end{proof}

\section{Impossibility of Approximate Maximization Under an Overlap Gap}\label{sec:OGP}

Here we explain the modifications of \cite{gamarnik2019overlap,GJW20} needed to establish Proposition~\ref{prop:OGP}. Throughout this section we assume that $\xi(t)=\sum_{p\in \{2,4,\dots,2P\}}c_p^2 t^p $ is an even polynomial and that the external field $(h,h,\dots,h)$ is constant and deterministic. We take $\lbq=\inf(\supp(\gamma_*^{\cuU}))$ and let $H_{N,0},H_{N,1}$ be i.i.d. copies of $H_N$ and for $t\in [0,1]$ let $H_{N,t}\equiv\sqrt{t}H_{N,1}+\sqrt{1-t}H_{N,0}$. Moreover we define $\mathcal H_N\equiv\{H_{N,t}:t\in [0,1]\}$.

\begin{definition} 
The model $(\xi,h)$ satisfies the path overlap gap property with parameters $\eps>0$ and $0<\nu_1<\nu_2<1$ if the following holds with probability at least $1-\frac{e^{-KN}}{K}$ for some $K=K(\xi,h)>0$. For every pair $H_s,H_t\in\mathcal H_N$ and every $\bsigma_s,\bsigma_t$ satisfying
\[
    \frac{H_r(\bsigma_r)}{N}
    \geq 
    GS_{\xi,h}-\eps,
    \quad\quad r\in\{s,t\}
\]
it also holds that
\[
    |\langle \bsigma_s,\bsigma_t\rangle_N|
    \in 
    [0,\nu_1]\cup[\nu_2,1].
\]
\end{definition}

\begin{definition} 
The pair $(H,H')$ of Hamiltonians are $\nu$-separated above $\mu$ if for any $\bx,\by\in \bSigma_N$ with $H(\bx)\geq\mu,H'(\by)\geq \mu$, it holds that $|\langle \bx,\by\rangle_N|\leq \nu$.

\end{definition}

\begin{lemma}
\label{lem:disorderchaos}
For any $t\in (0,1)$, there exists $q_t=q_t(t,\xi,h)\in [0,\lbq]$ such that the following holds. For any $\eps>0$ there is {\color{black}$\delta(\eps)>0$ and} $K(t,\xi,h,\eps)$ such that with probability at least $1-\frac{e^{-KN}}{K}$, every pair $\bsigma_0,\bsigma_t\in\bSigma_N$ satisfying
\[
    \frac{H_{N,s}(\bsigma_s)}{N}
    \geq 
    GS_{\xi,h}-{\color{black}\delta},\quad s\in\{0,t\}
\]
also satisfies $|\langle \bsigma_0,\bsigma_t\rangle_N|\in [q_t-\eps,q_t+\eps].$
\end{lemma}

\begin{proof}

This follows from \cite[Proof of Theorem 2]{chen2018energy}, in particular inequality (48) therein
{\color{black}
implies that for some $\delta(\eps)>0$:
\[
    \frac{1}{N}
    \max_{\substack{\bsig_0,\bsig_t \in\Sigma_N
    \\
    \la \bsig_0,\bsig_t \ra_N \notin [q_t-\eps,q_t+\eps]
    }}
    H_{N,0}(\bsig_0)+H_{N,t}(\bsig_t)
    \leq 
    \lt(\frac{1}{N}
    \max_{\bsig_0,\bsig_t \in\Sigma_N}
    H_{N,0}(\bsig_0)+H_{N,t}(\bsig_t)
    \rt)
    -
    \delta
    .
    \qedhere
\]
% Replacing $\eps$ with $\min(\eps,\delta)$ implies the desired statement.
}
\end{proof}

\begin{lemma}\label{lem:singleogp}

If $(\xi,h)$ is not optimizable, then there exist $\eps(\xi,h)>0$ and $\lbq<a<b<1$ and $K(\xi,h)$ such that with probability $1-\frac{e^{-KN}}{K}$ the following holds. Every pair $\bsigma_0,\bsigma_1\in\bSigma_N$ satisfying \[\frac{H_{N}(\bsigma_i)}{N}\geq GS_{\xi,h}-\eps,\quad i\in\{0,1\}\] also satisfies $|\langle\bsigma_0,\bsigma_1\rangle_N|\notin [a,b].$

\end{lemma}

\begin{proof}

The proof is identical to that of \cite[Theorem 3]{chen2019suboptimality} using the $2$-dimensional Guerra-Talagrand bound. Indeed \cite[Lemma 5.4]{chen2019suboptimality} exactly establishes that even pure $p$-spin models are not optimizable, i.e. 
\begin{equation}
\label{eq:violate}
\bbE[(\partial_x \Phi_{\gamma_*}(t,X_t))^2]<t
\end{equation}
holds for some $t\in [0,1)$ where $\gamma_*=\gamma_*^{\cuU}$. The remainder of the proof {\color{black}(just below \cite[Lemma 5.4]{chen2019suboptimality})} is fully general and we give an outline below. 
{\color{black}
The point is that by \eqref{eq:violate} must hold in some non-empty open subset of $(0,1)$, thus in a non-empty interval $[a,b]$. For each $t\in [a,b]$, one considers the Hamiltonian $\frac{H_N(\bsig)+H_N(\bsig')}{2}$ on two-replica configurations $(\bsig,\bsig')$ with overlap constraint $\langle \bsig,\bsig'\rangle_N=t$. 
The free energy of this constrained system can be upper-bounded using an interpolation argument; the relevant Parisi order parameter $\wt\gamma$ must increase, except that it may decrease by a factor of at most $2$ at $t$, i.e. it only needs to satisfy 
\begin{equation}
\label{eq:decreasing-two}
\lim_{s\uparrow t}\wt\gamma(s)\leq 2\lim_{s\downarrow t}\wt\gamma(s). 
\end{equation}
Taking $\wt\gamma=\gamma_*^{\cuU}$ recovers the single-replica value. However when \eqref{eq:violate} holds, $\wt\gamma=\gamma_*^{\cuU}$ is no longer locally optimal since $\wt\gamma$ lives in a larger function space due to the relaxation \eqref{eq:decreasing-two}. Hence the constrained two-replica ground state energy is strictly smaller. This argument can be applied for all $O(N)$ values $t\in [a,b]\cap \mathbb Z/N$, yielding the result.
% one directly finds that increasing $\gamma_*$ by a small factor $(1+\eps)$ on $[t,1]$ decreases the value of $\Par(\gamma)$. 
% The resulting perturbation $\wt\gamma$ may no longer be increasing, but it is increasing everywhere except at $t$, and satisfies $\lim_{s\uparrow t}\wt\gamma(s)\leq 2\lim_{s\downarrow t}\wt\gamma(s)$. This suffices to obtain the $2$-replica Guerra-Talagrand upper bound: the free energy of pairs $\bsig_0,\bsig_t$ with overlap $\la \bsig_0,\bsig_t\ra = t$ is at most $2\Par(\wt\gamma)<2\Par(\gamma_*)$ which implies the claimed overlap gap property.
}
\end{proof}

% We stop to remark that in \cite{huang2021tight}, a key insight is that by using a branching ultrametric tree, the analogous relaxation to \eqref{eq:decreasing-two} holds at an finite but arbitrarily dense set of values $t$, so that ``in the limit'' $\wt\gamma$ can approximate any $\gamma\in \cuL$.

\begin{lemma}\label{lem:pathogp}

If $(\xi,h)$ is not optimizable, then there exists $\eps(\xi,h)>0$ and $0\leq \nu_1<\nu_2\leq 1$ and $K>0$ such that with probability at least $1-\frac{e^{-KN}}{N}$:

\begin{enumerate}
    \item The model $(\xi,h)$ satisfies the path overlap gap property with parameters $(\eps,\nu_1,\nu_2)$.
    \item $H_{N,0},H_{N,1}$ are $\nu_1$ separated above $GS_{\xi,h}-\eps$.
\end{enumerate}

\end{lemma}

\begin{proof}

The proof is identical to \cite[Theorem 3.4]{gamarnik2019overlap}. In short, one discretizes $\mathcal H_N$ into $\{H_{N,k\delta}: 0\leq k\leq \delta^{-1}\}$ for some small $\delta>0$ using Proposition~\ref{prop:lip} and then applies Lemma~\ref{lem:disorderchaos} to control the values $H_{N,s}(\bsigma_s),H_{N,t}(\bsigma_t)$ for $s\neq t$ and Lemma~\ref{lem:singleogp} to control the cases that $s=t$.
{\color{black}
Indeed in the proof of \cite[Theorem 3.4]{gamarnik2019overlap}, 
the former is accomplished using \cite[Theorem 3]{chen2019suboptimality} while
the latter is accomplished using \cite[Theorem 2]{chen2018energy}. 
The preceding lemmas exactly generalize the relevant statements to non-optimizable models.
}
\end{proof}

\begin{proof}[Proof of Proposition~\ref{prop:OGP}]

Given Lemma~\ref{lem:pathogp}, the proof is identical to that of \cite[Theorem 3.3]{gamarnik2019overlap}. 
{\color{black}
Indeed, that proof does not depend on $\xi$. 
The main input is \cite[Conjecture 3.2]{gamarnik2019overlap}.
This is shown to be implied by the combination of \cite[Theorem 3.4]{gamarnik2019overlap} and \cite[Conjecture 3.6]{gamarnik2019overlap}. Lemma~\ref{lem:pathogp} above suitably extends the former, while the latter (for general $(\xi,h)$) is the main result of our subsequent work \cite{sellke2020approximate}.
The proof of \cite[Theorem 3.3]{gamarnik2019overlap} also uses \cite[Theorem 4.2 and 6.1]{gamarnik2019overlap}; these follow from general concentration of measure results on Gaussian space and easily extend to general $\xi$. 
}
\end{proof}

We remark that Lemma~\ref{lem:pathogp} is also the only property of pure even $p$-spin models used in \cite[Theorem 2.4]{GJW20} to rule out approximate maximization (in a slightly weaker sense) by constant degree polynomials. Therefore their result also applies under the more general conditions of Proposition~\ref{prop:OGP}.

\section{Proof of Lemmas~\ref{lem:identity},~\ref{lem:optimizable} and \ref{lem:identity2}}\label{sec:identity}

We first recall several existing results.

\begin{corollary}[{\cite[Corollary 6.6]{ams20}}]
\label{cor:ED2}
For any $\gamma\in\cuL$ and any $t\in [0,1)$, 
\begin{align*}
\E\lt[\partial_x\Phi_{\gamma}(t,X_{t})^2\rt] = \int_{0}^{t} \xi''(s)\, \E\lt[\big(\partial_{xx}\Phi_{\gamma}(s,X_s)\big)^2\rt]\, \de s\, .
\end{align*}
\end{corollary}

\begin{lemma}[{\cite[Corollary 6.6 and Lemma 6.7]{ams20}}]
\label{lem:continuous}
For any $\gamma\in \cuL$, the values 
\[\E\lt[\partial_x\Phi_{\gamma}(t,X_{t})^2\rt],\quad \quad \E\lt[\partial_{xx}\Phi_{\gamma}(t,X_{t})^2\rt]\]
are continuous functions of $t\in [0,1)$.

\end{lemma}

\begin{proposition}[{\cite[Proposition 6.8]{ams20}}]
\label{prop:DerivativeParisi}

Let $\gamma\in \cuL$, and $\delta: [0,1)\to \reals$ be such that $\|\xi''\delta\|_{TV[0,t]}<\infty$ for all $t\in [0,1)$,
$\|\xi''\delta\|_1<\infty$, and $\delta(t) = 0$ for $t\in (1-\eps,1]$, $\eps>0$. Further assume that $\gamma+s\delta\ge 0$ for all $s\in [0,s_0)$ for some positive $s_0$.
Then
\begin{align}
\label{eq:ParisiVariation}
    \frac{\de \Par}{\de s}(\gamma+s\delta)\big|_{s=0+} 
    = 
    \frac{1}{2}\int_{0}^1 \xi''(t) \delta(t) \big(\E\lt[\partial_x\Phi_{\gamma}(t,X_t)^2\rt]\, -\, t\big)\, \de t\, .
\end{align}

\end{proposition} 

\begin{lemma}[{\cite[Lemma 6.9]{ams20}}]
\label{lem:intervals}
The support of $\gamma\in\cuL_{\lbq}$ is a disjoint union of countably many intervals $S(\gamma) = \cup_{\alpha\in A}I_\alpha$, where
$I_{\alpha} = (a_{\alpha},b_{\alpha})$ or $I_{\alpha} = [a_{\alpha},b_{\alpha})$, $\lbq\leq a_{\alpha}<b_{\alpha}\leq 1$, and $A$ is countable.
\end{lemma}

%
% \begin{proof}
% If $t\in S(\gamma)$ for $t\neq 0$, then by right continuity there exists $\delta>0$ such that $[t_0,t_0+\delta)\subseteq S(\gamma)$. This implies the claim. 
% \end{proof}

\begin{lemma}
\label{lem:strict}
The function $\Par=\Par_{\xi,\mathcal L_h}$ is strictly convex on $\cuL$.
\end{lemma}

\begin{proof}
The proof is identical to \cite[Theorem 2]{auffinger2015parisi} and \cite[Lemma 5]{chen2018energy} which show strict convexity on $\cuU$.
\end{proof}

Throughout this section we let $\gamma_*^{\cuL_{\lbq}}$ be the minimizer of $\Par$ over $\cuL_{\lbq}$, assuming it exists. Note that we will eventually show in Lemma~\ref{lem:identity} that $\gamma_*^{\cuL_{\lbq}}=\gamma_*^{\cuL}$ if either minimizer exists.

\begin{lemma}
\label{lem:stationarity}
Assume $\gamma_*^{\cuL_{\lbq}}$ exists. Then
\begin{align}
\label{eq:Support_1}
    t\in \supp\big(\gamma_*^{\cuL_{\lbq}}\big) &\;\;\;\Rightarrow\;\;\; 
    \E[\partial_x\Phi_{\gamma_*^{\cuL{\lbq}}}(t,X_t)^2] =t \, ,
    \\
\label{eq:Support_2}
    t\geq \lbq &\;\;\;\Rightarrow\;\;\; \E[\partial_x\Phi_{\gamma_*^{\cuL{\lbq}}}(t,X_t)^2] \ge t\, . 
\end{align}
\end{lemma}

\begin{proof}
We first show Equation~\eqref{eq:Support_1}. For $\lbq\le t_1<t_2< 1$ we take $\delta(t) =\gamma_*^{\cuL_{\lbq}}(t)1_{t\in [t_1,t_2)}$. Clearly $\gamma_*^{\cuL_{\lbq}}+s\delta\in \cuL_{\lbq}$. Since $\gamma_*^{\cuL_{\lbq}}$ minimizes $\Par(\cdot)$ over $\cuL_{\lbq}$, 

\begin{align*}
    0 \leq \left. \frac{\de \Par}{\de s}(\gamma_*^{\cuL_{\lbq}}+s\delta)\rt|_{s=0} 
    = 
    \frac{1}{2}\int_{t_1}^{t_2} \xi''(t) \gamma_*^{\cuL_{\lbq}}(t) \big(\E\lt[\partial_x\Phi_{\gamma_*^{\cuL{\lbq}}}(t,X_t)^2\rt]\, -\, t\big)\, \de t
\end{align*}

Since $t_1,t_2$ are arbitrary, and $\xi''(t)>0$ for $t\in (0,1)$ this implies
$\gamma_*^{\cuL_{\lbq}}(t)  (\E\lt[\partial_x\Phi_{\gamma_*^{\cuL{\lbq}}}(t,X_t)^2\rt]-t)=0$ for almost every $t\in [\lbq,1)$.
Since $\gamma_*^{\cuL_{\lbq}}(t)$ is right-continuous and $\E\lt[\partial_x\Phi_{\gamma_*^{\cuL{\lbq}}}(t,X_t)^2\rt]$ is continuous by Lemma~\ref{lem:continuous}, it follows that $\gamma_*^{\cuL_{\lbq}}(t)  (\E\lt[\partial_x\Phi_{\gamma_*^{\cuL{\lbq}}}(t,X_t)^2\rt]-t)=0$ for every $t\in [\lbq,1)$.
This in turns implies $\E\lt[\partial_x\Phi_{\gamma_*^{\cuL{\lbq}}}(t,X_t)^2\rt]=t$ for every $t\in S(\gamma_*^{\cuL_{\lbq}})$ by right-continuity of $\gamma_*^{\cuL_{\lbq}}$. This can be extended to all $t\in \supp(\gamma_*^{\cuL_{\lbq}})$ by again using continuity of $t\mapsto \E\lt[\partial_x\Phi_{\gamma_*^{\cuL{\lbq}}}(t,X_t)^2\rt]$.

Next consider Eq.~(\ref{eq:Support_2}), where it suffices now to consider $t\in [\lbq,1)\setminus \supp(\gamma_*^{\cuL_{\lbq}})$. By Lemma \ref{lem:intervals}, $[\lbq,1)\setminus \supp(\gamma_*^{\cuL_{\lbq}})$
is a disjoint union of open intervals. Let $J$ be such an interval, and consider any $[t_1,t_2]\subseteq J$. 
Set $\delta(t) = \ind(t\in(t_1,t_2])$, and notice that $\gamma_*^{\cuL_{\lbq}}+s\delta \in\cuL_{\lbq}$ for $s\ge 0$.
By Proposition \ref{prop:DerivativeParisi}, we have
\begin{align*}
    0\le \left. \frac{\de \Par}{\de s}(\gamma+s\delta)\rt|_{s=0} = \frac{1}{2}\int_{t_1}^{t_2} \xi''(t) \big(\E\lt[\partial_x\Phi_{\gamma_*^{\cuL{\lbq}}}(t,X_t)^2\rt]\, -\, t\big)\, \de t\, .
\end{align*}
Since $t_1,t_2$ are arbitrary, $\xi''(t)>0$ for $t\in(0,1)$ and $t\mapsto \E\lt[\partial_x\Phi_{\gamma_*^{\cuL{\lbq}}}(t,X_t)^2\rt]$ is continuous,
this implies $\E\lt[\partial_x\Phi_{\gamma_*^{\cuL{\lbq}}}(t,X_t)^2\rt]\ge t$ for all $t\in J$, and hence all $t\in [\lbq,1)\setminus \supp(\gamma_*^{\cuL_{\lbq}})$.
\end{proof}

\begin{corollary}\label{cor:DX2_supp}
Assume $\gamma_*^{\cuL_{\lbq}}$ exists. Then
\begin{align*}
t\in \supp(\gamma_*^{\cuL_{\lbq}}) &\;\;\;\Rightarrow\;\;\; \xi''(t)\E\lt[\partial_{xx}\Phi_{\gamma_*^{\cuL{\lbq}}}(t,X_t)^2\rt] =1\, .
\end{align*}
\end{corollary}
\begin{proof}
By Lemma \ref{lem:intervals}, $\supp(\gamma_*^{\cuL_{\lbq}})$ is a disjoint union of closed intervals with non-empty interior.
Let $K$ be one such interval. Then, for any $[t_1,t_2]\in K$, Corollary~\ref{cor:ED2} and Lemma \ref{lem:stationarity} imply
\begin{align*}
    t_2-t_1 = \E\lt[\partial_x\Phi(t_2,X_{t_2})^2\rt] - \E\lt[\partial_x\Phi(t_1,X_{t_1})^2\rt] 
    =
    \int_{t_1}^{t_2}\xi''(t)\E\lt[\partial_{xx}\Phi(t,X_t)^2\rt]  \de t\, .
\end{align*}
Since $t_1,t_2$ are arbitrary, $\xi''(t) \E\lt[\partial_{xx}\Phi(t,X_t)^2\rt]=1$ for  almost every $t\in K$.
By Lemma~\ref{lem:continuous} it follows that $\xi''(t) \E\lt[\partial_{xx}\Phi(t,X_t)^2\rt]=1$ 
for all $t\in \supp(\gamma_*^{\cuL_{\lbq}})$.
\end{proof}

\begin{lemma}[{\cite[Lemma 6.12]{ams20}}]
\label{lemma:DensityLB}
Let $\gamma\in\cuL$ satisfy $\gamma(t)=0$ for all $t\in (t_1,1)$, where $t_1<1$. Then, for any 
$t_*\in (t_1,1)$, the probability distribution of $X_{t_*}$ has a density $p_{t_*}$ with respect to the Lebesgue measure. Further, for any
$t_*\in (t_1,1)$  and any $M\in \reals_{\ge 0}$, there exists $\eps(t_*,M,\gamma)>0$ such that
\begin{align*}
\inf_{|x|\le M, t\in [t_*,1]} p_{t}(x) \ge \eps(t_*,M,\gamma)\, .
\end{align*}
\end{lemma}

\begin{lemma}
\label{lem:FullSupport}
If a minimizer $\gamma_*^{\cuL_{\lbq}}$ exists, then $\supp(\gamma_*^{\cuL_{\lbq}})= [\lbq,1)$.
\end{lemma}

\begin{proof}

By Lemma \ref{lem:intervals}, $[\lbq,1)\setminus\supp(\gamma_*^{\cuL_{\lbq}})$ is a countable union of disjoint  intervals,
open in $[\lbq,1)$. 
First assume that at least one of these intervals is of the form $(t_1,t_2)$ with $\lbq\leq t_1<t_2<1$. By Lemma
\ref{lem:stationarity} and Corollary \ref{cor:DX2_supp} we know that
\begin{align}
  \E\lt[\partial_x\Phi_{\gamma_*^{\cuL{\lbq}}}(t_i,X_{t_i})^2\rt] = t_i\, ,\;\;\;\;  \xi''(t_i)\E\lt[\partial_{xx}\Phi_{\gamma_*^{\cuL{\lbq}}}(t_i,X_{t_i})^2\rt] = 1\, ,\;\;\;\; i\in\{1,2\}\, , \label{eq:Ti}\\
\E\lt[\partial_x\Phi_{\gamma_*^{\cuL{\lbq}}}(t,X_t)^2\rt] \ge t\, \;\;\;\;\; \forall t\in(t_1,t_2)\, . \label{eq:Ti2}
\end{align}
Further, for $t\in (t_1,t_2)$, $\Phi_{\gamma_*^{\cuL{\lbq}}}$ solves the PDE \[\partial_t\Phi_{\gamma_*^{\cuL{\lbq}}}(t,x)+\frac{\xi''(t)}{2}\partial_{xx}\Phi_{\gamma_*^{\cuL{\lbq}}}(t,x)=0\] which is simply the heat equation up to a time change. We therefore obtain 
\begin{align*}
\Phi_{\gamma_*^{\cuL{\lbq}}}(t,x) =\E^{Z\sim\normal(0,1)}\lt[\Phi_{\gamma_*^{\cuL{\lbq}}}(t_2,x+\sqrt{\xi'(t_2)-\xi'(t)} \, Z)\rt], \quad \forall t\in (t_1,t_2] .
\end{align*}
Differentiating this equation and using dominated convergence (recall that $\partial_{xx}\Phi_{\gamma_*^{\cuL{\lbq}}}(t_2,x)$
is bounded by Proposition~\ref{prop:phireg}), 
we obtain 
\[
    \partial_{xx}\Phi_{\gamma_*^{\cuL{\lbq}}}(t,x) = \E^{Z\sim\normal(0,1)}\lt[\partial_{xx}\Phi_{\gamma_*^{\cuL{\lbq}}}(t_2,x+\sqrt{\xi'(t_2)-\xi'(t)} \, Z)\rt]\,.
\]
Because $\de X_t = \sqrt{\xi''(t)}\, \de B_t$, we can rewrite the last equation as
\begin{align*}
\partial_{xx}\Phi_{\gamma_*^{\cuL{\lbq}}}(t,X_t) =\E\lt[\partial_{xx}\Phi_{\gamma_*^{\cuL{\lbq}}}(t_2,X_{t_2})|X_t\rt].
\end{align*}
By Jensen's inequality, 
\begin{align}
\E\lt[\partial_{xx}\Phi_{\gamma_*^{\cuL{\lbq}}}(t,X_t)^2\rt] \le \E\lt[\partial_{xx}\Phi_{\gamma_*^{\cuL{\lbq}}}(t_2,X_{t_2})^2\rt] = \frac{1}{\xi''(t_2)}\, ,
\end{align}
where in the last step we used Eq.~\eqref{eq:Ti}. Using Corollary \ref{cor:ED2} we get, for $t\in [t_1,t_2]$
\begin{align*}
\E\lt[\partial_{x}\Phi_{\gamma_*^{\cuL{\lbq}}}(t,X_t)^2\rt]& = \E\lt[\partial_{x}\Phi_{\gamma_*^{\cuL{\lbq}}}(t_1,X_{t_1})^2\rt] +
\int_{t_1}^{t}  \xi''(s)\E\lt[\partial_{xx}\Phi_{\gamma_*^{\cuL{\lbq}}}(s,X_{s})^2\rt] \,\de s\\
& \le t_1+\int_{t_1}^{t}  \frac{\xi''(s)}{\xi''(t_2)}\,\de s< t\, ,
\end{align*}
where in the last step we used the fact that $t\mapsto \xi''(t)$ is increasing. The last equation is in contradiction with Eq.~\eqref{eq:Ti2}, and therefore $[\lbq,1)\setminus\supp(\gamma_*^{\cuL_{\lbq}})$ is either empty or consists of a single interval $(t_1,1)$.

In order to complete the proof, we need to rule out the case $[\lbq,1)\setminus\supp(\gamma_*^{\cuL_{\lbq}})= (t_1,1)$.
Assume for sake of contradiction that indeed $[\lbq,1)\setminus\supp(\gamma_*^{\cuL_{\lbq}})= (t_1,1)$. For $t\in (t_1,1)$, let $r= r(t) =\xi'(1)-\xi'(t)$, and notice that 
$r(t)$ is decreasing with $r(t) = \xi''(1)(1-t)+O((1-t)^2)$ as $t\to 1$. 
By solving the Parisi PDE in the interval $(t_1,1)$, we find that for all $t\in(t_1,1)$,
\[\partial_x \Phi_{\gamma_*^{\cuL{\lbq}}}(t,x) = \E^{Z\sim\normal(0,1)}\lt[\sign\lt(Z+\frac{x}{\sqrt{r(t)}}\rt)\rt]\]
and therefore
\begin{align*}
1-\E\lt[\partial_x\Phi_{\gamma_*^{\cuL{\lbq}}}(t,X_t)^2\rt] &=\E \lt[Q\lt(\frac{X_t}{\sqrt{r(t)}}\rt)\rt] \, ,\\
Q(x) &\equiv 1- \E^{Z\sim\normal(0,1)}\lt[\sign(x+Z)\rt]^2\, .
\end{align*}
Note that $0\le Q(x)\le 1$ is continuous, with $Q(0) = 1$. Hence, there exists a numerical constant $\delta_0\in (0,1)$ such that $Q(x) \ge 1/2$ for $|x|\le \delta_0$.
Therefore, fixing $t_*\in (t_1,1)$, for any $t\in(t_*,1)$
\begin{align*}
1-\E\lt[\partial_x\Phi_{\gamma_*^{\cuL{\lbq}}}(t,X_t)^2\rt] &\ge \frac{1}{2}\prob\lt[|X_t|\le \delta_0\sqrt{r(t)}\rt] \\
& \stackrel{(a)}{\ge} \delta_0\eps(t_*,1,\gamma)\, \sqrt{r(t)}\stackrel{(b)}{\ge} C \sqrt{1-t}\, ,
\end{align*}
where $(a)$ follows by Lemma \ref{lemma:DensityLB} and $(b)$ holds for some $C=C(\gamma)>0$.
We therefore obtain $\E\lt[\partial_x\Phi_{\gamma_*^{\cuL{\lbq}}}(t,X_t)^2\rt]\le 1-C\sqrt{1-t}$, which contradicts Lemma~\ref{lem:stationarity} for $t$ close enough to $1$.
\end{proof}

In the next lemma, we show that minimization of $\Par$ over $\cuL$ subsumes minimization over $\cuL_{\lbq}$. A priori, one might expect that tuning the value of $\lbq$ could lead to many different minima.

{\color{black}

\begin{lemma}\label{lem:extra}

For $\gamma,\wh\gamma\in \cuL$, define the function $\gamma^{(\eps)}\in \cuL$ by:
\[
\gamma^{(\eps)}(t)
=
\begin{cases}
\gamma(t),\quad 0\leq t<1-\eps,
\\
\wh\gamma(t),\quad 1-\eps\leq t<1.
\end{cases}
\]
Then $\lim_{\eps\to 0} \Par(\gamma^{(\eps)})=\Par(\gamma)$.
\end{lemma}

\begin{proof}
    Using \cite[Proposition 6.1(c)]{ams20} and continuity, we have $\|\Par(\gamma)-\Par(\gamma^{(\eps)})\|\leq C\|\gamma-\gamma^{(\eps)}\|$ for a constant $C$ depending only on $\xi$. The right-hand side tends to zero as $\eps\to 0$ by the definition of $\cuL$.
\end{proof}

}

\begin{lemma}\label{lem:qnotneeded}

Suppose $\gamma_*^{\cuL_{\lbq}}$ exists. Then $\gamma_*^{\cuL}=\gamma_*^{\cuL_{\lbq}}$.

\end{lemma}

\begin{proof}

Let $f(t)= \E[\partial_x\Phi_{\gamma_*^{\cuL{\lbq}}}(t,X_t)^2].$ 
First we show that $f(t)\geq t$ for all $0\leq t\leq \lbq$.
Recall that $X_t$ is simply a time-changed Brownian motion on $0\leq t\leq \lbq$ while $\Phi_{\gamma_*^{\cuL{\lbq}}}$ solves the time-changed heat equation on the same time interval, therefore $\partial_{xx}\Phi_{\gamma_*^{\cuL{\lbq}}}(t,X_t)=\E^t[\partial_{xx}\Phi_{\gamma_*^{\cuL{\lbq}}}(\lbq,X_{\lbq})]$. By Jensen's inequality, it follows that for all $0\leq t\leq \lbq$, 
 \begin{align*}
    \E[\partial_{xx}\Phi_{\gamma_*^{\cuL{\lbq}}}(t,X_t)^2]
    &\leq 
    \E[\partial_{xx}\Phi_{\gamma_*^{\cuL{\lbq}}}(\lbq,X_{\lbq})^2]
    \\
    &= 
    \frac{1}{\xi''(\lbq)}
    \\
    &\leq 
    \frac{1}{\xi''(t)}. 
\end{align*}
In the last line we used that $\xi''$ is increasing as $\xi$ is a power series with non-negative coefficients. Next, from Lemma~\ref{lem:stationarity} and Lemma~\ref{lem:FullSupport} it follows that $f(\lbq)=\lbq$. In light of Corollary~\ref{cor:ED2}, we showed just above that $f'(t)\leq 1$ for $t\leq \lbq$. It now follows that $f(t)\geq t$ for all $0\leq t\leq \lbq$.

{\color{black}
Next, Proposition~\ref{prop:DerivativeParisi} combined with Corollary~\ref{cor:ED2} and Lemma~\ref{lem:stationarity} imply that $\frac{\de}{\de s}\Par\big((1-s)\gamma_*^{\cuL_{\lbq}}+s\gamma\big)|_{s=0+}\geq 0$ for any $\gamma\in\cuL$ 
agreeing with $\gamma_*^{\cuL_{\lbq}}$ on $[1-\eps,1)$ for some $\eps>0$. 
Using convexity of $\Par$ as guaranteed by Lemma~\ref{lem:strict}, this implies that 
\[
    \Par(\gamma_*^{\cuL_{\lbq}})\leq \Par(\gamma)
\]
for any such $\gamma$. Assuming for sake of contradiction that some $\gamma_*\in\cuL$ satisfies $\Par(\gamma_*)<\Par(\gamma_*^{\cuL_{\lbq}})$, define $\gamma^{(\eps)}\in\cuL$ by
\[
\gamma^{(\eps)}(t)
=
\begin{cases}
\gamma_*(t),\quad 0\leq t<1-\eps,
\\
\gamma_*^{\cuL_{\lbq}},\quad 1-\eps\leq t<1.
\end{cases}
\]
Then Lemma~\ref{lem:extra} implies that $\lim_{\eps\to 0}\Par(\gamma^{(\eps)})=\Par(\gamma_*)<\Par(\gamma_*^{\cuL_{\lbq}})$. This contradicts the above for small enough $\eps$, completing the proof.
}
\end{proof}

We now restate and prove Lemmas~\ref{lem:identity} and \ref{lem:optimizable}.

\lemidentity*

\begin{proof} Lemma~\ref{lem:strict} immediately implies uniqueness of minimizers. The second statement immediately implies the third, while Lemma~\ref{lem:qnotneeded} provides the converse result. To show that the first statement implies the third, we observe that Proposition~\ref{prop:DerivativeParisi} immediately yields
\[\frac{\de}{\de s}\Par((1-s)\gamma_*+s\gamma)|_{s=0+}=0\] for any $\gamma\in\cuL_{\lbq}$ when $\gamma_*$ is optimizable; this implies the third statement by again invoking Lemma~\ref{lem:strict}. It only remains to show that if $\Par(\gamma_*)=\inf_{\gamma\in\cuL}\Par(\gamma)$, then $\gamma_*$ is $\lbq$-optimizable, which follows from Lemmas~\ref{lem:stationarity} and~\ref{lem:FullSupport}. 
\end{proof}

\optimizable* 

\begin{proof}
    Fix $\lbq<t_1<t_2<1$ and define $\delta(t) = [\gamma_*^{\cuU}(t_1)-\gamma_*^{\cuU}(t)]\ind_{(t_1,t_2)}(t)$. It is easy to see that this satisfies the assumptions of Proposition~\ref{prop:DerivativeParisi} with $s_0=1$. Letting $\gamma^s=\gamma_*^{\cuU}+s\delta$,
\begin{align*}
  \left. \frac{\de \Par}{\de s}(\gamma^s)\rt|_{s=0+} = -\frac{1}{2}\int_{t_1}^{t_2} \xi''(t) \big(\gamma_*^{\cuU}(t)-\gamma_*^{\cuU}(t_1))\, 
  \big(\E\lt[\partial_x\Phi_{\gamma_*^{\cuU}}(t,X_t)^2\rt]\, -\, t\big)\, \de t.
\end{align*}
On the other hand, $\gamma^s\in \cuU$ for $s\in [0,1]$ (since $\gamma_*^{\cuU}$ is strictly increasing on $[\lbq,1)$), so
\begin{align*}
\int_{t_1}^{t_2} \xi''(t) \big(\gamma_*^{\cuU}(t)-\gamma_*^{\cuU}(t_1))\, 
  \big(\E[\partial_x\Phi_{\gamma_*^{\cuU}}(t,X_t)^2]\, -\, t\big)\, \de t \le 0\, .
\end{align*}
for all $t_1<t_2$. Since $\gamma_*^{\cuU}(t)-\gamma_*^{\cuU}(t_1)>0$ strictly for all $t>t_1$, this implies \[\E[\partial_x\Phi_{\gamma_*^{\cuU}}(t,X_t)^2] \le t\] for almost every $t$, and therefore for every $t$. The inequality \[\E[\partial_x\Phi_{\gamma_*^{\cuU}}(t,X_t)^2] \ge t\] is proved in the same way using $\delta(t) = [\gamma_*^{\cuU}(t_2)-\gamma_*^{\cuU}(t)]\ind_{(t_1,t_2)}(t)$.
  \end{proof}

Finally we turn to Lemma~\ref{lem:identity2}.

% \lemidentitytwo*

\begin{proof}[Proof of Lemma~\ref{lem:identity2}] 
First, \eqref{eq:id2} is clear given Corollary~\ref{cor:ED2}. To establish \eqref{eq:id3}, we first show that 
\begin{equation}
\label{eq:DxxPhi}
    \de(\partial_{xx}\Phi_{\gamma_*}(t,X_t))
    = -\xi''(t)\gamma_*(t)(\partial_{xx}\Phi_{\gamma_*}(t,X_t))^2\de t
    + \partial_{xxx}\Phi_{\gamma_*}(t,X_t)\sqrt{\xi''(t)}\de B_t.
\end{equation}
Indeed \eqref{eq:DxxPhi} follows by using Ito's formula to derive
\begin{align*}
    \de(\partial_{xx}\Phi_{\gamma_*}(t,X_t))
    &= \lt(
        \partial_{txx}\Phi_{\gamma_*}(t,X_t)
        + \partial_{x}\Phi_{\gamma_*}(t,X_t)\partial_{xxx}\Phi_{\gamma_*}(t,X_t)\xi''(t)\gamma_*(t)
        + \frac{\xi''(t)\partial_{xxxx}\Phi_{\gamma_*}(t,X_t)}{2}
    \rt)\de t \\
    &\quad +\partial_{xxx}\Phi_{\gamma_*}(t,X_t)\sqrt{\xi''(t)}\de B_t
\end{align*}
and taking the second derivative with respect to $x$ of the Parisi PDE to obtain
\begin{align*}
    0 &= \partial_{xx}\lt(
        \partial_t\Phi_{\gamma_*}(t,x)
        + \frac{\xi''(t)}{2}\lt(
            \partial_{xx}\Phi_{\gamma_*}(t,x)
            + \gamma_*(t)(\partial_x\Phi_{\gamma_*}(t,x))^2
        \rt)
    \rt) \\
    &= \partial_{txx}\Phi_{\gamma_*}(t,x)
        + \frac{\xi''(t)\partial_{xxxx}\Phi_{\gamma_*}(t,x)}{2}
        + \xi''(t)\gamma_*(t)\lt(
            (\partial_{xx}\Phi_{\gamma_*}(t,x))^2
            + \partial_{x}\Phi_{\gamma_*}(t,x)\partial_{xxx}\Phi_{\gamma_*}(t,x)
        \rt).
\end{align*}
In particular \eqref{eq:DxxPhi} implies that for all $t\in [0,1)$,
\begin{align*}
    \frac{\de}{\de t}\mathbb E[\partial_{xx}\Phi_{\gamma_*}(t,x)]
    &= -\xi''(t)\gamma_*(t)\mathbb E[(\partial_{xx}\Phi_{\gamma_*}(t,x))^2] \\
    &= -\gamma_*(t).
\end{align*}
Therefore to show \eqref{eq:id3} it suffices to show 
{\color{black}
$\lim_{t\to 1}\mathbb E[\partial_{xx}\Phi_{\gamma_*}(t,X_t)]\geq 0$, but this is clear by convexity of $\Phi_{\gamma_*}(t,\cdot)$}. 
\end{proof}

\section{Incremental AMP Proofs}
\label{ap:iamp}

We will prove Lemma~\ref{lem:BMlimit2} which generalizes Lemma~\ref{lem:BMlimit} to the setting of branching AMP and describes the limiting Gaussian processes $N^{\delta}_{\ell,a},Z^{\delta}_{\ell,a}.$ We recall the setup of Section~\ref{sec:branch} and in particular continue to use the value $q_B\in (\lbq,1)$ to define the time $\ell^{\delta}_{q_B}$ at which $Z_{\ell^{\delta}_{q_B},1}^\delta= Z_{\ell^{\delta}_{q_B},2}^\delta$ last holds. For the branching setting we slightly generalize the filtration \eqref{eq:Felldelta} to
\[
    \mathcal F_{\ell}^{\delta}=\sigma\lt((Z^{\delta}_{k,a},N^{\delta}_{k,a})_{0\leq k\leq \ell,a\in \{1,2\}}\rt).
\]

Crucially note that we restrict here to $k\geq 0$, i.e. we do not include the preparatory iterates with negative index. We remark that if we consider all the IAMP iterates $(Z_{\ell,a}^{\delta},N_{\ell,a}^{\delta})$ together in the linear order given by $(\ell,a)\to 2\ell+a$, then these are iterates of a standard AMP algorithm since each iterate depends only on the previous ones. 
{\color{black}
(This is because in \eqref{eq:onsagerdef} the last expectation is zero when $f_{\ell}$ does not depend on $Z^j$, which is the case if $f_{\ell}$ and $Z^j$ correspond to different values of $a$ after the two branches separate.)
}
Moreover it is easy to see that the Onsager correction terms are not affected by this rewriting. Therefore we may continue to use state evolution in the natural way even though we do not think of the iterates as actually being totally ordered.

\begin{lemma}\label{lem:independent}

In branching IAMP, $\mathcal F_{\ell}^{\delta}$ is jointly independent of the iterates $(Z^{-j})_{J_{\ell}< j\leq K}$ for \[J_{\ell}\equiv \max\lt(\{\ell\}\cup \{k_{i,a}+\ell-\ell^{\delta}_{q_i}:\ell^{\delta}_{q_i}\leq \ell,a\in \{1,2\}\}\rt).\]

\end{lemma}

\begin{proof} We proceed by induction over $\ell$, the base case $\ell=0$ following from Proposition~\ref{prop:init}. Because the random variables $Z^{\ell}_{k,a}$ form a Gaussian process it suffices to verify that 
\[
    \E\lt[Z^{\delta}_{\ell,a}Z^{-j}\rt]=0
\]
holds whenever $j>J_{\ell}$. 
By state evolution,
\[
    \E\lt[Z^{\delta}_{\ell,a}Z^{-j}\rt]=\xi'\lt(\lt[N^{\delta}_{\ell-1,a}Z^{-j-1}\rt]\rt).
\]
By definition $N^{\delta}_{\ell-1,a}$ is $\mathcal F^{\delta}_{\ell-1}$-measurable. Since $\xi'(0)=0$ it suffices to show that $\mathcal F^{\delta}_{\ell-1}$ is independent of $Z^{-j-1}$. By the inductive hypothesis, this holds if $j+1>J_{\ell-1}$. This in turn follows from the easy-to-verify fact that $J_{\ell}-1\geq J_{\ell-1}$, completing the proof.
\end{proof}

\begin{corollary}
\label{cor:independent}
Let $G^{\delta}_{q_j,a}$ be the state evolution limit of $\bg^{(q_j,a)}$ for each $(j,a)\in [m]\times [2]$. Then the law of $(G^{\delta}_{q_i,1},G^{\delta}_{q_i,2})$ conditioned on $\mathcal F^{\delta}_{\ell^{\delta}_{q_i}-1}$ is $\normal(0,I_2)$.
\end{corollary}

\begin{proof}
    Since $k_{i,1}\neq k_{i,2}$ it follows from Proposition~\ref{prop:init} that $(G^{\delta}_{q_i,1},G^{\delta}_{q_i,2})\sim \normal(0,I_2)$ holds as an unconditional law. Since we chose the values $k_{i,a}$ such that $k_{i,a}-\ell^{\delta}_{q_i}>k_{j,a'}-\ell^{\delta}_{q_j}>0$ for any $i>j$ and $a,a'\in \{1,2\}$, it follows that $k_{i,a}>J_{\ell^{\delta}_{q_i}-1}$. 
    Applying Lemma~\ref{lem:independent} now concludes the proof.
\end{proof}

\begin{restatable}{lemma}{BMlimit2}
\label{lem:BMlimit2}
The sequences $(Z_{\ul,a}^{\delta},Z_{\ul+1,a}^{\delta},\dots)$ and $(N_{\ul,a}^{\delta},N_{\ul+1,a}^{\delta},\dots)$ satisfy for $\ell\geq \ul$:
\begin{align} 
\label{BM1.2}
    \mathbb E[(Z^{\delta}_{\ell+1,a}-Z^{\delta}_{\ell,a})Z_{j,a}^{\delta}]
    &=0,
    \quad \text{for all }\ul+1\leq j\leq \ell
    \\
\label{BM2.2}
    \mathbb E[(Z^{\delta}_{\ell+1,a}-Z^{\delta}_{\ell,a})^2|\mathcal F_{\ell}^{\delta}]
    &=
    \xi'(q^{\delta}_{\ell+1})-\xi'(q^{\delta}_{\ell})
    \\
\label{BM2.52}
    \mathbb E[(Z^{\delta}_{\ell+1,1}-Z^{\delta}_{\ell,1})(Z^{\delta}_{\ell+1,2}-Z^{\delta}_{\ell,2})|\mathcal F_{\ell}^{\delta}]
    &=
    \lt(\xi'(q^{\delta}_{\ell+1})-\xi'(q^{\delta}_{\ell})\rt)\cdot 1_{\ell<\ell^{\delta}_{q_B}}
    \\
\label{BM3.2}
    \mathbb E[(Z^{\delta}_{\ell,a})^2]
    &=
    \xi'(q_{\ell}^{\delta})
    \\
\label{BM4.2}
    \mathbb E[(N^{\delta}_{\ell+1,a}-N^{\delta}_{\ell,a})|\mathcal F_{\ell}^{\delta}]
    &=0
    \\
\label{BM5.2}\mathbb E[(N^{\delta}_{\ell+1,a}-N^{\delta}_{\ell,a})^2|\mathcal F_{\ell}^{\delta}]&=\delta
    \\
\label{BM5.52}
    \mathbb E[(N^{\delta}_{\ell+1,1}-N^{\delta}_{\ell,1})(N^{\delta}_{\ell+1,2}-N^{\delta}_{\ell,2})|\mathcal F_{\ell}^{\delta}]
    &=
    \delta\cdot 1_{\ell<\ell^{\delta}_{q_B}}
    \\
\label{BM6.2}
    \mathbb E[(N^{\delta}_{\ell,a})^2]
    &=
    q_{\ell+1}^{\delta}.
\end{align}

\end{restatable}

\begin{proof}
We recall that $(Z^{\delta}_{\ell,a})_{\ell\geq\ul,a\in\{1,2\}}$ is a Gaussian process, which means we can ignore the conditioning on $\mathcal F_{\ell}^{\delta}$ in proving Equation~\eqref{BM2.2}. First we check that Equations~\eqref{BM3.2} and \eqref{BM6.2} hold for $\ell=\ul$. For Equation~\eqref{BM6.2}, {\color{black}by definition of $\delta$:
\[
    \mathbb E[(N_{\ul,a}^{\delta})^2]=(1+\eps_0)^2\mathbb E[(M^{\ul})^2]=(1+\eps_0)^2\lbq =\lbq+\delta=q_1^{\delta}.
\]
}
For Equation~\eqref{BM3.2},
\[
    \mathbb E[(Z_{\ul,a}^{\delta})^2]=\xi'\lt(\mathbb E[(M^{\ul-1})^2]\rt)=\xi'({\color{black}\lbq}).
\]
Observe now that if Equations~\eqref{BM1.2}, \eqref{BM2.2}, \eqref{BM4.2}, \eqref{BM5.2} hold for $\ul\leq \ell\leq k$ then so do Equations~\eqref{BM3.2} and \eqref{BM6.2}, as
\[
    \mathbb E[(N^{\delta}_{\ell+1,a})^2]
    =
    \mathbb E[(N^{\delta}_{\ell+1,a}-N^{\delta}_{\ell,a})^2] + 2\cdot \mathbb E[(N^{\delta}_{\ell+1,a}-N^{\delta}_{\ell,a})N_{\ell,a}^{\delta}]+\mathbb E[(N^{\delta}_{\ell,a})^2]
\]
and similarly for $\mathbb E[(Z^{\delta}_{\ell+1,a})^2]$. Therefore to show the six identities \eqref{BM1.2}, \eqref{BM2.2}, \eqref{BM4.2}, \eqref{BM5.2}, \eqref{BM3.2} and \eqref{BM6.2} it suffices to check the base case $\ell=\ul$ for Equations~\eqref{BM1.2}, \eqref{BM2.2}, \eqref{BM4.2}, \eqref{BM5.2} and to perform an inductive step to show these four identities for $\ell=k+1$, assuming all six of these equations as inductive hypotheses for $\ell\leq k$. We turn to doing this, and finally show Equations~\eqref{BM2.52}, \ref{BM5.52} at the end.

\paragraph{Base Case for Equations~\eqref{BM1.2}, \eqref{BM2.2}, \eqref{BM4.2}, \eqref{BM5.2}.}

Note that here, none of the perturbations $\bg^{(q_i,a)}$ appear yet. We begin with Equation~\eqref{BM1.2}:
\begin{align*}
    \E\lt[\lt(Z^{\delta}_{\ul+1,a}-Z^{\delta}_{\ul,a}\rt)Z^{\delta}_{\ul,a}\rt] 
    &=  
    \xi'\lt(\mathbb E\lt[N^{\delta}_{\ul,a} M^{\ul-1}\rt]\rt)- \xi'\lt(\E\lt[M^{\ul-1} M^{\ul-1}\rt]\rt) 
    \\
    &=  \xi'\lt((1+\eps_0)  \E\lt[ M^{\ul} M^{\ul-1}\rt]\rt)- \xi'({\color{black}\lbq})
    \\
    &=  \xi'((1+\eps_0)\phi(a_{\ul-1})) - \xi'({\color{black}\lbq}) 
    \\
    &=\xi'({\color{black}\lbq})-\xi'({\color{black}\lbq})
    \\
    &= 0.
\end{align*}
This means $\mathbb E[Z^{\delta}_{\ul+1,a}|Z^{\delta}_{\ul,a}]=Z^{\delta}_{\ul,a}$. Hence
\[
    \E\lt[\lt(Z^{\delta}_{\ul+2,a}-Z^{\delta}_{\ul+1,a}\rt)Z^{\delta}_{\ul+1,a}\rt]
    =
    \xi'\lt(\mathbb E[N^{\delta}_{\ul+1}N^{\delta}_{\ul,a}]\rt)-\xi'\lt(\mathbb E[N^{\delta}_{\ul,a}N^{\delta}_{\ul,a}]\rt).
\]
To see that the above expression vanishes, it suffices to show that
\[
    \mathbb E[(N^{\delta}_{\ul+1,a}-N^{\delta}_{\ul,a})N^{\delta}_{\ul,a}]=0.
\]
This follows since we just showed $\mathbb E[Z^{\delta}_{\ul+1,a}|Z^{\delta}_{\ul,a}]=Z^{\delta}_{\ul,a}$ and we have
\[
    \mathbb E[(N^{\delta}_{\ul+1,a}-N^{\delta}_{\ul,a})N^{\delta}_{\ul,a}]
    =
    \E[u^{\delta}_{\ul,a}(X^{\delta}_{\ul,a})(Z^{\delta}_{\ul+1,a}-Z^{\delta}_{\ul,a})]=\E[u^{\delta}_{\ul,a}(Z^{\delta}_{\ul,a})(Z^{\delta}_{\ul+1,a}-Z^{\delta}_{\ul,a})]
\]

Next we verify the base case for Equation~\eqref{BM2.2}. Using the base case of Equation~\eqref{BM1.2} in the first step we compute:
\begin{align*}
    \E\lt[\lt(Z^{\delta}_{\ul+1,a}-Z^{\delta}_{\ul,a}\rt)^2\rt] &= \E\lt[\lt(Z^{\delta}_{\ul+1,a}\rt)^2\rt] -\E\lt[\lt(Z^{\delta}_{\ul,a}\rt)^2\rt]
    \\
    &= 
    \xi'\lt(\E\lt[(N^{\delta}_{\ul,a})^2\rt]\rt) - \xi'({\color{black}\lbq})
    \\
    &= 
    \xi'\big((1+\eps_0)^2 {\color{black}\lbq}\big) -\xi'({\color{black}\lbq})
    % \\
    % & = 
    % \xi'\lt(\frac{q^2}{\phi(a_{\ul-1})^2}\rt)-\xi'({\color{black}\lbq})
    \\
    &=
    \xi'({\color{black}\lbq}+\delta)-\xi'({\color{black}\lbq})
    \\
    &=
    \xi'(q^{\delta}_{\ul+1})-\xi'({\color{black}\lbq}).
\end{align*}
Continuing, we verify the base case for Equation~\eqref{BM4.2}. First note that
\begin{align*}
    \E\lt[\lt(N^{\delta}_{\ul+1,a}-N^{\delta}_{\ul,a}\rt)|\mathcal F^{\delta}_{\ul}\rt]
    &=
    \mathbb E\lt[
        u^{\delta}_{\ul}
        (X^{\delta}_{\ul,a})
        (Z^{\delta}_{\ul+1,a}-Z^{\delta}_{\ul,a})
        \,|\,
        \mathcal F^{\delta}_{\ul}
    \rt]
    \\
    &=0.
\end{align*}
The last line holds because $X^{\delta}_{\ul,a}$ is $\mathcal F_{\ul}^{\delta}$-measurable and $\mathbb E[Z^{\delta}_{\ul+1,a}-Z^{\delta}_{\ul,a}|\mathcal F_{\ul}^{\delta}]=0$ as deduced above. Finally for Equation~\eqref{BM5.2} using the martingale property again we obtain:
\begin{align*}
    \E\lt[\lt(N^{\delta}_{\ul+1,a}-N^{\delta}_{\ul,a}\rt)^2\rt]
    &=
    \E\lt[(u^{\delta}_{\ul}(X^{\delta}_{\ul,a}))^2(Z^{\delta}_{\ul+1,a}-Z^{\delta}_{\ul,a})^2\rt]
    \\
    &=
    \E\lt[\frac{\delta}{\xi'(q^{\delta}_{\ul+1,a})-\xi'(q^{\delta}_{\ul})}\lt(\xi'(q^{\delta}_{\ul+1})-\xi'(q^{\delta}_{\ul})\rt)\rt]   
    \\
    &=\delta.
\end{align*}
Here the second line follows from the definition of $u^{\delta}_{\ell}$, and we can multiply the two expectations because $\mathbb E[(Z^{\delta}_{\ul+1,a}-Z^{\delta}_{\ul,a})^2|\mathcal F^{\delta}_{\ul,a}]$ is constant while the other term is $\mathcal F^{\delta}_{\ul,a}$ measurable.

\paragraph{Inductive step}

We now induct, assuming all $6$ identities \eqref{BM1.2}, \eqref{BM2.2}, \eqref{BM4.2}, \eqref{BM5.2}, \eqref{BM3.2} and \eqref{BM6.2}  up to $\ell$ and showing Equations~\eqref{BM1.2}, \eqref{BM2.2}, \eqref{BM4.2}, \eqref{BM5.2} for $\ell+1$. We begin with Equation~\eqref{BM1.2}. Let $\ul+1\le j \le \ell$. State evolution implies  
\begin{align*}
\E\lt[\lt(Z^{\delta}_{\ell+1,a}-Z^{\delta}_{\ell,a}\rt) Z^{\delta}_{j,a}\rt] &=  \xi'\lt(\E\lt[N^{\delta}_{\ell,a} N^{\delta}_{j-1,a}\rt]\rt) -  \xi'\lt(\E\lt[N^{\delta}_{\ell-1,a} N^{\delta}_{j-1,a}\rt]\rt).
\end{align*}
To show this equals $0$ we must show 
\[
    \E\lt[N^{\delta}_{\ell,a} N^{\delta}_{j-1,a}\rt] =\E\lt[N^{\delta}_{\ell-1,a} N^{\delta}_{j-1,a}\rt].
\]
When $\ell=\ell^{\delta}_{q_i}$ for some $i\in [m]$ this follows from Corollary~\ref{cor:independent}. Assuming $\ell\neq\ell^{\delta}_{q_i}$ for all $i$, the difference between the left and right sides is
\[
    \mathbb E[u^{\delta}_{\ell-1}(X^{\delta}_{\ell-1,a})(Z^{\delta}_{\ell,a}-Z^{\delta}_{\ell-1,a})N^{\delta}_{j-1,a}].
\]
Since $N^{\delta}_{j-1,a}$ is $\mathcal F_{\ell-1}^{\delta}$ measurable and $\E[Z^{\delta}_{\ell,a}|\mathcal F_{\ell-1}^{\delta}]=Z^{\delta}_{\ell-1,a}$ holds by inductive hypothesis, we conclude the inductive step for Equation~\eqref{BM1.2}.

%Therefore $(Z^{\delta}_j)_{j=\ul}^{\ell+1}$ has independent increments, and $(N^{\delta}_j)_{j=\ul}^{\ell+1}$ is a martingale. 
We continue to Equation~\eqref{BM2.2}. Using Equation~\eqref{BM1.2} just proven in the first step we get
\begin{align*}
\E\lt[\lt(Z^{\delta}_{\ell+1,a}-Z^{\delta}_{\ell,a}\rt)^2\rt] &=  \E\lt[\lt(Z^{\delta}_{\ell+1,a}\rt)^2-\lt(Z^{\delta}_{\ell,a}\rt)^2\rt] \\
&= \xi'\lt(\E[N^{\delta}_{\ell,a}N^{\delta}_{\ell,a}]\rt)-\xi'\lt(\E[N^{\delta}_{\ell-1,a}N^{\delta}_{\ell-1,a}]\rt)\\
&= \xi'\lt(q^{\delta}_{\ell+1}\rt)-\xi'\lt(q^{\delta}_{\ell}\rt)
\end{align*}

Next we show Equation~\eqref{BM4.2} continues to hold. If $\ell+1=\ell^{\delta}_{q_i}$ for some $i\in [m]$ again follows from Corollary~\ref{cor:independent}. When $\ell+1\neq\ell^{\delta}_{q_i}$ for all $i$, it follows from the definition of the sequence $N^{\delta}_{\ell,a}$ and the just proven fact that $(Z^{\delta}_{\ell,a})_{\ell\geq \ul+1}$ forms a martingale sequence. Finally we show Equation~\eqref{BM5.2} continues to hold inductively. Again for $\ell+1=\ell^{\delta}_{q_i}$ it follows from Corollary~\ref{cor:independent}, and otherwise by definition 
\[
    \mathbb E[(u^{\delta}_{\ell})^2]
    =
    \frac{\delta}{\xi'(q^{\delta}_{\ell})-\xi'(q^{\delta}_{\ell-1})}.
\]
Moreover what we showed before implies $\E[(Z^{\delta}_{\ell+1,a}-Z^{\delta}_{\ell,a})^2|\mathcal F^{\delta}_{\ell}]=\xi'(q^{\delta}_{\ell})-\xi'(q^{\delta}_{\ell-1})$. Applying these observations to the identity
\[
    \mathbb E[(N^{\delta}_{\ell+1,a}-N^{\delta}_{\ell})^2|\mathcal F_{\ell}]
    =
    (u^{\delta}_{\ell})^2\mathbb E\lt[(Z^{\delta}_{\ell+1,a}-Z^{\delta}_{\ell,a})^2\rt]
\]
implies Equation~\eqref{BM5.2} continues to hold. 

\paragraph{Equations~\eqref{BM2.52} and \eqref{BM5.52} }

Finally we consider \eqref{BM2.52} and \eqref{BM5.52}. For $\ell<\ell^{\delta}_{q_i}$ they follow directly from \eqref{BM2.2},~\eqref{BM5.2}. For $\ell=\ell^{\delta}_{q_i}$, \eqref{BM5.52} is trivial while \eqref{BM2.52} immediately follows from state evolution. For $\ell>\ell^{\delta}_{q_i}$, \eqref{BM5.52} follows from the inductive hypothesis and the computation
\[
    \mathbb E[
    (N^{\delta}_{\ell+1,1}-N^{\delta}_{\ell,1})
    (N^{\delta}_{\ell+1,2}-N^{\delta}_{\ell,2})
    |\mathcal F_{\ell}^{\delta}
    ]
    =
    (u^{\delta}_{\ell})^2
    \mathbb E\lt[
    (Z^{\delta}_{\ell+1,1}-Z^{\delta}_{\ell,1})
    (Z^{\delta}_{\ell+1,2}-Z^{\delta}_{\ell,2})
    |\mathcal F_{\ell}^{\delta}
    \rt]
    =0.
\]
Finally for $\ell>\ell^{\delta}_{q_i}$, \eqref{BM2.52} follows from the expansion
\begin{align*}
    \mathbb E[(Z^{\delta}_{\ell+1,1}-Z^{\delta}_{\ell,1})(Z^{\delta}_{\ell+1,2}-Z^{\delta}_{\ell,2})]
    &=
    \xi'\lt(\E[N^{\delta}_{\ell,1}N^{\delta}_{\ell,2}]\rt)
    -
    \xi'\lt(\E[N^{\delta}_{\ell-1,1}N^{\delta}_{\ell,2}]\rt)
    \\
    &\quad
    -
    \xi'\lt(\E[N^{\delta}_{\ell-1,1}N^{\delta}_{\ell,2}]\rt)
    +
    \xi'\lt(\E[N^{\delta}_{\ell-1,1}N^{\delta}_{\ell-1,2}]\rt)
\end{align*}
and the fact that all $4$ terms on the right hand side are equal thanks to \eqref{BM4.2},~\eqref{BM5.52}.
\end{proof}

\subsection{Diffusive Scaling Limit}

We begin with the following slight generalization of Lemma~\ref{lem:SDE} which allows for the additional perturbation steps of branching IAMP but still considers only a single sample path.

\begin{lemma}\label{lem:SDE1.5}

Fix $\ubq\in (\lbq,1)$ and an index $a$. There exists a coupling between the families of triples $\{(Z^{\delta}_{\ell,a},X^{\delta}_{\ell,a},N^{\delta}_{\ell,a})\}_{\ell\geq 0}$ and $\{(Z_t,X_t,N_t)\}_{t\geq 0}$ such that the following holds for a constant $C>0$. 
{\color{black}
For large enough $\ul$, and every $\ell\geq \ul$ with $q_{\ell}\leq \ubq$,
}
\begin{align}
    \max _{\ul \leq j \leq \ell} \mathbb{E}\lt[\lt(X_{j,a}^{\delta}-X_{q_{j}}\rt)^{2}\rt] 
    &\leq C 
    \delta,
    \\
    \max _{\ul \leq j \leq \ell} \mathbb{E}\lt[\lt(N_{j,a}^{\delta}-N_{q_{j}}\rt)^{2}\rt] 
    &\leq C 
    \delta.
\end{align}
\end{lemma}

\begin{proof}
We prove the scaling limits for $X_{\ell}^{\delta}$ and $N_{\ell}^{\delta}$ separately, inducting over $\ell$ in each proof. We suppress the index $a$ as it is irrevelant.

\paragraph{Scaling limit for $X^{\delta}_{\ell}$}

We begin by checking the claim for $\ell=\ul$. Recalling that $\int_0^{q_{\ul+1}}\sqrt{\xi''(t)}\de B_{t}=Z_{\ul+1}$, we have
\begin{align*}
\E\lt[\lt(X^{\delta}_{\ul} - X_{q}\rt)^2\rt] &= \E\lt[\lt(Z^{\delta}_{\ul} -\int_0^{q_{\ul+1}}\sqrt{\xi''(t)}\de B_{t}\rt)^2\rt]  \\
&\le  2\E\lt[\lt(Z^{\delta}_{\ul} - Z^{\delta}_{\ul+1}\rt)^2\rt] + 2\E\lt[\lt(\int_{q}^{q_{\ul+1}}\sqrt{\xi''(t)}\de B_{t}\rt)^2\rt] \\
&= 4(\xi'(q_{\ul+1})-\xi'(q))\\
&\leq C\delta.
\end{align*} 

We continue using a standard self-bounding argument. Let $\ell  \ge \ul+1 $ such that $q_{\ell}\le \ubq$. Define $\Delta^{X}_{\ell} = X^{\delta}_{\ell} - X_{q_{\ell}}$. 
Then
\begin{align*}
    \Delta^X_{\ell}  - \Delta^X_{\ell-1} 
    &= 
    \int_{q^{\delta}_{\ell-1}}^{q^{\delta}_{\ell}} \big(v(q^{\delta}_{\ell-1}; X^{\delta}_{\ell}) - v(t; X_{t})\big) \de t + Z^{\delta}_{\ell}  - Z^{\delta}_{\ell-1} - \int_{q^{\delta}_{\ell-1}}^{q^{\delta}_{\ell}} \sqrt{\xi''(s)}\de B_s
    \\
    &=  
    \int_{q^{\delta}_{\ell-1}}^{q^{\delta}_{\ell}} \big(v(q^{\delta}_{\ell-1}; X^{\delta}_{\ell}) - v(t; X_{t})\big) \de t\
    \\
    &=   
    \int_{q^{\delta}_{\ell-1}}^{q^{\delta}_{\ell}} \big(v(q^{\delta}_{\ell-1}; X^{\delta}_{\ell}) - v(q^{\delta}_{\ell-1}; X_{t})\big) \de t
    + 
    \int_{q^{\delta}_{\ell-1}}^{q^{\delta}_{\ell}} \big(v(q^{\delta}_{\ell-1}; X_{t}) - v(t; X_{t})\big) \de t.
\end{align*}
The first term just above is at most $C \int_{q^{\delta}_{j-1}}^{q^{\delta}_{j}}|X_j^{\delta} - X_t|\de t$ since $v$ is Lipschitz in space uniformly for $t\in [0,1]$. For the second term we estimate 
\begin{align*} 
    \sum_{k=\ul+1}^\ell &\int_{q^{\delta}_{k-1}}^{q^{\delta}_{k}} \big|v(q^{\delta}_{k-1}; X_{t}) - v(t; X_{t})\big| \de t 
    \\
    &\le 
    \sum_{k=\ul+1}^\ell \int_{q^{\delta}_{k-1}}^{q^{\delta}_{k}} \Big\{ \big|v(q^{\delta}_{k-1}; X_{t}) - v(t; X_{t})\big| +\big|v(t; X_{t})-v(q^{\delta}_k; X_{t})\big|\Big\} \de t
    \\
    &\le 
    \delta \sum_{k=\ul+1}^\ell \sup_{q^{\delta}_{k-1} \le t \le q^{\delta}_{k}}  \Big\{\big|v(q^{\delta}_{k-1}; X_{t}) - v(t; X_{t})\big| + \big|v(t; X_{t})-v(q^{\delta}_k; X_{t})\big|\Big\}
    \\
    &\le 
    \delta \sup_{t_1,\cdots,t_k} \sum_{k=\ul+1}^\ell  \Big\{\big|v(q^{\delta}_{k-1}; X_{t_k}) - v(t_k; X_{t_k})\big| + \big|v(t_k; X_{t_k})-v(q^{\delta}_k; X_{t_k})\big|\Big\}
    \\
     &\le 
     C \delta,
\end{align*}
where the last inequality follows from {\color{black}the strong total variation of $v$}. Combining the bounds and summing over $j$, we find
\[
    \big|\Delta^X_{\ell}\big| \le |\Delta_{\ul}^X|+\sum_{j=\ul+1}^{\ell} \big|\Delta^X_{j}  - \Delta^X_{j-1}\big| 
    \le 
    C  \sum_{j=\ul+1}^\ell \int_{q^{\delta}_{j-1}}^{q^{\delta}_{j}}  |X_j^{\delta} - X_t| \de t + 2C \delta.
\]
Squaring and taking expectations,
\begin{align*}
    \E\big[(\Delta^X_{\ell})^2\big] 
    &\le 2C^2 
    \E\Big(\sum_{j=\ul+1}^\ell \int_{q^{\delta}_{j-1}}^{q^{\delta}_{j}}  |X_j^{\delta} - X_t| \de t\Big)^2 + 10C^2 \delta^2
    \\
    &\le 
    2C^2 (\ell-\ul) \delta \sum_{j=\ul+1}^\ell \int_{q^{\delta}_{j-1}}^{q^{\delta}_j}  \E|X_j^{\delta} - X_t|^2 \de t + 10C^2 \delta^2.
\end{align*}
Furthermore, $\E|X_j^{\delta} - X_t|^2 \le 2 \E|X_j^{\delta} - X_{q^{\delta}_{j}}|^2+2 \E|X_{q^{\delta}_{j}} -X_t|^2$. It is clear that $\E|X_t -X_s|^2 \le C|t-s|$ for all $t,s$, as $\xi''$ is bounded on $[0,1]$. Therefore 
\[\E\big[(\Delta^X_{\ell})^2\big] \le 4C^2 (\ell-\ul) \delta^2 \sum_{j=\ul+1}^\ell \E\big[(\Delta^X_{j})^2\big] + 4C^3 (\ell-\ul) \delta  \sum_{j=\ul+1}^\ell\int_{q^{\delta}_{j-1}}^{q^{\delta}_j} \delta \de t + 10C^2 \delta^2.\]
The middle term is proportional to $(\ell-\ul)^2 \delta^3$. Using $(\ell-\ul) \delta \le 1$ we obtain that for $\delta$ smaller than an absolute constant, it holds that 
\[\E\big[(\Delta^X_{\ell})^2\big] \le C \delta \sum_{j=\ul+1}^{\ell-1} \E\big[(\Delta^X_{j})^2\big] + C\delta,\]
for a different absolute constant $C$. This implies $\E\big[(\Delta^X_{\ell})^2\big] \le C \delta$ as desired. 
%%%

\paragraph{Scaling limit for $N^{\delta}_{\ell}$}

Again we begin by checking that $\ell=\ul$. We compute:
\begin{align*}
    \E\big[\big(N^{\delta}_{\underline{\ell}} - N_{q}\big)^2\big] 
    &= 
    \E\big[\big((1+\eps_0)\partial_x\Phi_{\gamma_*}(q, X_{\underline{\ell}}) - \partial_x\Phi_{\gamma_*}(q, X_{q})\big)^2\big]  
    \\
    &\le  
    2\eps_0^2\E\big[\big(\partial_x\Phi_{\gamma_*}(q, X_{\underline{\ell}})\big)^2\big] 
    + 2\E\big[\big(\partial_x\Phi_{\gamma_*}(q, X_{\underline{\ell}}) - \partial_x\Phi_{\gamma_*}(q, X_{q})\big)^2\big] 
    \\
    &= 
    C\eps_0^2 + C\E\big[\big(X^{\delta}_{\underline{\ell}} - X_{q}\big)^2\big]
    \\
    &\le 
    C\delta.
\end{align*}
Here we have used again the inequality $(x-z)^2\leq 2(x-y)^2+2(y-z)^2$ and the fact that derivatives of $\Phi_{\gamma_*}$ are bounded, as well as $\eps_0^2\leq \delta/q$. At the end we use the bound on $\E\lt[\lt(X^{\delta}_{\underline{\ell}} - X_{q}\rt)^2\rt]$ shown in the previous part of this proof. Next we turn to $\ell\geq \ul+1$. We have
\begin{align*}
    (N^{\delta}_{j+1}-N_{q^{\delta}_{j+1}})-(N^{\delta}_{j}-N_{q^{\delta}_{j}})
    &=
    u_j^{\delta}(X^{\delta}_{j})(Z^{\delta}_{j+1}-Z^{\delta}_{j}) - \int_{q^{\delta}_j}^{q^{\delta}_{j+1}} \sqrt{\xi''(t)} u(t,X_t)\de B_t
    \\
    &=
    \int_{q^{\delta}_j}^{q^{\delta}_{j+1}} \sqrt{\xi''(t)} \lt(u_j^{\delta}(X^{\delta}_{j})-u(t,X_t)\rt)\de B_t.
\end{align*}
and so
\begin{align}
\label{eq:bound_martingale_disc}
    \E\lt[\lt(N^{\delta}_\ell - N_{q^{\delta}_{\ell}}\rt)^2\rt] 
    &\leq 
    2\cdot\E\lt[\lt(N^{\delta}_{\ul} - N_{q}\rt)^2\rt]+2\cdot \E\lt[\lt(\sum_{j=\ul}^{\ell-1}  \int_{q^{\delta}_j}^{q^{\delta}_{j+1}} \sqrt{\xi''(t)} \lt(u_j^{\delta}(X^{\delta}_{j})-u(t,X_t)\rt)\de B_t\rt)^2\rt]
    \\
\nonumber
    &\leq 
    2C\delta+2 \sum_{j=\ul}^{\ell-1}\int_{q^{\delta}_j}^{q^{\delta}_{j+1}}  \E\lt[\lt(u_j^{\delta}(X^{\delta}_{j}) -  u(t,X_t)\rt)^2 \rt]\, \xi''(t) \de t.
\end{align}
Recall that $u_j^{\delta}(x) = u(q^{\delta}_{j}; x)/\Sigma^{\delta}_j$ for $j \ge 1$ where $\Sigma^{\delta}_j$ is given by
\[
    (\Sigma^{\delta}_j)^2 = \frac{\xi'(q^{\delta}_{j+1}) - \xi'(q^{\delta}_{j})}{\delta} \E[u(q^{\delta}_{j}; X^{\delta}_{j})^2].
\]
We first show the bound
\begin{equation}
\label{eq:sigma_bound}
    \big|(\Sigma^{\delta}_j)^2 - 1\big| \le C\sqrt{\delta}
\end{equation}
for $\delta$ small enough, which is of independent interest. 
Since $u$ is bounded and $\xi'''$ is bounded on $[0,1]$, we have
\[
    \big|(\Sigma^{\delta}_j)^2 - \xi''(q^{\delta}_{j}) \E[u(q^{\delta}_{j}; X^{\delta}_{j})^2]\big| \le C\delta.
\] 
Observe now that
\[
    \mathbb E\lt[|X^{\delta}_{j}-X_{q^{\delta}_{j}}|\rt]
    \leq 
    \sqrt{\delta}+\frac{\mathbb E\lt[|X^{\delta}_{j}-X_{q^{\delta}_{j}}|^2\rt]}{\sqrt{\delta}}\leq C\sqrt{\delta}.
\]
Since $u$ is Lipschitz in space and bounded, this implies
\[
    \big|(\Sigma^{\delta}_j)^2 - \xi''(q^{\delta}_{j}) \E[u(q^{\delta}_{j}; X_{q^{\delta}_{j}})^2]\big| 
    \le
    C\sqrt{\delta}.
\] 
Since $\E[N_t^2]=t$ for all $t \in [0,1]$ and $t \mapsto u(t,X_t)$ is a.s.\ continuous, Lebesgue's differentiation theorem implies that for all $t\in[0,1]$,
\[
    \xi''(t) \E[u(t; X_{t})^2]  = 1,
\]
and hence $\big|(\Sigma^{\delta}_j)^2 - 1| \le C\sqrt{\delta}$ for $\delta$ smaller than some absolute constant.
This implies the bound $|u^{\delta}_j(X^{\delta}_j) - u(q^{\delta}_{j}; X^{\delta}_j)| \le C \big|\frac{1}{\Sigma^{\delta}_j} - 1\big| \le C \sqrt{\delta}$.
Now, going back to Eq.~\eqref{eq:bound_martingale_disc}, we have
\begin{align*}
    \E\lt[\lt(N^{\delta}_\ell - N_{q^{\delta}_{\ell}}\rt)^2\rt]  
    &\le 
    2\sum_{j=\ul}^{\ell-1}\int_{q^{\delta}_j}^{q^{\delta}_{j+1}}  \E\lt[\lt(u_j^{\delta}(X^{\delta}_{j}) -  u(q^{\delta}_{j}; X^{\delta}_j)\rt)^2\rt]\, \xi''(t) \de t
    \\
    &~~~+
    2\sum_{j=\ul}^{\ell-1}\int_{q^{\delta}_j}^{q^{\delta}_{j+1}}  \E\lt[\lt(u(q^{\delta}_{j}; X^{\delta}_j) -  u(t,X_t)\rt)^2\rt] \xi''(t) \de t
\end{align*}
From what we just established the first term is at most $C (\ell-\ul)\delta^{2} \le C \delta$. To estimate the second term we compute:
\begin{align*}
    \sum_{j=\ul}^{\ell-1}\int_{q^{\delta}_j}^{q^{\delta}_{j+1}} & \E\lt[\lt(u(q^{\delta}_{j}; X^{\delta}_j) -  u(t,X_t)\rt)^2\rt] \xi''(t) \de t 
    \\ 
    &~~~\le 
    C\sum_{j=\ul}^{\ell-1}\int_{q^{\delta}_j}^{q^{\delta}_{j+1}}\E\lt[\lt(u(q^{\delta}_{j}; X^{\delta}_j) -  u(q^{\delta}_{j},X_{q^{\delta}_{j}})\rt)^2\rt]  \de t
    \\
    &~~+
    C\sum_{j=\ul}^{\ell-1}\int_{q^{\delta}_j}^{q^{\delta}_{j+1}}\E\lt[\lt(u(q^{\delta}_{j}; X_{q^{\delta}_{j}}) -  u(q^{\delta}_{j},X_t)\rt)^2\rt] \de t
    \\
    &~~+
    C\sum_{j=\ul}^{\ell-1}\int_{q^{\delta}_j}^{q^{\delta}_{j+1}}\E\lt[\lt(u(q^{\delta}_{j}; X_t) -  u(t,X_t)\rt)^2\rt]   \de t
    \\
    &= I +II + III. 
\end{align*}
Since $u$ is Lipschitz in space, we obtain $I \le C(\ell-\ul)\delta^2$. From $\E[|X_t-X_s|^2] \le C|t-s|$, we obtain $II \le C(\ell-\ul)\delta^2$. Finally, since $u$ is Lipschitz in time uniformly in space and $(\ell-\ul)\delta \le 1$, it follows that $III \le C\delta$. Altogether we obtain
\[
    \E\lt[\lt(N^{\delta}_\ell - N_{q^{\delta}_{\ell}}\rt)^2\rt] \le C\delta
\]
concluding the proof.
\end{proof}

We now extend Lemmas~\ref{lem:SDE}, \ref{lem:SDE1.5} to describe the joint scaling limit of multiple branches, which become independent at the branching time. Let $(B_t^{a})_{t\in [0,1],a\in \{1,2\}}$ be standard Brownian motions with $B_t^1=B_t^2$ for $t\leq q_B$ and with independent increments after time $q_B$. Couple $B_t^a$ with $(Z_{\ell,a}^{\delta})_{\ell\geq 0}$ via 
\[
    Z_{j,a}^{\delta}=\int_0^{q_j^{\delta}} \sqrt{\xi''(t)}\de B_{t}^a
\] 
and natural filtration $(\mathcal F_t)_{t\in [0,1]}$ with $\mathcal F_t=\sigma\lt((B_s^1,B_s^2)_{s\leq t}\rt)$. We consider for $a\in\{1,2\}$ the SDE
\[
    dX_t^a=\gamma_*(t)\partial_x\Phi_{\gamma_*}(t,X_t^a)\de t+\sqrt{\xi''(t)}\de B_t^a
\]
with initial condition $X_0^a=0$, and define
\begin{align*}
    N_t^a
    &\equiv
    \partial_x\Phi_{\gamma_*}(q,X_{\lbq}^a)+\int_{\lbq}^t \sqrt{\xi''(s)}u(s,X_s^a)\de B_s^a=\partial_x\Phi_{\gamma_*}(t,X_t^a),
    \\ 
    Z_t^a
    &\equiv
    \int_0^t \sqrt{\xi''(s)}\de B_s^a.
\end{align*}

\begin{restatable}{lemma}{SDE2}
\label{lem:SDE2}
Fix $\ubq\in (\lbq,1)$. There exists a coupling between the families of triples $\{(Z^{\delta}_{\ell,a},X^{\delta}_{\ell,a},N^{\delta}_{\ell,a})\}_{\ell\geq 0,a\in\{1,2\}}$ and $\{(Z_t^a,X_t^a,N_t^a)\}_{t\geq 0,a\in\{1,2\}}$ such that the following holds. For some $\delta_0>0$ and constant $C>0$, for every $\delta\leq\delta_0$ and $\ell\geq \ul$ with $q_{\ell}\leq \ubq$ we have

\begin{align*}
    \max _{\ul \leq j \leq \ell} \mathbb{E}\lt[\lt(X_{j,a}^{\delta}-X^a_{q_{j}}\rt)^{2}\rt] 
    &\leq 
    C \delta,
    \\
    \max _{\ul \leq j \leq \ell} \mathbb{E}\lt[\lt(N_{j,a}^{\delta}-N^a_{q_{j}}\rt)^{2}\rt] 
    &\leq 
    C \delta.
\end{align*}

\end{restatable}

\begin{proof}
We generate the desired ``grand coupling" by starting with $(B_t^1,B_t^2)$ as above, generating $(Z_t^1,Z_t^2)$, and then setting $Z_{j,a}^{\delta}=Z_{q_j^{\delta}}^a$ for each $a\in\{1,2\}$ and $j\leq \ubl$ as in the coupling of Lemma~\ref{lem:SDE1.5}. It follows from Lemma~\ref{lem:BMlimit2} that this results in the correct law for $(Z_{j,a}^{\delta})_{j\in\mathbb N,a\in \{1,2\}}.$ Now, all $3$ continuous-time functions in the coupling of Lemma~\ref{lem:SDE1.5} are determined almost surely by $Z_t$. Furthermore all $3$ discrete-time functions are determined almost surely by the sequence $Z_j^{\delta}$. Therefore the coupling just constructed between $\{Z^{\delta}_{\ell,a}\}_{\ell\geq 0,a\in\{1,2\}}$ and $\{Z_t^a\}_{t\geq 0,a\in\{1,2\}}$ automatically extends to a coupling of $\{(Z^{\delta}_{\ell,a},X^{\delta}_{\ell,a},N^{\delta}_{\ell,a})\}_{\ell\geq 0,a\in\{1,2\}}$ and $\{(Z_t^a,X_t^a,N_t^a)\}_{t\geq 0,a\in\{1,2\}}$. Since the two $a$-marginals of the coupling just constructed both agree with that of Lemma~\ref{lem:SDE1.5}, the claimed approximation estimates carry over as well, concluding the proof.
\end{proof}

\subsection{The Energy Gain of Incremental AMP}

% \IAMPenergy*

Here we prove Lemma~\ref{lem:iampenergy}, stated for the branching case.

\begin{lemma}
\label{lem:iampenergy2}
\begin{equation}
\label{eq:iampenergy2}
    \lim_{\ubq\to 1}
    \lim_{\ul\to\infty}
    \plim_{N \rightarrow \infty}
    \frac{H_{N}\lt(\bn^{\ubl,a}\rt)-H_{N}\lt(\bn^{\ul,a}\rt)}{N} 
    = 
    \int_{\lbq}^{1} 
    \xi^{\prime \prime}(t) 
    \mathbb{E}
    \lt[u\lt(t, X_{t}\rt)\rt] 
    \mathrm{d} t.
\end{equation}
\end{lemma}

\begin{proof}
We give the main part of the proof for the ordinary (non-branching) version of the algorithm and explain at the end why the same arguments apply in the branching case. Recall that $\delta=\delta(\ul)\to 0$ as $\ul\to\infty$, which we will implicitly use throughout the proof. Observe also that $\langle h,\bn^{\ubl}-\bn^{\ul}\rangle_N\simeq 0$ because the values $(N_{\ell}^{\delta})_{\ell\geq\ul}$ form a martingale sequence. Therefore it suffices to compute the in-probability limit of $\frac{\wt{H}_{N}\lt(\bn^{\ubl}\rt)-\wt{H}_{N}\lt(\bn^{\ul}\rt)}{N}$. The key is to write 
\[
    \frac{\wt{H}_{N}\lt(\bn^{\ubl}\rt)-\wt{H}_{N}\lt(\bn^{\ul}\rt)}{N}=\sum_{\ell=\ul}^{\ubl-1}\frac{\wt{H}_{N}\lt(\bn^{\ell+1}\rt)-\wt{H}_{N}\lt(\bn^{\ell}\rt)}{N}
\]
and use a Taylor series approximation of the summand. In particular for $F\in C^3(\mathbb R)$, applying Taylor's approximation theorem twice yields
\begin{align*}F(1)-F(0)&=aF'(0)+\frac{1}{2}F''(0)+O(\sup_{a\in [0,1]}|F'''(a)|)\\
&= F'(0)+\frac{1}{2}(F'(a)-F'(0))+O(\sup_{a\in [0,1]}|F'''(a)|)\\
&=\frac{1}{2}(F'(1)+F'(0))+O(\sup_{a\in [0,1]}|F'''(a)|) .\end{align*}

Assuming {\color{black}$\sup_{\ell}\frac{\|\bn^{\ell}\|}{\sqrt{N}}\leq 1+\eta$, which holds with high probability for any $\eta>0$ if $\ul$ is large enough}, we apply this estimate with $F(a)=\wt{H}_N\lt((1-a)\bn^{\ell}+a\bn^{\ell+1}\rt)$. {\color{black}Recalling \eqref{eq:operator-norm-def},} the result is:
\begin{align*}
    \lt|\wt{H}_{N}\lt(\bn^{\ell+1}\rt)-\wt{H}_{N}\lt(\bn^{\ell}\rt) -\frac{1}{2}\lt\la \nabla \wt{H}_N(\bn^{\ell})+\nabla \wt{H}_N(\bn^{\ell+1}),\bn^{\ell+1}-\bn^{\ell}\rt\ra  \rt|
    &
    \\
    \leq  
    O\lt(\sup_{\|\bv\|\leq (1+\eta)\sqrt{N}}\left\|\nabla^3 \wt{H}_N(\bv)\right\|_{\op} \rt)\|\bn^{\ell+1}-\bn^{\ell}\|^3
    &
    \,.
\end{align*}
Proposition~\ref{prop:lip} implies that 
\[
  \sup_{|\bv|\leq (1+\eta)\sqrt{N}}\left\|\nabla^3 \wt{H}_N(\bv)\right\|_{\op} \leq O(N^{-1/2}
\]
with high probability. On the other hand $\plim_{N\to\infty}\|\bn^{\ell+1}-\bn^{\ell}\|=\sqrt{\delta N}$ for each $\ul\leq \ell\leq \ubl-1$. Summing and recalling that $\ubl-\ul\leq \delta^{-1}$ yields the high-probability estimate
\begin{align*}
  \sum_{\ell=\ul}^{\ubl-1}
  &
  \lt|\wt{H}_{N}\lt(\bn^{\ell+1}\rt)-\wt{H}_{N}\lt(\bn^{\ell}\rt) -\frac{1}{2}\lt\la  \nabla \wt{H}_N(\bn^{\ell})+\nabla \wt{H}_N(\bn^{\ell+1}),\bn^{\ell+1}-\bn^{\ell}\rt\ra  \rt|
  \\ 
  &\leq 
  \sum_{\ell=\ul}^{\ubl-1} O\lt(\sup_{\|\bx\|\leq (1+\eta)\sqrt{N}}\left\|\nabla^3 \wt{H}_N(\bx)\right\|_{\op} \rt)\cdot \sup_{\ell}\|\bn^{\ell+1}-\bn^{\ell}\|^3
  \\
  &\leq O(N\sqrt{\delta}).
\end{align*}
Because $\ul\to\infty$ implies $\delta\to 0$ this term vanishes in the limit, and it remains to show
\[
  \lim_{\ubq\to 1}\lim_{\ul\to\infty}\plim_{N \rightarrow \infty} \sum_{\ell=\ul}^{\ubl-1}\frac{1}{2}\lt\la  \nabla \wt{H}_N(\bn^{\ell})+\nabla \wt{H}_N(\bn^{\ell+1}),\bn^{\ell+1}-\bn^{\ell}\rt\ra_N  
  =
  \int_{q}^{1} \xi^{\prime \prime}(t) \mathbb{E}\lt[u\lt(t, X_{t}\rt)\rt] \mathrm{d}t.
\]
Next, observe by \eqref{eq:general_amp} that:
\begin{equation}
\label{eq:amprearrange}
\nabla \wt{H}_N(\bn^{\ell})=\bz^{\ell+1}+ \sum_{j=0}^\ell d_{\ell, j} \bn^{j-1}.
\end{equation}
Passing to the limiting Gaussian process $(Z^{\delta}_k)_{k\in\mathbb Z^+}$ via state evolution, and ignoring for now the constant number of branching updates,
\begin{align*}
  \plim_{N\to\infty}
  \lt\la 
  \nabla \wt{H}_N(\bn^{\ell}),\bn^{\ell+1}-\bn^{\ell}
  \rt\ra_N
  &=
  \mathbb E\lt[ 
  Z^{\delta}_{\ell+1}(N^{\delta}_{\ell+1}-N^{\delta}_{\ell})
  \rt]
  +
  \sum_{j=0}^{\ell} d_{\ell,j}
  \mathbb E\lt[
  N^{\delta}_{j-1}(N^{\delta}_{\ell+1}-N^{\delta}_{\ell})
  \rt],
  \\
  \plim_{N\to\infty}\lt\la \nabla \wt{H}_N(\bn^{\ell+1}),\bn^{\ell+1}-\bn^{\ell}\rt\ra_N 
  &=
  \mathbb E\lt[ Z^{\delta}_{\ell+2}(N^{\delta}_{\ell+1}-N^{\delta}_{\ell})\rt]
  + \sum_{j=0}^{\ell+1} d_{\ell+1,j}\mathbb E\lt[N^{\delta}_{j-1}(N^{\delta}_{\ell+1}-N^{\delta}_{\ell})\rt].
\end{align*}

As $(N^{\delta}_k)_{k\geq \mathbb Z^+}$ is a martingale process, it follows that the right-hand expectations all vanish. Similarly it holds that
\begin{align*}
  \mathbb E[Z_{\ell+2}^{\delta}(N_{\ell+1}^{\delta}-N_{\ell}^{\delta})]&=\mathbb E[Z_{\ell+1}^{\delta}(N_{\ell+1}^{\delta}-N_{\ell}^{\delta})]
  \\
  \mathbb E[Z_{\ell}^{\delta}(N_{\ell+1}^{\delta}-N_{\ell}^{\delta})]&=0.
\end{align*}
Rewriting and using Lemma~\ref{lem:BMlimit2} in the last step,
\begin{align*}
  \plim_{N\to\infty}\frac{1}{2}\lt\la \nabla \wt{H}_N(\bn^{\ell})+\nabla \wt{H}_N(\bn^{\ell+1}),\bn^{\ell+1}-\bn^{\ell}\rt\ra_N &=\mathbb E[(Z_{\ell+1}^{\delta}-Z^{\delta}_{\ell})(N_{\ell+1}^{\delta}-N_{\ell}^{\delta})]
  \\
  &=
  \mathbb E[u_{\ell}^{\delta}(X^{\delta}_{\ell})(Z_{\ell+1}^{\delta}-Z^{\delta}_{\ell})^2]\\
  &=
  \mathbb E\lt[\mathbb E[u_{\ell}^{\delta}(X^{\delta}_{\ell})(Z_{\ell+1}^{\delta}-Z^{\delta}_{\ell})^2|\mathcal F^{\delta}_{\ell}]\rt]\\
  &=
  (\xi'(q^{\delta}_{\ell+1})-\xi'(q^{\delta}_{\ell}))\cdot \mathbb E[u_{\ell}^{\delta}(X^{\delta}_{\ell})]\\
  &=
  (\xi'(q^{\delta}_{\ell+1})-\xi'(q^{\delta}_{\ell}))\cdot \frac{\mathbb E[u_{q^{\delta}_{\ell}}(X^{\delta}_{\ell})]}{\Sigma_{\ell}^{\delta}}.
\end{align*}

Recalling \eqref{eq:sigma_bound}, the fact that $u_t(x)$ is uniformly Lipschitz in $x$ for $t\in [0,\ubq]$, the fact that $\xi'(q^{\delta}_{\ell+1})-\xi'(q^{\delta}_{\ell})=\delta\xi''(q^{\delta}_{\ell})+O(\delta^2)$ and the coupling of Lemma~\ref{lem:SDE1.5}, it follows that 
\[
  \plim_{N\to\infty}\frac{1}{2}\lt\la \nabla \wt{H}_N(\bn^{\ell})+\nabla \wt{H}_N(\bn^{\ell+1}),\bn^{\ell+1}-\bn^{\ell}\rt\ra_N=\delta \xi''(q_{\ell}^{\delta})\mathbb E[u_{q_{\ell}^{\delta}}(X_{q_{\ell}^{\delta}})]+O_{\ubq}(\delta^{3/2}).
\]
Summing over $\ell$ and using continuity of $u_t(x)$ in $t$, it follows that 
\[
  \lim_{\ul\to\infty}\plim_{N \rightarrow \infty} \frac{\wt{H}_N(\bn^{\ubl})-\wt{H}_N(\bn^{\ul})}{N}=\int_{\lbq}^{\ubq}\xi''(t)\mathbb E[u(t,X_t)]\de t.
\]
Sending $\ubq\to 1$ now concludes the proof when there are no branching steps. Extending the proof to cover branching steps is not difficult and we explain it now. Everything up to \eqref{eq:amprearrange} is still valid, and if the number $|Q|$ of branching steps is $m$, then the full analysis applies to all but $m$ terms. However the simple uniform bound
\begin{align*}
  \wt{H}_N(\bn^{\ell+1})-\wt{H}_N(\bn^{\ell})&\leq \|\bn^{\ell+1}-\bn^{\ell}\|\cdot \sup_{\|\bx\|_N\leq 1+\frac{\eta}{2}}\|\nabla \wt{H}_N(\bx)\|\\
  &\leq 
  O(\sqrt{N\delta})\cdot O(\sqrt{N})
  \\
  &\leq 
  O(N\sqrt{\delta}) 
\end{align*}
holds with high probability. Here we have used Proposition~\ref{prop:1lip} to deduce $\|\bn^{\ell}\|_N,\|\bn^{\ell+1}\|_N\leq 1+o(1)$ with high probability, and also Proposition~\ref{prop:lip} and Equation~\eqref{BM5.2}. Therefore all telescoping terms, branching or not, uniformly contribute $O(N\sqrt{\delta})$ energy in probability. As a result, even when a constant number of non-branching terms are replaced by branching terms, the same analysis applies up to error $O(N\sqrt{\delta})$, yielding the same asymptotic energy for branching $\lbq$-IAMP and completing the proof.
\end{proof}

\footnotesize
\bibliographystyle{alpha}
\bibliography{all-bib}

\end{document}